\documentclass[11pt,a4paper]{article}%
\usepackage{amsmath}
\usepackage{amssymb}
\usepackage{graphicx}
\usepackage{amscd}
\usepackage{amsfonts}
\usepackage{amstext}%
\providecommand{\U}[1]{\protect\rule{.1in}{.1in}}
\newtheorem{theorem}{Theorem}

\newtheorem{algorithm}[theorem]{Algorithm}
\newtheorem{assumption}[theorem]{Assumption}

\newtheorem{corollary}[theorem]{Corollary}

\newtheorem{definition}[theorem]{Definition}

\newtheorem{lemma}[theorem]{Lemma}

\newtheorem{proposition}[theorem]{Proposition}
\newtheorem{remark}[theorem]{Remark}

\newenvironment{proof}[1][Proof]{\noindent\textbf{#1.} }{\ \rule{0.5em}{0.5em}}
\begin{document}

\title{The consensus problem for opinion dynamics with local average random interactions}
\author{M. Gianfelice, G. Scola\\Dipartimento di Matematica e Informatica\\Universit\`{a} della Calabria\\Campus di Arcavacata\\Ponte P. Bucci - cubo 30B\\I-87036 Arcavacata di Rende\\gianfelice@mat.unical.it\\giuseppe.scola@unical.it}
\maketitle

\begin{abstract}
We study the consensus formation for an agents based model, generalizing that
originally proposed by Krause \cite{Kr}, by allowing the communication
channels between any couple of agents to be switched on or off randomly, at
each time step, with a probability law depending on the proximity of the
agents' opinions. Namely, we consider a system of agents sharing their
opinions according to the following updating protocol. At time $t+1$ the
opinion $X_{i}\left(  t+1\right)  \in\left[  0,1\right]  $ of any agent $i$ is
updated at the weighted average of the opinions of the agents communicating
with it at time $t.$ The weights model the confidence level an agent assigns
to the opinions of the other agents and are kept fixed by the system dynamics,
but the set of agents communicating with any agent $i$ at time $t+1$ is
randomly updated in such a way that the agent $j$ can be chosen to belong to
this set independently of the other agents with a probability that is a non
increasing function of $\left\vert X_{i}\left(  t\right)  -X_{j}\left(
t\right)  \right\vert .$ This condition models the fact that a communication
among the agents is more likely to happen if their opinions are close. We
prove that if the agent's communication graph at time one, conditionally on 
the initial believes' configuration, is sufficiently connected,
the system reaches consensus at geometric rate, i.e., more precisely, as the time tends to infinity
the agents' opinions will reach the same value geometrically fast. We also
discuss the consensus formation for a system of infinitely many agents. In
particular we analyze the evolution of the empirical average of the
agents' opinions in the limit as the size of the system tends to infinity
and characterize its fixed points in terms of agents' consensus proving that
this is reached geometrically fast with the same rate computed for the finite
system. 

\end{abstract}
\tableofcontents

\bigskip

\bigskip

\begin{description}
\item[AMS\ subject classification:] {\small 91D30,60J20. }

\item[Keywords and phrases:] {\small Opinion dynamics, agent-based models,
multi-agents systems, consensus, emergent behaviour, synchronization, complex
networks, non-linear Markov processes, mean-field limit of randomly perturbed
coupled dynamical systems.}

\item[Acknowledgement:] {\small M. Gianfelice is partially supported by
G.N.A.M.P.A..\newline G. Scola acknowledges financial support from MIUR, PRIN
2017 project MaQuMA, PRIN201719VMAST01.}
\end{description}

\bigskip

\section{Introduction, notations and results}

Opinion dynamics is a topic in applied mathematics which has witnessed a
growing interest in the last decades. This is due to the possibility to
describe the emergence of collective phenomena such as the reaching of
consensus in a community of peers by means of simple models of interacting
agents \cite{CFL}. The literature on the subject concentrates mainly on two
families of models: those cast in the framework of interacting particle
systems (IPS) e.g. the voter model and the majority-vote process \cite{Li},
the Axelrod model and its generalizations \cite{La}, the Deffuant-Weisbuch
model \cite{DNAW}, \cite{Ha}, \cite{HH}, and those belonging to the family of
coupled dynamical systems (CDS) e.g. \cite{Kr} (see \cite{FF} for linear models).

The common feature of these models is that the interactions among the agents
are designed in such a way that an agent adjust its opinion to that of its
neighbours and that agents are more likely to interact with those sharing
similar opinions. On the other hand, the main difference between IPS type
models and CDS type models, aside form the fact that IPS type models typically
evolve in continuous time while CDS type models evolve in discrete time, is
that in CDS type models all the agents can change the value of their opinion
at each time step while, in IPS type models, at a given time (usually at the
tick of a Poisson clock), only the elements of a randomly chosen subset of
agents are allowed to change the value of their opinion, and the rule under
which the opinion of the selected agents are updated, being stochastic in
general (see \cite{CF} for a stochastic version of the Deffuant model), can be
deterministic too \cite{DNAW}, \cite{HH}.

A renewed interest in the consensus problem for systems of interacting agents
has recently arisen from the architecture of Transformers, a particular class
of Deep Neural Networks (DNN) used in the construction of so-called Large
Language Models. Unlike traditional DNN models that operate on a single input
at a time, represented by a vector in $\mathbb{R}^{d},$ Transformers process a
very large number of inputs simultaneously. In this setup, the value of each
individual data point, called \emph{token}, is updated within each layer of
the network as a function of the values of the others. The equations governing
the flow of data evolving within a Transformer are precisely those that
describe a consensus model with local average interactions such as e.g. that
of Krause \cite{Kr}.\ Therefore, in the limit as the number of token becoming
very large, the characteristic features of the evolution of the data within
the Transformer is best captured by the evolution of the empirical
distribution of their values rather than by the evolution of the single data
just like it is prescribed by kinetic theory. We refer the reader to
\cite{WAWJJ}, \cite{GLPR} and \cite{CACP} for a more detailed account on this topic.

In this paper we present a model of consensus formation which, although it
represents a modification of that originally proposed by Krause, it is defined
by an updating rule of the agents opinions that recalls those characterizing
IPS type models. In particular, at each time step, firstly the set of
neighboring peers of any agent $i$ is selected at random in such a way that
the events that any two distinct agents belong to the neighbourhood of $i$ are
independent. Then, each agent update the value of its opinion to the average
value of the opinions of its neighbours.

More precisely, we consider a collection of agents which form the set of
vertices of a directed graph $G:=\left(  V,E\right)  ,E\subseteq V\times V.$
We assume that agent $u$ communicate with agent $v$ if the directed edge
$\left(  u,v\right)  $ is in $E.$ Each agent hold an opinion (\emph{belief})
represented by a variable taking values in $\left[  0,1\right]  .$ Agent's
beliefs evolve in time in such a way that the opinion $X_{u}\left(
t+1\right)  $ of the agent labelled by the graph vertex $u$ at time $t+1$ is
updated at the weighted average of the beliefs $X_{v}\left(  t\right)  $ of
the agents communicating with $u$ at time $t.$ The weights appearing in the
just mentioned average represent the quality of the information exchange among
the agents and do not change in time. On the other hand, the set of the agents
communicating with agent $u$ at time $t$ is randomly chosen according to a
probability distribution which gives more chance to a communication exchange
between agent $u$ and agent $v$ to happen if the values of their opinions at time $t-1$
were close. It is also natural to assume that the information exchange an agent has
with itself is always maximal.

In the rest of the section we introduce the notation used throughout the
paper, formally display the definition of the model and present the results
obtained in the following sections about the emergence of consensus for the
finite size system as well as the extension of these results when the size of
the system is very large.

In this last case, we show that if the size of the neighborhood of the agents
is finite it is possible to define directly the evolution of the system on
very large, possibly infinite, graphs as in \cite{La} and \cite{HH}.

Furthermore we focus on the evolution of the empirical distribution of the
agents' opinions, in the spirit of the kinetic limit for models of flocking
(see e.g. \cite{GO}, \cite{CCH}), as it has already been discussed for the
stochastic version of Deffuant model in \cite{CF} and more recently for other
IPS models \cite{AM}. We prove that, in the limit of the size of the system
that tends to infinity, despite the randomness of the agents evolution, that
of the empirical average of their opinions converges to an evolution defined
by a self-consistent transfer operator acting on the space of probability
measures on $\left(  \left[  0,1\right]  ^{2},\mathcal{B}\left(  \left[
0,1\right]  ^{2}\right)  \right)  $ endowed with the weak topology.

We stress that, as in \cite{Kr}, in this work, the dynamics of the state of
the edges of $G,$ representing the communication channels amid the agents, is
synchronous, i.e. they are all updated at each time step. On the other hand,
one may wish to modify the system's dynamics by letting the state of the edges
of $G$ to evolve under an asynchronous dynamics. In the case of a finite size
system, this can be realized, for example, by updating at each time step the
state of just one edge sampled uniformly at random among the elements of $E$
and leaving unchanged the states of the other edges. In fact, it seems that
the techniques used to analyse the emergence of consensus in the synchronous
case do not apply to this particular case. Therefore, the discussion about the
possibility to reach consensus for the opinion exchange model characterized by
the same updating rule for the values of the agents' opinions presented in
this paper, but subject to the asynchronous evolution of the communication
exchange among the agents just described, are deferred to a forthcoming paper.

\subsection{Notations}

If $A$ is a set and $B\subseteq A,\mathbf{1}_{B}$ denotes the indicator
function of $A$ and $B^{c}:=A\backslash B.$ Let $\mathcal{P}\left(  A\right)
$ the set of the subsets of $A.$ For any $k\geq1,$ we set $\mathcal{P}%
_{k}\left(  A\right)  :=\left\{  B\in\mathcal{P}\left(  A\right)  :\left\vert
B\right\vert =k\right\}  $ and denote by $\mathcal{P}_{0}\left(  A\right)
:=\bigvee\limits_{k\geq1}\mathcal{P}_{k}\left(  A\right)  $ the set of finite
subsets of $A.$

If $A$ is a metric space, $\mathcal{B}\left(  A\right)  $ denotes its Borel
$\sigma$algebra. We denote by $BM\left(  A\right)  $ the Banach space of
bounded real-valued measurable functions on $A,$ by $C\left(  A\right)  $ the
Banach space of real-valued continuous functions on $A,$ so that, if
$\varphi\in BM\left(  A\right)  ,\mathsf{supp}\varphi$ denotes the support of
$\varphi$ and $\left\Vert \varphi\right\Vert $ the sup-norm. Moreover, if
$Lip\left(  A\right)  $ is the Banach space of real-valued bounded Lipschitz
functions on $A,$ for any $\varphi\in Lip\left(  A\right)  ,$ we denote its
norm by $\left\Vert \varphi\right\Vert _{Lip}.$

For $A,A^{\prime}$ metric spaces, $BL\left(  A,A^{\prime}\right)  $ denotes
the space of bounded linear operators on $A$ with values in $A^{\prime}.$ If
$A^{\prime}=A$ we set $BL\left(  A,A\right)  :=BL\left(  A\right)  .$ In
particular, if $\mathbb{V}$ is a finite set, we denote by $St\left(
\mathbb{V}\right)  $ the convex subset of $BL\left(  \mathbb{R}^{\mathbb{V}%
}\right)  $ of stochastic matrices.

We denote by $\mathbb{E}$ the expected value of a random element when there is
no need to specify the probability space on which it is defined and
consequently write $\mathbb{P}\left\{  B\right\}  $ for the expected value of
the indicator function $\mathbf{1}_{B}$ of an event $B\subseteq A.$ The same
notation will be also kept when considering conditional expectations and
conditional probabilities. Besides, given a $\sigma$algebra $\mathcal{A}$ of
subsets of $A,$ we denote by $\mathfrak{P}\left(  A,\mathcal{A}\right)  $ the
set of probability measures on $\left(  A,\mathcal{A}\right)  .$ If $\mu
\in\mathfrak{P}\left(  A,\mathcal{A}\right)  ,\mathsf{supp}\mu$ denotes the
support of $\mu,$ and, for $\mu,\nu\in\mathfrak{P}\left(  A,\mathcal{A}%
\right)  ,\left\Vert \mu-\nu\right\Vert $ denotes the total variation distance
between the two measures.

Let $\mathbb{A}:=A^{\mathbb{V}}$ and denote by $\mathbf{a}:=\left\{
a_{v}\right\}  _{v\in\mathbb{V}}.$ If $\mathbb{A}$ is a poset w.r.t. the
partial order: $\mathbf{a}\leq\mathbf{a}^{\prime}$ if $a_{v}\leq a_{v}%
^{\prime},$ for any $v\in\mathbb{V},$ we say that a real-valued function
$\varphi$ on $\mathbb{A}$ is \emph{non-decreasing} if $\varphi\left(
\mathbf{a}\right)  \leq\varphi\left(  \mathbf{a}^{\prime}\right)  $ whenever
$\mathbf{a}\leq\mathbf{a}^{\prime}.$ Given two probability measures
$\mathbb{P},\mathbb{P}^{\prime}$ on $\left(  \mathbb{A},\mathcal{A}%
^{\otimes\mathbb{V}}\right)  $ we say that $\mathbb{P}$\ is stochastically
dominated by $\mathbb{P}^{\prime},$ and denote this property by $\mathbb{P}%
\overset{st}{\leq}\mathbb{P}^{\prime},$ if for any bounded non-decreasing
function $\varphi,\mathbb{E}\left[  \varphi\right]  =\int d\mathbb{P}\left(
a\right)  \varphi\left(  a\right)  \leq\int d\mathbb{P}^{\prime}\left(
a\right)  \varphi\left(  a\right)  =\mathbb{E}^{\prime}\left[  \varphi\right]
.$ Moreover, if $\mathbb{A}$ is finite, a probability measure $\mathbb{P}$ on
$\left(  \mathbb{A},\mathcal{P}\left(  \mathbb{A}\right)  \right)  $ is called
\emph{irreducible} if starting from any element of $\mathbb{A}$ with positive
$\mathbb{P}$-probability one can reach any other element with positive
$\mathbb{P}$-probability via successive coordinate changes without passing
through elements with zero $\mathbb{P}$-probability \cite{GHM}, \cite{Gr}.

If $\mathbb{A}:=A^{\mathbb{N}},$ for any $N\in\mathbb{N},$ we set
$\mathbf{a}_{N}:=\left(  a_{i},..,a_{N}\right)  $ and denote by $\mathcal{C}%
\left(  \mathbb{A}\right)  $ the \emph{cylinder }$\sigma$\emph{algebra} that
is the $\sigma$algebra generated by the cylinder subsets
\begin{equation}
\mathbf{C}_{N}\left(  B\right)  :=\left\{  \mathbf{a}\in\mathbb{A}%
:\mathbf{a}_{N}\in B\right\}  \ ,
\end{equation}
with $B\subseteq A^{N}$ if $A$ is a discrete set, while $B\in\mathcal{B}%
\left(  A^{N}\right)  =\mathcal{B}\left(  A\right)  ^{\otimes N}$ if $A$ is a
metric space.

If $\mathfrak{L}\left(  \mathbb{A}\right)  $ denotes the algebra of
real-valued bounded local (cylinder) functions on $\mathbb{A}$ we denote by
$\mathfrak{\bar{L}}\left(  \mathbb{A}\right)  $ the space of real-valued
bounded quasilocal functions on $\mathbb{A}$ that is the closure in the
topology of uniform convergence of the algebra of cylinder functions \cite{Ge}.

If $\mathbb{V}$ is denumerable we denote by $\mathcal{E}_{\mathbb{A}}$ the
product $\sigma$algebra $\mathcal{A}^{\otimes\mathbb{V}}$ on $\mathbb{A}%
:=A^{\mathbb{V}}.$

\subsubsection{Graphs}

We recall some basic definition of graph theory useful to give a mathematical
definition of consensus for the system. The connection of graph theory with
Markov chains will be exploited in the next section. We refer the reader to
basic textbooks such as \cite{Bo} and \cite{St} for an account on this subject.

A directed graph $G$ is a ordered pair of sets $\left(  V,E\right)  $ where
$V$ is a finite set called \emph{set of vertices} and $E\subseteq V\times V$
is called \emph{set of edges }or\emph{\ bonds.} $G^{\prime}=\left(  V^{\prime
},E^{\prime}\right)  $ such that $V^{\prime}\subseteq V$ and $E^{\prime
}\subseteq\left(  V^{\prime}\times V^{\prime}\right)  \cap E$ is said to be a
\emph{subgraph} of $G$ and this property is denoted by $G^{\prime}\subseteq
G.$ If $G^{\prime}\subseteq G,$ we denote by $V\left(  G^{\prime}\right)  $
and $E\left(  G^{\prime}\right)  $ respectively the set of vertices and the
collection of the edges of $G^{\prime}.\ \left\vert V\left(  G^{\prime
}\right)  \right\vert $ is called the \emph{order} of $G^{\prime}$ while
$\left\vert E\left(  G^{\prime}\right)  \right\vert $ is called its
\emph{size}. Given $G_{1},G_{2}\subseteq G,$ we denote by $G_{1}\cup
G_{2}:=\left(  V\left(  G_{1}\right)  \cup V\left(  G_{2}\right)  ,E\left(
G_{1}\right)  \cup E\left(  G_{2}\right)  \right)  \subset G$ the \emph{graph
union }of $G_{1}$ and $G_{2}.$ Moreover, we say that $G_{1},G_{2}\subseteq G$
are \emph{disjoint} if $V\left(  G_{1}\right)  \cap V\left(  G_{2}\right)
=\varnothing.$ For any $E^{\prime}\subseteq E,$ we denote by $G\left(
E^{\prime}\right)  :=\left(  V,E^{\prime}\right)  $ the \emph{spanning} graph
of $E^{\prime}.$ We also define $V\left(  e\right)  $ the subset $\left\{
v,v^{\prime}\right\}  $ of $V$ such that $e$ is either equal to $\left(
v,v^{\prime}\right)  $ or to $\left(  v^{\prime},v\right)  $ and consequently
\begin{equation}
V\left(  E^{\prime}\right)  :=\left(  \bigcup_{e\in E^{\prime}}e\right)
\subset V\ .
\end{equation}
Given $V^{\prime}\subseteq V,$ we set
\begin{equation}
E\left(  V^{\prime}\right)  :=\left\{  e\in E:e\subset V^{\prime}\right\}
\end{equation}
and denote by $G\left[  V^{\prime}\right]  :=\left(  V^{\prime},E\left(
V^{\prime}\right)  \right)  $ that is called the subgraph of $G$
\emph{induced} or \emph{spanned} by $V^{\prime}.$

Two vertices $u,v$ are said to be \emph{adjacents} if belong to the same bond
i.e. if $V\left(  e\right)  =\left\{  u,v\right\}  .$ If $e=\left(
u,v\right)  ,$ $e$ is said to be \emph{outgoing} from $u$ and \emph{ingoing}
in $v.$ Let
\begin{equation}
E_{v}^{-}:=\left\{  e\in E:e=\left(  u,v\right)  ,\ u\in V\right\}
\ ,\ E_{v}^{+}:=\left\{  e\in E:e=\left(  v,u\right)  ,\ u\in V\right\}
\end{equation}
be the set of edges respectively ingoing in $v,$ outgoing from $v.$ We denote by
$\mathcal{N}^{-}\left(  v\right)  :=\left(  \cup_{e\in E_{v}^{-}}V\left(
e\right)  \right)  \subseteq V$ the \emph{closed ingoing neighborhood} of $v$
and by $\mathcal{N}^{+}\left(  v\right)  :=\left(  \cup_{e\in E_{v}^{+}%
}V\left(  e\right)  \right)  \subseteq V$ the \emph{closed outgoing
neighborhood} of $v.$ Moreover, for any $W\subset V,$ we set $\mathcal{N}%
^{+}\left(  W\right)  :=\cup_{v\in W}\mathcal{N}^{+}\left(  v\right)  $ to be
the closed outgoing neighborhood of $W.$ Given $v\in V,$ we set $\mathcal{N}%
_{1}^{+}\left(  v\right)  :=\mathcal{N}^{+}\left(  v\right)  $ and, for
$k\geq2,\ \mathcal{N}_{k}^{+}\left(  v\right)  :=\mathcal{N}^{+}\left(
\mathcal{N}_{k-1}^{+}\left(  v\right)  \right)  $ to be the \emph{outgoing
}$k$\emph{-neighborhood} of $v.$ Given two vertices $u$ and $v,$ $v$ is said
to \emph{communicate} with $u$ if there exists $k\geq1$ such that
$u\in\mathcal{N}_{k}^{+}\left(  v\right)  .$ Therefore, $u,v\in V$ are said to
be \emph{connected} if one communicates with the other. Indeed, since if
$u\in\mathcal{N}_{k}^{+}\left(  v\right)  $ for some $k\geq1,$ then
$u\in\mathcal{N}_{l}^{+}\left(  v\right)  ,\ \forall l>k,$ and for $u$ and $v$ to
be connected there must be $k_{1},k_{2}\geq1$ such that $u\in\mathcal{N}%
_{k_{1}}^{+}\left(  v\right)  $ and $v\in\mathcal{N}_{k_{2}}^{+}\left(
u\right)  ,$ that is $u\in\mathcal{N}_{k_{1}\vee k_{2}}^{+}\left(  v\right)
,v\in\mathcal{N}_{k_{1}\vee k_{2}}^{+}\left(  u\right)  .\ G$ is then said to
be \emph{strongly connected} if any two distinct vertices are connected. The
maximal connected subgraphs of $G$ are called \emph{components} of $G$ and to
denote that $G^{\prime}\subset G$ is a component of $G$ we write $G^{\prime
}\sqsubset G.$

An example of directed graph is the one which can be associated to a Markov
chain. In this case, $V$ coincides with the set of states of the chain and,
denoting by $P$ the transition matrix associated to the chain, $E=E\left(
P\right)  :=\left\{  \left(  u,v\right)  \in V\times V:P_{u,v}>0\right\}  .$
Then, the directed graph associated to the Markov chain with transition matrix
$P$ is denoted by $G\left(  P\right)  .$ Hence, the Markov chain and therefore
$P$ are said to be \emph{irreducible} if and only if $G\left(  P\right)  $ is
strongly connected.

In the following, if $e=\left(  u,v\right)  $ is an edge of a directed graph
$\left(  V,E\right)  ,$ we will occasionally note $\bar{e}$ for the edge
$\left(  v,u\right)  \in E.$

\subsection{Description of the model and results}

In the following, unless differently specified, we will be concerned only with
graphs $G$ being subgraphs of the complete directed graph $\mathbf{G}=\left(
\mathbf{V},\mathbf{E}\right)  $ of finite order where $\mathbf{E}:=\left(
\mathbf{V}\times\mathbf{V}\right)  .$

A bond (or edge) configuration is a map $\mathbf{E}\ni e\longmapsto\omega
_{e}\in\left\{  0,1\right\}  $ so that a bond $e$ is said to be \emph{open} if
$\omega_{e}=1.$ Setting $\forall u,v\in\mathbf{V},\omega_{u,v}:=\omega
_{e}\delta_{e,\left(  u,v\right)  }$ and defining
\begin{equation}
\Omega:=\left\{  \omega\in\left\{  0,1\right\}  ^{\mathbf{E}}:\forall
u\in\mathbf{V}\ ,\ \omega_{u,u}=1\right\}  \ , \label{Omega}%
\end{equation}
we define
\begin{equation}
\Omega\ni\omega\longmapsto E\left(  \omega\right)  :=\left\{  e\in
\mathbf{E}:\omega_{e}=1\right\}  \in\mathcal{P}\left(  \mathbf{E}\right)
\label{Eomega}%
\end{equation}
and consequently
\begin{equation}
G\left(  \omega\right)  :=G\left(  E\left(  \omega\right)  \right)
\subseteq\mathbf{G}\ . \label{G}%
\end{equation}
We also set, $\forall v\in\mathbf{V},$%
\begin{align}
E_{v}^{-}\left(  \omega\right)   &  :=\left\{  e\in E\left(  \omega\right)
:e=\left(  u,v\right)  ,\ u\in\mathbf{V}\right\}  \ ,\label{defE-vomega}\\
E_{v}^{+}\left(  \omega\right)   &  :=\left\{  e\in E\left(  t\right)
:e=\left(  v,u\right)  ,\ u\in\mathbf{V}\right\}  \ ,\\
\mathcal{N}^{\pm}\left(  v,\omega\right)   &  :=\left(  \cup_{e\in E_{v}^{\pm
}\left(  \omega\right)  }V\left(  e\right)  \right)  \ ,\\
\ \mathcal{N}_{1}^{+}\left(  v,\omega\right)   &  :=\mathcal{N}^{+}\left(
v,\omega\right)  \ ;\ \mathcal{N}_{k}^{+}\left(  v,\omega\right)
:=\mathcal{N}^{+}\left(  \mathcal{N}_{k-1}^{+}\left(  v,\omega\right)
\right)  \ ,\;k\geq2\ .
\end{align}
Moreover, given a $\Omega$-valued sequence $\left\{  \omega\left(
t\right)  \right\}  _{t\geq0}$ we set $E\left(  t\right)  :=E\left(
\omega\left(  t\right)  \right)  ,G\left(  t\right)  :=G\left(  \omega\left(
t\right)  \right)  $ as well as, $\forall v\in\mathbf{V},E_{v}^{\pm}\left(
t\right)  :=E_{v}^{\pm}\left(  \omega\left(  t\right)  \right)  $ and $\forall
k\geq1,\mathcal{N}_{k}^{\pm}\left(  v,t\right)  :=\mathcal{N}_{k}^{\pm}\left(
v,\omega\left(  t\right)  \right)  .$

A \emph{belief configuration} is a map $\mathbf{V}\ni v\longmapsto X_{v}%
\in\left[  0,1\right]  .$ We set $\Xi:=\left[  0,1\right]  ^{\mathbf{V}}$ and
consider the sequence $\left\{  X\left(  t\right)  \right\}  _{t\geq0}$
representing the beliefs evolution in time.

\subsubsection{Beliefs dynamics}

The beliefs evolution is given by the system of equations
\begin{equation}
\left\{
\begin{array}
[c]{l}%
\begin{array}
[c]{ll}%
X_{v}\left(  t+1\right)  := & \frac{\sum_{u\in\mathcal{N}^{-}\left(
v,t\right)  }r_{u,v}X_{u}\left(  t\right)  }{\sum_{u\in\mathcal{N}^{-}\left(
v,t\right)  }r_{u,v}}=\frac{\sum_{u\in\mathbf{V}}r_{u,v}\omega_{u,v}\left(
t\right)  X_{u}\left(  t\right)  }{\sum_{u\in\mathbf{V}}r_{u,v}\omega
_{u,v}\left(  t\right)  }\\
& =X_{v}\left(  t\right)  +\frac{\sum_{u\in\mathbf{V}}r_{u,v}\omega
_{u,v}\left(  t\right)  \left[  X_{u}\left(  t\right)  -X_{v}\left(  t\right)
\right]  }{\sum_{u\in\mathbf{V}}r_{u,v}\omega_{u,v}\left(  t\right)  }%
\end{array}
\\
X_{v}\left(  0\right)  =X_{v}^{0}%
\end{array}
\right.  \;,\;v\in\mathbf{V}\ ,\ t\geq0\ , \label{evX}%
\end{equation}
which, by (\ref{defE-vomega}), can be rewritten as
\begin{equation}
\left\{
\begin{array}
[c]{l}%
\begin{array}
[c]{ll}%
X_{v}\left(  t+1\right)  = & \frac{\sum_{e\in E_{v}^{-}\left(  t\right)
}r_{e}\sum_{u\in\mathbf{V}}\delta_{e,\left(  u,v\right)  }X_{u}\left(
t\right)  }{\sum_{e\in E_{v}^{-}\left(  t\right)  }r_{e}}=X_{v}\left(
t\right)  +\frac{\sum_{e\in E_{v}^{-}\left(  t\right)  }r_{e}\Delta
_{e}X\left(  t\right)  }{\sum_{e\in E_{v}^{-}\left(  t\right)  }r_{e}}\\
& =X_{v}\left(  t\right)  +\frac{\sum_{e\in E_{v}^{-}}r_{e}\omega_{e}\left(
t\right)  \Delta_{e}X\left(  t\right)  }{\sum_{e\in E_{v}^{-}}r_{e}\omega
_{e}\left(  t\right)  }=\frac{\sum_{e\in E_{v}^{-}}r_{e}\omega_{e}\left(
t\right)  \sum_{u\in\mathbf{V}}\delta_{e,\left(  u,v\right)  }X_{u}\left(
t\right)  }{\sum_{e\in E_{v}^{-}}r_{e}\omega_{e}\left(  t\right)  }%
\end{array}
\\
X_{v}\left(  0\right)  =X_{v}^{0}%
\end{array}
\right.  \;,\;v\in\mathbf{V}\ ,\ t\geq0\ , \label{evX1}%
\end{equation}
where, $\forall e\in\mathbf{E},$%
\begin{equation}
\Delta_{e}X\left(  t\right)  :=\left(  X_{u}\left(  t\right)  -X_{v}\left(
t\right)  \right)  \mathbf{1}_{\left(  v,u\right)  }\left(  e\right)
\label{deltaX}%
\end{equation}
and $r_{e}\in\left[  0,1\right]  $ is the communication rate between the
agents labelled by the the vertices incident in $e,$ namely $r_{u,v}%
:=r_{e}\delta_{e,\left(  u,v\right)  },$ which represents the confidence level
assigned by the agent $v$ to the belief of the agent $u.$

\subsubsection{Communication channels dynamics}

For any $v\in\mathbf{V},$ the $\mathcal{P}\left(  \mathbf{E}\right)  $-valued
sequence $\left\{  E_{v}^{-}\left(  t\right)  \right\}  _{t\geq0},$ as well as
the $\mathbf{G}$-valued sequence $\left\{  G\left(  t\right)  \right\}
_{t\geq0},$ are constructed by $\left\{  \omega\left(  t\right)  \right\}
_{t\geq0}$ through the random evolution described by the collection of regular
conditional probabilities
\begin{align}
\mathbb{P}\left\{  \omega_{e}\left(  t+1\right)  =\omega_{e}^{\prime}|X\left(
t\right)  \right\}   &  =\delta_{\omega_{e}^{\prime},1}p\left(  \left\vert
\Delta_{e}X\left(  t\right)  \right\vert \right)  +\delta_{\omega_{e}^{\prime
},0}\left(  1-p\left(  \left\vert \Delta_{e}X\left(  t\right)  \right\vert
\right)  \right) \label{pomega}\\
&  =\omega_{e}^{\prime}p\left(  \left\vert \Delta_{e}X\left(  t\right)
\right\vert \right)  +\left(  1-\omega_{e}^{\prime}\right)  \left(  1-p\left(
\left\vert \Delta_{e}X\left(  t\right)  \right\vert \right)  \right)
\ ,\;e\in\mathbf{E}\ ,\nonumber
\end{align}
where $X\left(  t\right)  \in\Xi$ is the belief configuration at time $t\geq0$
and $p:\left[  0,1\right]  \circlearrowleft$ is a nonincreasing function such
that $p\left(  0\right)  =1.$

Notice that, for any $t\geq0,$ given $e,f\in\mathbf{E}$ such that $\Delta
_{e}X\left(  t\right)  =\Delta_{f}X\left(  t\right)  ,$ the r.v.'s $\omega
_{e}\left(  t+1\right)  $ and $\omega_{\bar{e}}\left(  t+1\right)  $ have the
same conditional probabilities w.r.t. $X\left(  t\right)  .$ In particular
this holds for $e=\left(  u,v\right)  $ and $f=\bar{e}=\left(  v,u\right)  ,$
although the edge configurations $\omega_{e}$ and $\omega_{\bar{e}}$ are
different in general.

In the following we will consider $\forall e\in\mathbf{E},r_{e}>0.$ As a
matter of fact, since the $r_{e}$'s are fixed, we can restrict ourselves to
consider instead of $\mathbf{G}$ each component of its spanning subgraph
$\mathbf{G}_{r}:=G\left(  \mathbf{E}_{r}\right)  ,$ where
\begin{equation}
\mathbf{E}_{r}:=\left\{  e\in\mathbf{E}:r_{e}>0\right\}  \ ,
\end{equation}
because, by (\ref{evX}), if $G_{1},G_{2}\sqsubset\mathbf{G}_{r}$ the evolution
of the beliefs labeled by the vertices of $G_{1}$ is never affected by those
labeled by the vertices of $G_{2}.$

Moreover, since it is reasonable to assume that the agents put maximal
confidence on their own beliefs, we can set $r_{\left(  u,u\right)  }=1,$ for
any $u\in\mathbf{V.}$

\subsubsection{Results for the finite system}

Let $\mathbf{V}$ be a finite set. If $X^{0}\in\Xi$ is such that $\forall
v\in\mathbf{V},X_{v}^{0}=x\in\left[  0,1\right]  ,$ then by (\ref{evX})
$X\left(  t\right)  =X^{0},\forall t\geq0.$ Hence these configuration, called
\emph{consensus} configurations, are stationary for the system evolution.

Making use of probabilistic techniques borrowed from percolation theory, in
the next section we will prove the following result.

\begin{theorem}
\label{main}The agents system reaches consensus for any realization of the
initial value of the noise $\omega_{0}\in\Omega$ and any initial configuration
$X^{0}\in\left\{  X\in\Xi:\Gamma\left( W\left(  X\right)\right)  >0\right\}  .$
\end{theorem}

Where in view of the definition of $\Gamma:\Xi\longrightarrow\lbrack0,1)$
given in (\ref{Gamma}), $\Gamma\left(  W\left(  X\right)  \right)  $ is
defined in (\ref{GammaW}).

Moreover, we will also prove that, the random sequence $\left\{  \left(
X\left(  t\right)  ,\omega\left(  t\right)  \right)  \right\}  _{t\geq0}$
started at $\left(  X^{0},\omega^{0}\right)  \in\left\{  X\in\Xi:\Gamma\left(  W\left(  X\right)  \right)
>0\right\}  \times\Omega$ in the limit as $t$ tends to infinity
weakly converges at geometric rate to $\left(  X^{\infty},\bar{1}\right)  ,$
where $X^{\infty}\in\Xi$ is such that, for any $v\in\mathbf{V},X_{v}^{\infty
}=x,$ for some $x\in\left[  0,1\right]  ,$ and $\bar{1}$ is the element of
$\Omega$ such that all its entries are equal to $1.$

As a byproduct of this result we will obtain that the random sequence
$\left\{  \left(  X\left(  2t-2\right)  ,X\left(  2t-1\right)  \right)
\right\}  _{t\geq0}$ started at $\left(  X\left(  -2\right)  ,X\left(
-1\right)  \right)  =\left(  X^{0},X^{0}\right)  $ with $X^{0}\in\left\{
X\in\Xi:\Gamma\left(  W\left(  X\right)  \right)  >0\right\}  ,$ which turns out to be an
homogeneous Markov chain with degenerate transition probability kernel, i.e. a
dynamical system on $\Xi^{2},$ will also weakly converge, in the limit as $t$
tends to infinity, to $\left(  X^{\infty},X^{\infty}\right)  $ at geometric rate.

\subsubsection{Results for the very large system}

Let $\mathbf{V}:=\mathbb{N},\mathbf{E}:=\left\{  \left(  u,v\right)
\in\mathbb{N}\times\mathbb{N}\right\}  $ and set $\Xi:=\left[  0,1\right]
^{\mathbb{N}}$ and $\Omega$ as in (\ref{Omega}). Given $N\in\mathbb{N},$ let
$\mathbf{V}_{N}:=\left\{  1,..,N\right\}  \subset\mathbb{N}$ and
$\mathbf{E}_{N}:=\left\{  \left(  u,v\right)  \in\mathbf{V}_{N}\times
\mathbf{V}_{N}\right\}  .$ We denote by $X_{N}:=\left(  X_{1},..,X_{N}\right)
$ the element of $\Xi_{N}:=\left[  0,1\right]  ^{N}$ representing the
restriction of the beliefs configuration $X\in\Xi$ to $\mathbf{V}_{N},$ by
$\omega_{N}$ the restriction of the configuration $\omega\in\Omega$ to
$\Omega_{N}:=\left\{  0,1\right\}  ^{\mathbf{E}_{N}}$ and, by (\ref{Eomega}),
if $E:=E\left(  \omega\right)  ,$ we set $E_{N}:=E\cap\mathbf{E}_{N}.$

Assuming that $R:=\left\{  \left(  u,v\right)  \in\mathbf{V}\times
\mathbf{V}:r_{u,v}>0\right\}  $ is finite, in Proposition \ref{cls} we prove
that the random sequence $\left\{  \left(  X\left(  2t-2\right)  ,X\left(
2t-1\right)  \right)  \right\}  _{t\geq0}$ started at $\left(  X\left(
-2\right)  ,X\left(  -1\right)  \right)  =\left(  X^{0},X^{0}\right)  $ with
$X^{0}\in\left\{  X^{\prime}\in\Xi:\inf_{N\in\mathbb{N}}\Gamma\left( W\left(
X_{N}^{\prime}\right)\right)  >0\right\}  ,$ in the limit as $t$ tends to infinity,
weakly converges at geometric rate to $\left(  X^{\infty},X^{\infty}\right),$
where, as in the finite system case, $X^{\infty}\in\Xi$ is such
that $\forall v\in\mathbf{V},X_{v}^{\infty}=x$ for some $x\in\left[
0,1\right]  .$

\paragraph{Monokinetic-type limit}

Let us consider the non-linear Markov chain with degenerate transition
probability kernel $\left\{  \mathbf{Z}_{t}^{\varrho}\right\}  _{t\geq0}$ on
$\left(  \left[  0,1\right]  ^{2},\mathcal{B}\left(  \left[  0,1\right]
^{2}\right)  \right)  $ such that, for any $t\geq0$ and any bounded measurable
$\varphi:\left[  0,1\right]  ^{2}\longrightarrow\mathbb{R},\mathbf{Z}%
_{t}^{\varrho}:=\left(  Z_{t}^{\varrho,\left(  1\right)  },Z_{t}%
^{\varrho,\left(  2\right)  }\right)  $ and
\begin{align}
\mathbb{E}\left[  \varphi\left(  \mathbf{Z}_{t+1}^{\varrho}\right)
|\mathbf{Z}_{t}^{\varrho}\right]   &  =\mathbb{E}\left.  \left[
\varphi\left(  \mathbf{Z}_{t+1}^{\varrho,\left(  1\right)  },\mathbf{Z}%
_{t+1}^{\varrho,\left(  2\right)  }\right)  \right.  \left(  Z_{t}%
^{\varrho,\left(  1\right)  },Z_{t}^{\varrho,\left(  2\right)  }\right)
\right] \\
&  =\varphi\left(  \theta_{\mu_{t}}^{\varrho}\left(  Z_{t}^{\varrho,\left(
1\right)  },Z_{t}^{\varrho,\left(  2\right)  }\right)  ,\theta_{\mu_{t}%
}^{\varrho}\circ\theta_{\mu_{t}}^{\varrho}\left(  Z_{t}^{\varrho,\left(
1\right)  },Z_{t}^{\varrho,\left(  2\right)  }\right)  \right) \nonumber
\end{align}
where, for any $t\geq0,\mu_{t}$ is the law of $\mathbf{Z}_{t}^{\varrho}$ and,
for any $\mu\in\mathfrak{P}\left(  \left[  0,1\right]  ^{2},\mathcal{B}\left(
\left[  0,1\right]  ^{2}\right)  \right)  ,$%
\begin{equation}
\left[  0,1\right]  ^{2}\ni\left(  x,y\right)  \longmapsto\theta_{\mu
}^{\varrho}\left(  x,y\right)  :=\frac{\int_{\left[  0,1\right]  ^{2}}%
\mu\left(  dx^{\prime},dy^{\prime}\right)  p\left(  \left\vert x-x^{\prime
}\right\vert \right)  \varrho\left(  y,y^{\prime}\right)  y^{\prime}}%
{\int_{\left[  0,1\right]  ^{2}}\mu\left(  dx^{\prime},dy^{\prime}\right)
p\left(  \left\vert x-x^{\prime}\right\vert \right)  \varrho\left(
y,y^{\prime}\right)  }\in\left[  0,1\right]  \ . \label{theta}%
\end{equation}

In other words, the sequence $\left\{  \mathbf{Z}_{t}^{\varrho}\right\}
_{t\geq0}$ represents the trajectories of the non-homogeneous dynamical system
defined on $\left[  0,1\right]  ^{2}$ by the sequence of mappings $\left\{
\Theta_{\mu_{t}}^{\varrho}\right\}  _{t\geq1}$ such that, for any $t\geq1,$%
\begin{equation}
\left[  0,1\right]  ^{2}\ni\left(  x,y\right)  \longmapsto\Theta_{\mu_{t}%
}^{\varrho}\left(  x,y\right)  :=\left(  \theta_{\mu_{t}}^{\varrho}\left(
x,y\right)  ,\theta_{\mu_{t}}^{\varrho}\circ\theta_{\mu_{t}}^{\varrho}\left(
x,y\right)  \right)  \in\left[  0,1\right]  ^{2}\ ,
\end{equation}
i.e. the projection on the second component of the homogeneous dynamical
system
\begin{equation}
\mathbb{N}\times\left[  0,1\right]  ^{2}\ni\left(  t,\mathbf{z}\right)
\longmapsto\Phi\left(  t,\mathbf{z}\right)  :=\left(  \sigma\left(  t\right)
,\vartheta_{\mu_{t}}^{\varrho}\left(  \mathbf{z}\right)  \right)
\in\mathbb{N}\times\left[  0,1\right]  ^{2}\ ,
\end{equation}
where $\sigma$ is the left shift operator, namely $\mathbb{N}\ni
t\longmapsto\sigma\left(  t\right)  :=t+1\in\mathbb{N},$ and, for any
$t\in\mathbb{N},$ if $\mathbf{z=}\left(  x,y\right)  ,\vartheta_{\mu_{t}%
}^{\varrho}\left(  \mathbf{z}\right)  :=\Theta_{\mu_{t}}^{\varrho}\left(
x,y\right)  $ so that $\mathbf{Z}_{t+1}^{\varrho}=\vartheta_{\mu_{t}}%
^{\varrho}\left(  \mathbf{Z}_{t}^{\varrho}\right)  .$

Under the assumption that $\mathsf{supp}p=\left[  0,1\right]  $ and that for
any $u,v\in\mathbf{V}_{N},r_{u,v}:=\varrho\left(  X_{u},X_{v}\right)  ,$ where
$\mathsf{supp}\varrho=\left[  0,1\right]  ^{2}$ and $\varrho\in Lip\left(
\left[  0,1\right]  ^{2},\left[  0,1\right]  \right)  ,$ denoting by $\mu
_{N}^{X,Y}:=\frac{1}{N}\sum_{v\in\mathbf{V}_{N}}\delta_{\left\{
X_{v}\right\}  }\otimes\delta_{\left\{  Y_{v}\right\}  }$ the empirical
probability measure on $\left(  \left[  0,1\right]  ^{2},\mathcal{B}\left(
\left[  0,1\right]  ^{2}\right)  \right)  $ relative to the believes
configuration $\left(  X,Y\right)  \in\Xi_{N}^{2},$ by means of Theorem
\ref{mf4}, we show that, if the sequence $\left\{  \mu_{N}^{X^{\left(
1\right)  },X^{\left(  2\right)  }}\right\}  _{N\geq1}\subset\mathfrak{P}%
\left(  \left[  0,1\right]  ^{2},\mathcal{B}\left(  \left[  0,1\right]
^{2}\right)  \right)  $ such that, for any $N\geq1,\mu_{N}^{X^{\left(
1\right)  },X^{\left(  2\right)  }}$ is supported on the initial datum
$\left(  X_{N}\left(  -2\right)  ,X_{N}\left(  -1\right)  \right)  =\left(
X^{\left(  1\right)  },X^{\left(  2\right)  }\right)  \in\Xi_{N}^{2}$ of the
Markov chain $\left\{  \left(  X_{N}\left(  2t-2\right)  ,X_{N}\left(
2t-1\right)  \right)  \right\}  _{t\geq0},$ weakly converges to $\mu
\in\mathfrak{P}\left(  \left[  0,1\right]  ^{2},\mathcal{B}\left(  \left[
0,1\right]  ^{2}\right)  \right)  ,$ then, for any $t\geq1,$ the sequence \linebreak
$\left\{  \mu_{N}^{\left(  X\left(  2t-2\right)  ,X\left(  2t-1\right)
\right)  }\right\}  _{N\geq1}\subset\mathfrak{P}\left(  \left[  0,1\right]
^{2},\mathcal{B}\left(  \left[  0,1\right]  ^{2}\right)  \right)  $ weakly
converges to the probability distribution $\mu_{t}\in\mathfrak{P}\left(
\left[  0,1\right]  ^{2},\mathcal{B}\left(  \left[  0,1\right]  ^{2}\right)
\right)  $ of $\mathbf{Z}_{t}^{\varrho}.$

To our knowledge the proof of this result is not standard and relies on the
self-average property of the random interactions among the agents which are
only locally mean-field although, at any time $t\geq1,$ their law depend on
the value of the believes variables at time $t-1.$ This is a crucial fact that
prevented us from using well established techniques such as those described in
the dynamical systems literature in \cite{Ta} section 3 and reference therein,
as well as in \cite{CCH} and in particular in \cite{GO} for what concerns the
kinetic limit literature.

Moreover, in Proposition \ref{mfc} we prove that, if $\varrho$ is also
strictly positive and assumes the value $1$ on the set $\left\{  \left(
x,y\right)  \in\left[  0,1\right]  ^{2}:x=y\right\}  ,$ given an initial datum
$\mu\in\mathfrak{P}\left(  \left[  0,1\right]  ^{2},\mathcal{B}\left(  \left[
0,1\right]  ^{2}\right)  \right)  ,$ the sequence $\left\{  \mu_{t}\right\}
_{t\geq0}\subset\mathfrak{P}\left(  \left[  0,1\right]  ^{2},\mathcal{B}%
\left(  \left[  0,1\right]  ^{2}\right)  \right)  $ weakly converges
to the Dirac mass at $\left(  x,x\right)  ,$ for some $x\in\left[  0,1\right]
,$ at geometric rate.

\section{Finite system evolution}

Let $\mathbf{V}$ be a finite set. The evolution of the system is given by the following algorithm:

\begin{algorithm}
\label{A1}

\begin{enumerate}
\item Label the elements of $\mathbf{V}$ from $1$ to $N$ in such a way that
$\mathbf{V}:=\left\{  1,..,N\right\}  $ and consequently label $\left(
i,j\right)  $ the elements of $\mathbf{E}:=\mathbf{V}\times\mathbf{V},$ then
go to the next step.

\item Set $t:=0,X\left(  0\right)  =\left(  X_{1}\left(  0\right)
,..,X_{N}\left(  0\right)  \right)  :=\left(  X_{1}^{0},..,X_{N}^{0}\right)
\in\left[  0,1\right]  ^{N},\omega\left(  0\right)  :=\left\{  \omega
_{i,j}^{0}\right\}  _{\left(  i,j\right)  \in\mathbf{E}}\in\left\{
0,1\right\}  ^{\mathbf{E}}$ such that $\forall i=1,..,N,\omega_{i,i}^{0}=1,$
and go to the next step.

\item Set $i:=1$ and go to the next step.

\begin{enumerate}
\item Set $j:=1$ and go to the next step.

\item Compute $p_{i,j}\left(  t\right)  :=p\left(  \left\vert X_{i}\left(
t\right)  -X_{j}\left(  t\right)  \right\vert \right)  $ and form the vector
\begin{equation}
\mathbf{p}\left(  t\right)  :=\left(  p_{1,1}\left(  t\right)  ,..,p_{1,N}%
\left(  t\right)  ,p_{2,1}\left(  t\right)  ,..,p_{2,N}\left(  t\right)
..,p_{i,1}\left(  t\right)  ,..,p_{i,j}\left(  t\right)  \right)
\end{equation}
and go to the step.

\item Set $j:=j+1.$ If $j+1\leq N$ go back to step 3.b, otherwise go to the
next step.

\item Set $i:=i+1.$ If $i+1\leq N$ go back to step 3.a, otherwise go to the
next step.
\end{enumerate}

\item Set $i:=1$ and go to the next step.

\begin{enumerate}
\item Compute $X_{i}\left(  t+1\right)  $ according to (\ref{evX}) and form
the vector $X\left(  t+1\right)  :=\left(  X_{1}\left(  t+1\right)  ,\right.
$ $\left.  ..,X_{i}\left(  t+1\right)  \right)  ,$ then go to the next step.

\item Set $i:=i+1.$ If $i+1\leq N$ go back to step 4.a, otherwise go to the
next step.
\end{enumerate}

\item Read $X\left(  t\right)  =\left(  X_{1}\left(  t\right)  ,..,X_{N}%
\left(  t\right)  \right)  .$ If $X\left(  t+1\right)  =X\left(  t\right)  $
stop, otherwise go to the next step.

\item Set $i:=1$ and go to the next step.

\begin{enumerate}
\item Set $j=1$ and go to the next step.

\item Read the $p_{i,j}\left(  t\right)  $ entry of the vector $\mathbf{p}%
\left(  t\right)  \mathbf{.}$ Sample a random variable $U$ uniformly
distributed on $\left[  0,1\right]  .$ If $U\leq p_{i,j}\left(  t\right)  $
then set $\omega_{i,j}\left(  t+1\right)  :=\omega_{i,j}\left(  t\right)  $ if
$\omega_{i,j}\left(  t\right)  =1,$ otherwise set $\omega_{i,j}\left(
t+1\right)  :=1-\omega_{i,j}\left(  t\right)  .$ If $U>p_{i}\left(  t\right)
$ then set $\omega_{i,j}\left(  t+1\right)  :=1-\omega_{i,j}\left(  t\right)
$ if $\omega_{i,j}\left(  t\right)  =1,$ otherwise set $\omega_{i,j}\left(
t+1\right)  :=\omega_{i,j}\left(  t\right)  .$ Form the vector
\begin{equation}
\omega\left(  t+1\right)  :=\left(  \omega_{1,1}\left(  t+1\right)
,..,\omega_{1,N}\left(  t+1\right)  ,..,\omega_{i,1}\left(  t+1\right)
,..,\omega_{i,j}\left(  t+1\right)  \right)
\end{equation}
and go to the next step.

\item Set $j:=j+1.$ If $j+1\leq N$ go back to step 6.b, otherwise go to the
next step.

\item Set $i:=i+1.$ If $i+1\leq N$ go back to step 6.a, otherwise go to the
next step.
\end{enumerate}

\item Set $t:=t+1,X\left(  0\right)  :=X\left(  t+1\right)  ,\omega\left(0\right)  :=\omega\left(  t+1\right)  $ 
and go back to step 3.
\end{enumerate}
\end{algorithm}

In terms of stochastic process the system evolution can be described as follows.

Let $\Omega\ni\omega\longmapsto P\left(  \omega\right)  \in St\left(
\mathbf{V}\right)  $ the stochastic matrix-valued function on $\mathbb{R}%
^{\mathbf{V}}$ such that, for any $\omega\in\Omega,$
\begin{equation}
P_{v,u}\left(  \omega\right)  :=\frac{\sum_{e\in E_{v}^{-}\left(
\omega\right)  }\delta_{e,\left(  u,v\right)  }r_{e}}{\sum_{e\in E_{v}%
^{-}\left(  \omega\right)  }r_{e}}=\frac{r_{u,v}\mathbf{1}_{\mathcal{N}%
^{-}\left(  v,\omega\right)  }\left(  u\right)  }{\sum_{u\in\mathcal{N}%
^{-}\left(  v,\omega\right)  }r_{u,v}}=\frac{r_{u,v}\omega_{u,v}}{\sum
_{u\in\mathbf{V}}r_{u,v}\omega_{u,v}}\ ,\;u,v\in\mathbf{V\ .} \label{Pomega}%
\end{equation}

\begin{remark}
We remark that, given $\omega\in\Omega,v\in\mathbf{V},$ by (\ref{Pomega})
$u\in\mathcal{N}^{-}\left(  v,\omega\right)  $ iff $P_{v,u}\left(
\omega\right)  >0.$ Therefore, denoting by $\overline{G}\left(  \omega\right)
:=G\left(  P\left(  \omega\right)  \right)  $ the graph associated to
$P\left(  \omega\right)  ,$ this is the spanning graph of $\overline{E}\left(
\omega\right)  :=\left\{  e\in\mathbf{E}:e=\left(  u,v\right)  \ if\ \left(
v,u\right)  \in E\left(  \omega\right)  \right\}  .$
\end{remark}

Considering $\Xi:=\left[  0,1\right]  ^{\mathbf{V}}\subset\mathbb{R}%
^{\mathbf{V}}$ endowed with the norm $\left\Vert X\right\Vert :=\left\Vert
X\right\Vert _{\infty}=\sup_{v\in\mathbf{V}}\left\vert X_{v}\right\vert ,$ let
$\left(  \mathfrak{X},\mathcal{F}\right)  $ be the measurable space such that
$\mathfrak{X}:=\Xi\times\Omega$ and, since $\mathbf{V}$ is a finite set,
$\mathcal{F}:=\mathcal{B}\left(  \Xi\right)  \otimes\mathcal{P}\left(
\Omega\right)  .$

For any $\omega\in\Omega,P\left(  \omega\right)  \in BL\left(  \Xi\right)  ,$
therefore we set
\begin{equation}
\mathfrak{X}\ni\left(  X,\omega\right)  \longmapsto\mathcal{T}_{v}\left(
X,\omega\right)  :=\sum_{u\in\mathbf{V}}P_{v,u}\left(  \omega\right)  X_{u}%
\in\left[  0,1\right]  \ , \label{Tomega}%
\end{equation}
and consider the measurable map
\begin{equation}
\mathfrak{X}\ni\left(  X,\omega\right)  \longmapsto\mathcal{T}\left(
X,\omega\right)  :=\left\{  \mathcal{T}_{v}\left(  X,\omega\right)  \right\}
_{v\in\mathbf{V}}\in\Xi\ . \label{T}%
\end{equation}
Defining, by (\ref{pomega}), the probability kernel from $\left(
\Xi,\mathcal{B}\left(  \Xi\right)  \right)  $ to $\left(  \Omega
,\mathcal{P}\left(  \Omega\right)  \right)  $%
\begin{align}
\mathfrak{X} \ni\left(  X,\omega\right)  \longmapsto\Pi\left(  \omega
|X\right)   &  :=\prod\limits_{e\in\mathbf{E}}\left[  \delta_{\omega_{e}%
,1}p\left(  \left\vert \Delta_{e}X\right\vert \right)  +\delta_{\omega_{e}%
,0}\left(  1-p\left(  \left\vert \Delta_{e}X\right\vert \right)  \right)
\right] \label{pomega1}\\
&  =\prod\limits_{e\in\mathbf{E}}\left[  \omega_{e}p\left(  \left\vert
\Delta_{e}X\right\vert \right)  +\left(  1-\omega_{e}\right)  \left(
1-p\left(  \left\vert \Delta_{e}X\right\vert \right)  \right)  \right]
\in\left[  0,1\right]  \ ,\nonumber
\end{align}
we introduce the positive linear operator on $BM\left(  \mathfrak{X}\right)  $
such that
\begin{equation}
BM\left(  \mathfrak{X}\right)  \ni\varphi\longmapsto\mathfrak{T}\varphi\left(
X,\omega\right)  :=\sum_{\omega^{\prime}\in\Omega}\varphi\left(
\mathcal{T}\left(  X,\omega\right)  ,\omega^{\prime}\right)  \Pi\left(
\omega^{\prime}|X\right)  \in BM\left(  \mathfrak{X}\right)  \ . \label{TT}%
\end{equation}

Let $\mathbb{P}_{0}$ be the probability distribution on $\left(
\mathfrak{X}^{\mathbb{Z}_{+}},\mathfrak{C}\right)  ,$ where $\mathfrak{C}%
:=\mathcal{C}\left(  \Xi\right)  \otimes\mathcal{C}\left(  \Omega\right)  ,$
describing the homogeneous discrete time Markov process started at $\left(
X^{0},\omega^{0}\right)  $ defined by the one-step transition probability
kernel associated to $\mathfrak{T}.$ We denote by $\left\{  \chi_{t}\right\}
_{t\geq0}$ the random process on $\left(  \mathfrak{X}^{\mathbb{Z}_{+}%
},\mathfrak{C},\mathbb{P}_{0}\right)  $ such that, $\forall t\geq0,$%
\begin{equation}
\mathfrak{X}^{\mathbb{Z}_{+}}\ni\mathbf{x}\longmapsto\chi_{t}\left(
\mathbf{x}\right)  =\left(  X\left(  t\right)  ,\omega\left(  t\right)
\right)  \in\mathfrak{X} \label{chit}%
\end{equation}
and by $\left\{  \mathfrak{F}_{t}\right\}  _{t\geq0},$ with $\mathfrak{F}%
_{t}:=\bigvee\limits_{s=0}^{t}\chi_{s}^{-1}\left(  \mathcal{B}\left(
\Xi\right)  \otimes\mathcal{P}\left(  \Omega\right)  \right)  ,$ the
associated natural filtration. Therefore, denoting by $\mathbb{E}_{0}$ the
expectation value w.r.t. $\mathbb{P}_{0},$ for any bounded measurable function
$\varphi$ on $\mathfrak{X,}$%

\begin{equation}
\mathbb{E}_{0}\left[  \varphi\circ\chi_{t+1}|\mathfrak{F}_{t}\right]
=\mathbb{E}_{0}\left[  \varphi\circ\chi_{t+1}|\chi_{t}\right]  =\left(
\mathfrak{T}\varphi\right)  \left(  \chi_{t}\right)  \;\mathbb{P}_{0}-a.s.\ .
\label{TT1}%
\end{equation}

Notice that, by (\ref{Tomega}), $\mathfrak{T}:C\left(  \mathfrak{X}%
,\mathbb{R}\right)  \circlearrowleft,$ that is $\left\{  \chi_{t}\right\}
_{t\geq0}$ is a Feller process.

Setting $\pi_{\omega}:\mathfrak{X}\longmapsto\Omega,\pi_{X}:\mathfrak{X}%
\longmapsto\Xi,$ we denote by $\left\{  \mathfrak{w}_{t}\right\}  _{t\geq
0},\left\{  \mathfrak{x}_{t}\right\}  _{t\geq0}$ the random processes on
$\left(  \mathfrak{X}^{\mathbb{Z}_{+}},\mathfrak{C},\mathbb{P}_{0}\right)  $
such that, $\forall t\geq0,$%
\begin{equation}
\mathfrak{X}^{\mathbb{Z}_{+}}\ni\mathbf{x}\longmapsto\mathfrak{w}_{t}\left(
\mathbf{x}\right)  :=\pi_{\omega}\circ\chi_{t}\left(  \mathbf{x}\right)
=\omega\left(  t\right)  \in\Omega\label{w_t}%
\end{equation}
and
\begin{equation}
\mathfrak{X}^{\mathbb{Z}_{+}}\ni\mathbf{x}\longmapsto\mathfrak{x}_{t}\left(
\mathbf{x}\right)  :=\pi_{X}\circ\chi_{t}\left(  \mathbf{x}\right)  =X\left(
t\right)  \in\Xi\ . \label{x_t}%
\end{equation}
Hence, $\left\{  \chi_{t}\right\}  _{t\geq0}$ can be represented as $\left\{
\left(  \mathfrak{x}_{t},\mathfrak{w}_{t}\right)  \right\}  _{t\geq0}.$ We
also set $\left\{  \mathfrak{F}_{t}^{\omega}\right\}  _{t\geq0},$ with
$\mathfrak{F}_{t}^{\omega}:=\bigvee\limits_{s=0}^{t}\mathfrak{w}_{s}%
^{-1}\left(  \mathcal{P}\left(  \Omega\right)  \right)  ,$ and $\left\{
\mathfrak{F}_{t}^{X}\right\}  _{t\geq0},$ with $\mathfrak{F}_{t}^{X}%
:=\bigvee\limits_{s=0}^{t}\mathfrak{x}_{s}^{-1}\left(  \mathcal{B}\left(
\Xi\right)  \right)  .$

\begin{remark}
\label{Rem1}Notice that neither $\left\{  \mathfrak{w}_{t}\right\}  _{t\geq0}$
nor $\left\{  \mathfrak{x}_{t}\right\}  _{t\geq0}$ are Markov processes.
Indeed, by (\ref{pomega}), for any $t\geq0,\mathfrak{w}_{t+1}$ is independent
of $\mathfrak{w}_{t}.$ Moreover, since $\forall t\geq0,\mathfrak{F}_{t}^{X}$
and $\mathfrak{F}_{t}^{\omega}$ are sub$\sigma$algebras of $\mathfrak{F}_{t},$
for any $B\in\mathcal{P}\left(  \Omega\right)  ,$
\begin{align}
\mathbb{P}_{0}\left(  \left\{  \mathfrak{w}_{t+1}\in B\right\}  |\mathfrak{F}%
_{t}^{\omega}\right)   &  =\mathbb{E}_{0}\left[  \mathbb{E}_{0}\left[
\mathbf{1}_{B}\circ\pi_{\omega}\circ\chi_{t+1}|\mathfrak{F}_{t}\right]
|\mathfrak{F}_{t}^{\omega}\right] \\
&  =\mathbb{E}_{0}\left[  \mathbb{E}_{0}\left[  \mathbf{1}_{B}\circ\pi
_{\omega}\circ\chi_{t+1}|\chi_{t}\right]  \mathfrak{F}_{t}^{\omega}\right]
\nonumber\\
&  =\mathbb{E}_{0}\left[  \mathfrak{T}\left(  \mathbf{1}_{B}\circ\pi_{\omega
}\right)  \left(  \chi_{t}\right)  |\mathfrak{F}_{t}^{\omega}\right]
\nonumber\\
&  =\mathbb{E}_{0}\left.  \left[  \sum_{\omega^{\prime}\in\Omega}%
\mathbf{1}_{B}\left(  \omega^{\prime}\right)  \Pi\left(  \omega^{\prime
}|\mathfrak{x}_{t}\right)  \right\vert \mathfrak{F}_{t}^{\omega}\right]
\nonumber\\
&  =\mathbb{E}_{0}\left.  \left[  \sum_{\omega^{\prime}\in B}\Pi\left(
\omega^{\prime}|\mathfrak{x}_{t}\right)  \right\vert \mathfrak{F}_{t}^{\omega
}\right]  \neq\mathbb{P}_{0}\left(  \left\{  \mathfrak{w}_{t+1}\in B\right\}
|\mathfrak{w}_{t}\right) \nonumber\\
&  =\mathbb{P}_{0}\left\{  \mathfrak{w}_{t+1}\in B\right\}  =\left(
\mathfrak{T}^{t+1}\mathbf{1}_{B}\circ\pi_{\omega}\right)  \left(  X^{0}%
,\omega^{0}\right) \nonumber
\end{align}
while, by (\ref{pomega}), (\ref{TT}) and (\ref{TT1}), $\forall\varphi\in
BM\left(  \Xi,\mathbb{R}\right)  ,$ since for any $t\geq0,\mathfrak{F}_{t}%
^{X}$ is a sub$\sigma$algebra of $\mathfrak{F}_{t},$%
\begin{align}
\mathbb{E}_{0}\left[  \varphi\circ\mathfrak{x}_{t+1}|\mathfrak{F}_{t}%
^{X}\right]   &  =\mathbb{E}_{0}\left[  \varphi\circ\pi_{X}\circ\chi
_{t+1}|\mathfrak{F}_{t}^{X}\right]  =\mathbb{E}_{0}\left[  \mathbb{E}%
_{0}\left[  \varphi\circ\pi_{X}\circ\chi_{t+1}|\mathfrak{F}_{t}\right]
|\mathfrak{F}_{t}^{X}\right] \label{EfiXt+1}\\
&  =\mathbb{E}_{0}\left[  \mathbb{E}_{0}\left[  \varphi\circ\pi_{X}\circ
\chi_{t+1}|\chi_{t}\right]  |\mathfrak{F}_{t}^{X}\right]  =\mathbb{E}%
_{0}\left[  \left(  \mathfrak{T}\left(  \varphi\circ\pi_{X}\right)  \right)
\left(  \chi_{t}\right)  |\mathfrak{F}_{t}^{X}\right] \nonumber\\
&  =\sum_{\omega^{\prime}\in\Omega}\sum_{\mathfrak{w}_{t}\in\Omega}%
\varphi\left(  \mathcal{T}\left(  \mathfrak{x}_{t},\mathfrak{w}_{t}\right)
\right)  \Pi\left(  \omega^{\prime}|\mathfrak{x}_{t}\right)  \Pi\left(
\mathfrak{w}_{t}|\mathfrak{x}_{t-1}\right) \nonumber\\
&  =\sum_{\omega\in\Omega}\varphi\left(  \mathcal{T}\left(  \mathfrak{x}%
_{t},\omega\right)  \right)  \Pi\left(  \omega|\mathfrak{x}_{t-1}\right)
=\mathbb{E}_{0}\left[  \varphi\circ\mathfrak{x}_{t+1}|\mathfrak{x}%
_{t},\mathfrak{x}_{t-1}\right]  \;\mathbb{P}_{0}-a.s.\ .\nonumber
\end{align}
In particular, by (\ref{EfiXt+1}), we get that $\left\{  \mathfrak{y}%
_{t}\right\}  _{t\geq0}$ such that $\forall t\geq0,\mathfrak{y}_{t}:=\left(
\mathfrak{x}_{2t-2},\mathfrak{x}_{2t-1}\right)  ,$ with $\mathfrak{x}%
_{-2}=\mathfrak{x}_{-1}=\mathfrak{x}_{0},$ is a homogeneous Markov process on
$\left(  \mathfrak{X}^{\mathbb{Z}_{+}},\mathfrak{C},\mathbb{P}_{0}\right)  .$
Indeed, denoting by $\left\{  \mathfrak{F}_{t}^{\mathfrak{y}}\right\}
_{t\geq0}$ the filtration generated by $\left\{  \mathfrak{y}_{t}\right\}
_{t\geq0},$ since $\forall t\geq0,\mathfrak{F}_{t}^{\mathfrak{y}}%
=\mathfrak{F}_{2t-1}^{X},$ for any bounded measurable function $\varphi$ on
$\Xi^{2},$%
\begin{align}
\mathbb{E}\left[  \varphi\circ\mathfrak{y}_{t+1}|\mathfrak{F}_{t}%
^{\mathfrak{y}}\right]   &  =\mathbb{E}\left[  \varphi\left(  \mathfrak{x}%
_{2t},\mathfrak{x}_{2t+1}\right)  |\mathfrak{F}_{2t-1}^{X}\right]
=\mathbb{E}\left[  \varphi\left(  \mathfrak{x}_{2t},\mathfrak{x}%
_{2t+1}\right)  |\mathfrak{x}_{2t-1},\mathfrak{x}_{2t-2}\right] \\
&  =\mathbb{E}\left[  \varphi\circ\mathfrak{y}_{t+1}|\mathfrak{y}_{t}\right]
\ .\nonumber
\end{align}
Therefore, the transition operator associated to $\left\{  \mathfrak{y}%
_{t}\right\}  _{t\geq0}$ is
\begin{equation}
\left(  \mathbf{T}\varphi\right)  \left(  X_{1},X_{2}\right)  :=\sum
_{\omega,\omega^{\prime}\in\Omega}\varphi\left(  \mathcal{T}\left(
X_{2},\omega\right)  ,\mathcal{T}\left(  \mathcal{T}\left(  X_{2}%
,\omega\right)  ,\omega^{\prime}\right)  \right)  \Pi\left(  \omega
|X_{1}\right)  \Pi\left(  \omega^{\prime}|X_{2}\right)  \label{TTT}%
\end{equation}
and, setting
\begin{equation}
\Xi^{2}\ni\left(  X_{1},X_{2}\right)  \longmapsto\pi_{i}\left(  X_{1}%
,X_{2}\right)  :=X_{1}\delta_{i,1}+X_{2}\delta_{i,2}\in\Xi\ ,\ i=1,2\ ,
\label{pi_i}%
\end{equation}
for any bounded measurable function $\varphi$ on $\Xi,$ we have,
$\mathbb{P}_{0}-a.s.,$%
\begin{align}
\mathbb{E}_{0}\left[  \varphi\circ\mathfrak{x}_{t+1}|\mathfrak{x}%
_{t},\mathfrak{x}_{t-1}\right]   &  =\mathbb{E}_{0}\left[  \left(
\varphi\circ\pi_{1}\right)  \left(  \mathfrak{y}_{s+1}\right)  |\mathfrak{y}%
_{s}\right]  \delta_{t,2s-1}+\mathbb{E}_{0}\left[  \left(  \varphi\circ\pi
_{2}\right)  \left(  \mathfrak{y}_{s+1}\right)  |\mathfrak{y}_{s}\right]
\delta_{t,2s}\\
&  =\mathbb{E}_{0}\left[  \mathbf{T}\left(  \varphi\circ\pi_{1}\right)
\left(  \mathfrak{y}_{s}\right)  \right]  \delta_{t,2s-1}+\mathbb{E}%
_{0}\left[  \mathbf{T}\left(  \varphi\circ\pi_{2}\right)  \left(
\mathfrak{y}_{s}\right)  \right]  \delta_{t,2s}\;,\;s\geq0\ .\nonumber
\end{align}

\end{remark}

\subsection{Consensus}

If $X^{0}\in\Xi$ is such that $\forall v\in\mathbf{V},X_{v}^{0}=x\in\left[
0,1\right]  ,$ then by (\ref{evX}) $X\left(  t\right)  =X^{0},\forall t\geq0.$
Hence these configuration, called \emph{consensus} configurations, are
stationary for the system evolution.

We denote by
\begin{equation}
\mathcal{I}:=\bigcup\limits_{x\in\left[  0,1\right]  }\mathcal{I}_{x}
\label{I}%
\end{equation}
where
\begin{equation}
\mathcal{I}_{x}:=\left\{  X\in\Xi:X_{v}=x\ ,\ \forall v\in\mathbf{V}\right\}
\label{Ix}%
\end{equation}
and by $\mathbb{M}:\Xi\longrightarrow\Xi$ the \emph{consensus projection map},
that is the map associating to each belief configuration $X$ the consensus
configuration $\mathbb{M}X$ such that $\forall v\in\mathbf{V},\left(
\mathbb{M}X\right)  _{v}:=\frac{1}{\left\vert \mathbf{V}\right\vert }%
\sum_{u\in\mathbf{V}}X_{u}.$ It is easy to see that $\mathbb{M}$ is a
projection operator on $\mathcal{I},$ moreover an orthogonal projection if
$\Xi$ is endowed with the Euclidean structure $\left\langle \cdot
,\cdot\right\rangle $ of $\mathbb{R}^{\mathbf{V}}.$ Indeed, $\forall X\in
\Xi,\mathbb{M}^{2}X=\mathbb{M}X.$ Therefore, $\forall X\in\Xi,$ we set
\begin{equation}
dist\left(  \mathcal{I},X\right)  :=\inf_{Y\in\mathcal{I}}\left\Vert
X-Y\right\Vert \leq\left[  \mathbb{I}-\mathbb{M}\right]  X\ , \label{dI}%
\end{equation}
where we denote by $\mathbb{I}$ the identity operator on $\mathbb{R}%
^{\mathbf{V}}.$

Consequently, since if $X=\mathbb{M}X,\forall\left(  u,v\right)  \in
\mathbf{E},X_{u}-X_{v}=0,$ we can modify the algorithm \ref{A1} erasing the
line 5\ and adding the line

\begin{enumerate}
\item[3.e] If $\sum_{i=1}^{N}\sum_{j=1}^{N}\left(  1-\delta_{i,j}\right)
p_{i,j}\left(  t\right)  = N\left(  N-1\right)  $ stop, otherwise proceed to
the next step.
\end{enumerate}

\subsection{Invariant measures for $\mathfrak{T}$ and $\mathbf{T}$}

Setting $\mathcal{X}:=\left(  \mathbb{I}-\mathbb{M}\right)  \Xi,$ we can
represent $\Xi$ as $\mathcal{I\oplus X}.$ Moreover, for any $\omega\in
\Omega,\mathcal{I}$ is invariant under $\mathcal{T}\left(  \cdot
,\omega\right)  ,$ since $X\in\Xi,$ by (\ref{T}) we get $\mathcal{T}\left(
\mathbb{M}X,\omega\right)  =\mathbb{M}X.$ Therefore
\begin{equation}
\mathcal{T}\left(  X,\omega\right)  =\mathcal{T}\left(  \mathbb{M}X+\left(
\mathbb{I-M}\right)  X,\omega\right)  =\mathbb{M}X+\mathcal{T}\left(  \left(
\mathbb{I-M}\right)  X,\omega\right)  \ . \label{decT}%
\end{equation}
Moreover, by (\ref{pomega1}), for any $\omega\in\Omega,\Pi\left(
\omega|X\right)  =\Pi\left(  \omega|\left(  \mathbb{I-M}\right)  X\right)  .$
Hence, denoting by $\delta_{\bar{1}}^{\omega}$ the Dirac measure at $\bar{1},$
by the definition of $p$ and by (\ref{pomega}), given $X\in\mathcal{I}%
,\forall\omega\in\Omega,\Pi\left(  \omega|X\right)  =\prod\limits_{e\in
\mathbf{E}}\delta_{\omega_{e},1}=\delta_{\bar{1}}^{\omega}.$

Denoting by $\delta^{X}$ the probability measure on $\left(  \Xi,B\left(
\Xi\right)  \right)  $ concentrated on the beliefs configuration $X\in\Xi,$
let $\delta_{\mathcal{I}}^{X}$ be the probability measure on $\left(
\Xi,\mathcal{B}\left(  \Xi\right)  \right)  $ putting mass $1$ on the
configuration $X\in\mathcal{I}.$ It is easy to see that the probability
measure $\delta_{\mathcal{I}}^{X}\otimes\delta_{\bar{1}}^{\omega}$ on $\left(
\mathfrak{X},\mathcal{F}\right)  $ is invariant for the evolution given by
$\mathfrak{T}.$ Indeed, if $X\in\mathcal{I},$ by (\ref{I}) and (\ref{Ix}),
there exists $x\in\left[  0,1\right]  $ such that $X\in\mathcal{I}_{x}.$
Hence, by (\ref{Tomega}), for any $\omega\in\Omega,\mathcal{T}\left(
X,\omega\right)  =X.$ Therefore, given any bounded measurable function
$\varphi$ on $\mathfrak{X},$ by (\ref{T}), $\forall\omega,\omega^{\prime}%
\in\Omega,\delta_{\mathcal{I}}^{X}\left[  \varphi\left(  \mathcal{T}\left(
\cdot,\omega\right)  ,\omega^{\prime}\right)  \right]  =\delta_{\mathcal{I}%
}^{X}\left[  \varphi\left(  \cdot,\omega^{\prime}\right)  \right]  .$ Thus, by
(\ref{TT}),
\begin{align}
\delta_{\mathcal{I}}^{X}\otimes\delta_{\bar{1}}^{\omega}\left[  \mathfrak{T}%
\varphi\right]   &  =\delta_{\bar{1}}^{\omega}\left[  \delta_{\mathcal{I}}%
^{X}\left[  \sum_{\omega^{\prime}\in\Omega}\varphi\left(  \mathcal{T}\left(
\cdot,\omega\right)  ,\omega^{\prime}\right)  \Pi\left(  \omega^{\prime}%
|\cdot\right)  \right]  \right] \\
&  =\delta_{\bar{1}}^{\omega}\left[  \sum_{\omega^{\prime}\in\Omega}%
\delta_{\mathcal{I}}^{X}\left[  \varphi\left(  \mathcal{T}\left(  \cdot
,\omega\right)  ,\omega^{\prime}\right)  \Pi\left(  \omega^{\prime}%
|\cdot\right)  \right]  \right] \nonumber\\
&  =\delta_{\bar{1}}^{\omega}\left[  \sum_{\omega^{\prime}\in\Omega}%
\delta_{\mathcal{I}}^{X}\left[  \varphi\left(  \cdot,\omega^{\prime}\right)
\prod\limits_{e\in\mathbf{E}}\delta_{\omega_{e}^{\prime},1}\right]  \right]
\nonumber\\
&  =\delta_{\mathcal{I}}^{X}\otimes\delta_{\bar{1}}^{\omega}\left[
\sum_{\omega^{\prime}\in\Omega}\varphi\left(  \cdot,\omega^{\prime}\right)
\prod\limits_{e\in\mathbf{E}}\delta_{\omega_{e}^{\prime},1}\right] \nonumber\\
&  =\delta_{\mathcal{I}}^{X}\otimes\delta_{\bar{1}}^{\omega}\left[
\varphi\left(  \cdot,\bar{1}\right)  \right]  =\delta_{\mathcal{I}}^{X}%
\otimes\delta_{\bar{1}}^{\omega}\left[  \varphi\right]  \ .\nonumber
\end{align}
Thus the set $\mathfrak{I}_{\mathfrak{T}}$ of invariant probability measures
under $\mathfrak{T}$ is the weak limit of convex combinations of elements of
the set $\left\{  \delta_{\bar{1}}^{\omega}\otimes\delta_{\mathcal{I}}%
^{X}\right\}  _{X\in\mathcal{I}}\subset\mathfrak{P}\left(  \mathfrak{X}%
,\mathcal{F}\right)  .$

Since for any $X_{2}\in\mathcal{I}$ and any $\mu\in\mathfrak{P}\left(
\Xi,\mathcal{B}\left(  \Xi\right)  \right)  ,$ from (\ref{TTT}) it follows
that
\begin{align}
\mu\otimes\delta_{\mathcal{I}}^{X_{2}}\left[  \mathbf{T}\varphi\right]   &
=\int\mu\left(  dX_{1}\right)  \delta_{\mathcal{I}}^{X_{2}}\left[
\sum_{\omega,\omega^{\prime}\in\Omega}\varphi\left(  \mathcal{T}\left(
\cdot,\omega\right)  ,\mathcal{T}\left(  \mathcal{T}\left(  \cdot
,\omega\right)  ,\omega^{\prime}\right)  \right)  \Pi\left(  \omega
|X_{1}\right)  \Pi\left(  \omega^{\prime}|\cdot\right)  \right] \\
&  =\int\mu\left(  dX_{1}\right)  \sum_{\omega,\omega^{\prime}\in\Omega
}\varphi\left(  X_{2},,\mathcal{T}\left(  X_{2},\omega^{\prime}\right)
\right)  \Pi\left(  \omega|X_{1}\right)  \prod\limits_{e\in\mathbf{E}}%
\delta_{\omega_{e}^{\prime},1}\nonumber\\
&  =\int\mu\left(  dX_{1}\right)  \sum_{\omega\in\Omega}\varphi\left(
X_{2},,\mathcal{T}\left(  X_{2},\bar{1}\right)  \right)  \Pi\left(
\omega|X_{1}\right) \nonumber\\
&  =\int\mu\left(  dX_{1}\right)  \sum_{\omega\in\Omega}\varphi\left(
X_{2},X_{2}\right)  \Pi\left(  \omega|X_{1}\right)  =\varphi\left(
X_{2},X_{2}\right)  \ ,\nonumber
\end{align}
we have that the set $\mathfrak{I}_{\mathbf{T}}$ of invariant probability
measures under $\mathbf{T}$ is the weak limit of convex combinations of
elements of the set $\left\{  \delta_{\mathcal{I}}^{\left(  X,X\right)
}\right\}  _{X\in\mathcal{I}}\subset\mathfrak{P}\left(  \Xi^{2},\mathcal{B}%
\left(  \Xi^{2}\right)  \right)  ,$ where $\delta_{\mathcal{I}}^{\left(
X,X\right)  }:=\delta_{\mathcal{I}}^{X}\otimes\delta_{\mathcal{I}}^{X}.$

\subsection{Emergence of consensus}

Given $X\in\Xi,$ let
\begin{equation}
W\left(  X\right)  :=\max_{u,v\in\mathbf{V}}\left\vert X_{u}-X_{v}\right\vert
\ . \label{W}%
\end{equation}
Since
\begin{equation}
W\left(  \left[  \mathbb{I}-\mathbb{M}\right]  X\right)  =W\left(  X\right)
\ ,
\end{equation}
$W$ is a seminorm on $\mathbb{R}^{\mathbf{V}}$ and therefore induces a norm on
$\mathbb{W}:=\mathbb{R}^{\mathbf{V}}/Ran\mathbb{M}.$

Hence, because $\mathbb{M}\mathcal{I}=\mathcal{I},$ for any $Y\in\mathcal{I},$
we have
\begin{equation}
\left\Vert X-Y\right\Vert =W\left(  X-Y\right)  =W\left(  \left[
\mathbb{I}-\mathbb{M}\right]  \left(  X-Y\right)  \right)  =W\left(  X\right)
\ ,
\end{equation}
which implies
\begin{equation}
dist\left(  \mathcal{I},X\right)  =W\left(  X\right)  \ .
\end{equation}

For any $t\geq0,$ let $W\left(  t\right)  :=W\left(  X\left(  t\right)
\right)  .$ In the following we will prove that the random sequence $\left\{
W\left(  t\right)  \right\}  _{t\geq0}$ converges to zero w.p.1 w.r.t. the
noise, hence proving Theorem \ref{main}.

\begin{definition}
\label{D1}Given $E\subseteq\mathbf{E},$ consider the spanning graph $G\left(
E\right)  =\left(  \mathbf{V},E\right)  .$ We call \emph{pivots} the elements
$w$ of $\mathbf{V}$ such that $\mathcal{N}^{+}\left(  w\right)  =\mathbf{V}$
and denote their collection by $\mathbf{P}\left(  E\right)  .$ Moreover, for
any $\omega\in\Omega,$ we set $\mathbf{P}\left(  \omega\right)  :=\mathbf{P}%
\left(  E\left(  \omega\right)  \right)  $ and define $\Omega_{\mathbf{P}%
}:=\left\{  \omega\in\Omega:\mathbf{P}\left(  \omega\right)  \neq
\varnothing\right\}  .$
\end{definition}

Let us denote by $\gamma$ the r.v.\footnote{Notice that $1-\gamma$ is the
coefficient of ergodicity \cite{Se} of the transition probability matrix of
the Markov chain on $\mathbf{V}^{2}$ whose components are two independent
versions of the Markov chain defined by the transition probability matrix
$\left\{  P_{u,v}\left(  \omega\right)  \right\}  _{u,v\in\mathbf{V}}.$}
\begin{equation}
\Omega\ni\omega\longmapsto\gamma\left(  \omega\right)  :=\min_{u,v\in
\mathbf{V}\ :\ u\neq v}\sum_{w,z\in\mathbf{V}}P_{u,w}\left(  \omega\right)
P_{v,z}\left(  \omega\right)  \wedge P_{u,z}\left(  \omega\right)
P_{v,w}\left(  \omega\right)  \in\lbrack0,1) \label{gamma}%
\end{equation}
and by $\Gamma$ the r.v.
\begin{equation}
\Xi\ni X\longmapsto\Gamma\left(  X\right)  :=\mathbb{E}\left[  \gamma
|X\right]  =\sum_{\omega\in\Omega}\Pi\left(  \omega|X\right)  \gamma\left(
\omega\right)  \in\lbrack0,1)\ . \label{Gamma}%
\end{equation}

\begin{lemma}
Given $\omega\in\Omega,\gamma\left(  \omega\right)  >0$ if and only if
$\mathbf{P}\left(  \omega\right)  $ is not empty.
\end{lemma}

\begin{proof}
Let $\omega\in\Omega$ be such that $\mathbf{P}\left(  \omega\right)
\neq\varnothing.$ Denoting by $u=u\left(  \omega\right)  ,v=v\left(
\omega\right)  $ the elements of $\mathbf{V}$ such that
\begin{gather}
\sum_{w,z\in\mathbf{V}}P_{u,w}\left(  \omega\right)  P_{v,z}\left(
\omega\right)  \wedge P_{u,z}\left(  \omega\right)  P_{v,w}\left(
\omega\right)  =\\
\min_{u^{\prime},u^{\prime\prime}\in\mathbf{V}}\sum_{w,z\in\mathbf{V}%
}P_{u^{\prime},w}\left(  \omega\right)  P_{u^{\prime\prime},z}\left(
\omega\right)  \wedge P_{u^{\prime},z}\left(  \omega\right)  P_{u^{\prime
\prime},w}\left(  \omega\right)  \ ,\nonumber
\end{gather}
\ for any $\bar{w}\in\mathbf{P}\left(  \omega\right)  ,$ we have
\begin{gather}
\gamma\left(  \omega\right)  =\sum_{w,z\in\mathbf{V}}P_{u,w}\left(
\omega\right)  P_{v,z}\left(  \omega\right)  \wedge P_{u,z}\left(
\omega\right)  P_{v,w}\left(  \omega\right)  =\sum_{w\in\mathbf{V}}%
P_{u,w}\left(  \omega\right)  P_{v,w}\left(  \omega\right)  +\\
\sum_{w,z\in\mathbf{V}\ :\ w\neq z}P_{u,w}\left(  \omega\right)
P_{v,z}\left(  \omega\right)  \wedge P_{u,z}\left(  \omega\right)
P_{v,w}\left(  \omega\right)  \geq\left(  P_{u,\bar{w}}\left(  \omega\right)
\wedge P_{v,\bar{w}}\left(  \omega\right)  \right)  ^{2}>0\ .\nonumber
\end{gather}
Conversely by (\ref{Pomega}), $\gamma\left(  \omega\right)  >0$ iff,\ for any
$u,v\in\mathbf{V}$ such that $u\neq v,\mathcal{N}^{-}\left(  u,\omega\right)
\cap\mathcal{N}^{-}\left(  v,\omega\right)  \neq\varnothing,$ which is
equivalent to say that $\gamma\left(  \omega\right)  >0$ implies that there
exists at least one $\bar{w}=\bar{w}\left(  \omega\right)  $ in $\mathbf{V}$
such that, by (\ref{Puv(t)}), $\mathcal{N}^{+}\left(  \bar{w},\omega\right)
=\mathbf{V},$ or, in other words, by Definition \ref{D1}, that $\mathbf{P}%
\left(  \omega\right)  $ is not empty.
\end{proof}

\begin{proposition}
\label{Prop1}The sequence $\left\{  W\left(  t\right)  \right\}  _{t\geq0}$ is
non-increasing hence bounded. Moreover, $\left\{  W\left(  t\right)  \right\}
_{t\geq0}$ is a non-negative $L^{1}$-supermartingale w.r.t. $\left\{
\mathfrak{F}_{t}^{X}\right\}  _{t\geq0},$ therefore $\mathbb{P}_{0}$-a.s.
convergent to a $L^{1}\left(  \mathfrak{X},\mathcal{F},\mathbb{P}_{0}\right)
$ r.v. which we denote by $W.$
\end{proposition}

\begin{proof}
By (\ref{evX}), given $u,v\in\mathbf{V}$ such that $u\neq v,$ for $t\geq0,$%
\begin{align}
X_{u}\left(  t+1\right)  -X_{v}\left(  t+1\right)   &  =\left(  X_{u}\left(
t+1\right)  -X_{u}\left(  t\right)  \right)  -\left(  X_{v}\left(  t+1\right)
-X_{v}\left(  t\right)  \right)  +X_{u}\left(  t\right)  -X_{v}\left(
t\right) \\
&  =X_{u}\left(  t\right)  -X_{v}\left(  t\right)  +\frac{\sum_{e\in E_{u}%
^{-}\left(  t\right)  }r_{e}\mathbf{1}_{\left(  u,w\right)  }\left(  e\right)
\left[  X_{w}\left(  t\right)  -X_{u}\left(  t\right)  \right]  }{\sum_{e\in
E_{u}^{-}\left(  t\right)  }r_{e}}\nonumber\\
&  -\frac{\sum_{e^{\prime}\in E_{v}^{-}\left(  t\right)  }r_{e^{\prime}%
}\mathbf{1}_{\left(  v,z\right)  }\left(  e^{\prime}\right)  \left[
X_{z}\left(  t\right)  -X_{v}\left(  t\right)  \right]  }{\sum_{e^{\prime}\in
E_{v}^{-}\left(  t\right)  }r_{e^{\prime}}}\nonumber\\
&  =\frac{\sum_{e\in E_{u}^{-}\left(  t\right)  }r_{e}\mathbf{1}_{\left(
u,w\right)  }\left(  e\right)  }{\sum_{e\in E_{u}^{-}\left(  t\right)  }r_{e}%
}X_{w}\left(  t\right)  -\frac{\sum_{e^{\prime}\in E_{v}^{-}\left(  t\right)
}r_{e^{\prime}}\mathbf{1}_{\left(  v,z\right)  }\left(  e^{\prime}\right)
}{\sum_{e^{\prime}\in E_{v}^{-}\left(  t\right)  }r_{e^{\prime}}}X_{z}\left(
t\right)  \ .\nonumber
\end{align}
By (\ref{Pomega}), setting
\begin{equation}
P_{u,v}\left(  t\right)  :=P_{u,v}\left(  \omega\left(  t\right)  \right)
=\frac{\sum_{e\in E_{v}^{-}\left(  t\right)  }\delta_{e,\left(  v,u\right)
}r_{e}}{\sum_{e\in E_{u}^{-}\left(  t\right)  }r_{e}}=\frac{r_{v,u}%
\mathbf{1}_{\mathcal{N}^{-}\left(  u,t\right)  }\left(  v\right)  }{\sum
_{v\in\mathcal{N}^{-}\left(  u,t\right)  }r_{v,u}} \label{Puv(t)}%
\end{equation}
we can rewrite the previous expression as
\begin{equation}
X_{u}\left(  t+1\right)  -X_{v}\left(  t+1\right)  =\sum_{w\in\mathbf{V}%
}P_{u,w}\left(  t\right)  X_{w}\left(  t\right)  -\sum_{z\in\mathbf{V}}%
P_{v,z}\left(  t\right)  X_{z}\left(  t\right)  \ .
\end{equation}
Since, $\forall t\geq0,\sum_{v\in\mathbf{V}}P_{u,v}\left(  t\right)  =1,$ we
have
\begin{equation}
X_{u}\left(  t+1\right)  -X_{v}\left(  t+1\right)  =\sum_{w,z\in\mathbf{V}%
}P_{u,w}\left(  t\right)  P_{v,z}\left(  t\right)  \left[  X_{w}\left(
t\right)  -X_{z}\left(  t\right)  \right]
\end{equation}
and, since $\left[  X_{w}\left(  t\right)  -X_{z}\left(  t\right)  \right]
=-\left[  X_{z}\left(  t\right)  -X_{w}\left(  t\right)  \right]  \ ,$we
obtain
\begin{equation}
X_{u}\left(  t+1\right)  -X_{v}\left(  t+1\right)  =\frac{1}{2}\sum
_{w,z\in\mathbf{V}}\left\{  P_{u,w}\left(  t\right)  P_{v,z}\left(  t\right)
-P_{u,z}\left(  t\right)  P_{v,w}\left(  t\right)  \right\}  \left[
X_{w}\left(  t\right)  -X_{z}\left(  t\right)  \right]  \ .
\end{equation}

Hence
\begin{equation}
\left\vert X_{u}\left(  t+1\right)  -X_{v}\left(  t+1\right)  \right\vert
\leq\frac{1}{2}\sum_{w,z\in\mathbf{V}}\left\vert P_{u,w}\left(  t\right)
P_{v,z}\left(  t\right)  -P_{u,z}\left(  t\right)  P_{v,w}\left(  t\right)
\right\vert \left\vert X_{w}\left(  t\right)  -X_{z}\left(  t\right)
\right\vert \ .
\end{equation}
Since $\forall a,b\in\mathbb{R},a\wedge b=\frac{a+b-\left\vert a-b\right\vert
}{2},$%
\begin{align}
\left\vert X_{u}\left(  t+1\right)  -X_{v}\left(  t+1\right)  \right\vert  &
\leq\sum_{w,z\in\mathbf{V}}\left\{  \frac{P_{u,w}\left(  t\right)
P_{v,z}\left(  t\right)  +P_{u,z}\left(  t\right)  P_{v,w}\left(  t\right)
}{2}\right. \\
&  \left.  -P_{u,w}\left(  t\right)  P_{v,z}\left(  t\right)  \wedge
P_{u,z}\left(  t\right)  P_{v,w}\left(  t\right)  \right\}  \left\vert
X_{w}\left(  t\right)  -X_{z}\left(  t\right)  \right\vert \nonumber\\
&  \leq\sum_{w,z\in\mathbf{V}}\left\{  \frac{P_{u,w}\left(  t\right)
P_{v,z}\left(  t\right)  +P_{u,z}\left(  t\right)  P_{v,w}\left(  t\right)
}{2}\right. \nonumber\\
&  \left.  -P_{u,w}\left(  t\right)  P_{v,z}\left(  t\right)  \wedge
P_{u,z}\left(  t\right)  P_{v,w}\left(  t\right)  \right\}  \max
_{w,z\in\mathbf{V}}\left\vert X_{w}\left(  t\right)  -X_{z}\left(  t\right)
\right\vert \nonumber\\
&  \leq\left\{  1-\sum_{w,z\in\mathbf{V}}P_{u,w}\left(  t\right)
P_{v,z}\left(  t\right)  \wedge P_{u,z}\left(  t\right)  P_{v,w}\left(
t\right)  \right\}  \max_{w,z\in\mathbf{V}}\left\vert X_{w}\left(  t\right)
-X_{z}\left(  t\right)  \right\vert \ .\nonumber
\end{align}
Therefore, choosing $u,v\in\mathbf{V}$ such that
\begin{equation}
\left\vert X_{u}\left(  t+1\right)  -X_{v}\left(  t+1\right)  \right\vert
=\max_{w,z\in\mathbf{V}}\left\vert X_{w}\left(  t+1\right)  -X_{z}\left(
t+1\right)  \right\vert \ ,
\end{equation}
by (\ref{gamma}) we get
\begin{align}
W\left(  t+1\right)   &  \leq\left\{  1-\min_{u,v\in\mathbf{V}\ :\ u\neq
v}\sum_{w,z\in\mathbf{V}}P_{u,w}\left(  t\right)  P_{v,z}\left(  t\right)
\wedge P_{u,z}\left(  t\right)  P_{v,w}\left(  t\right)  \right\}  W\left(
t\right) \label{decreasing}\\
&  =\left(  1-\gamma\left(  \omega\left(  t\right)  \right)  \right)  W\left(
t\right)  \ ;\nonumber
\end{align}
hence, $\forall t\geq0,W\left(  t\right)  \leq W\left(  0\right)  \leq1.$

Thus, representing the random sequence $\left\{  W\left(  t\right)  \right\}
_{t\geq0}$ as $\left\{  W\circ\mathfrak{x}_{t}\right\}  _{t\geq0},$ from
(\ref{decreasing}) we get
\begin{gather}
\mathbb{E}_{0}\left[  W\left(  t+1\right)  |\mathfrak{F}_{t}^{X}\right]
=\mathbb{E}_{0}\left[  W\circ\mathfrak{x}_{t+1}|\mathfrak{F}_{t}^{X}\right]
=\mathbb{E}_{0}\left[  \mathbb{E}_{0}\left[  W\circ\pi_{X}\circ\chi
_{t+1}|\mathfrak{F}_{t}\right]  |\mathfrak{F}_{t}^{X}\right]  \label{superm}\\
=\mathbb{E}_{0}\left[  \mathbb{E}_{0}\left[  W\circ\pi_{X}\circ\chi_{t+1}%
|\chi_{t}\right]  |\mathfrak{F}_{t}^{X}\right]  \leq\mathbb{E}_{0}\left[
\left\{  1-\gamma\circ\pi_{\omega}\circ\chi_{t}\right\}  W\circ\pi_{X}%
\circ\chi_{t}|\mathfrak{F}_{t}^{X}\right]  \nonumber\\
=\mathbb{E}_{0}\left[  \left\{  1-\gamma\left(  \pi_{\omega}\circ\chi
_{t}\right)  \right\}  W\circ\mathfrak{x}_{t}|\mathfrak{F}_{t}^{X}\right]
\nonumber\\
=\left\{  1-\mathbb{E}_{0}\left[  \gamma\left(  \pi_{\omega}\circ\chi
_{t}\right)  |\mathfrak{F}_{t}^{X}\right]  \right\}  W\left(  t\right)  \leq
W\left(  t\right)  \ ,\nonumber
\end{gather}
that is $\left\{  W\left(  t\right)  \right\}  _{t\geq0}$ is a $L^{1}%
$-supermartingale w.r.t. $\left\{  \mathfrak{F}_{t}^{X}\right\}  _{t\geq0}.$
\end{proof}

\subsubsection{Asymptotic estimate of $\mathbb{E}_{0}\left[  W\left(
t\right)  \right]  $}

\begin{lemma}
\label{pred}The sequence $\left\{  \mathbb{E}_{0}\left[  \gamma|\mathfrak{F}%
_{t}^{X}\right]  \right\}  _{t\geq0}$ is predictable w.r.t. the filtration
$\left\{  \mathfrak{F}_{t}^{X}\right\}  _{t\geq0}.$
\end{lemma}

\begin{proof}
For any $t\geq1,$ by (\ref{Pomega}),%
\begin{gather}
\mathbb{E}_{0}\left[  \gamma|\mathfrak{F}_{t}^{X}\right]  =\mathbb{E}%
_{0}\left[  \left[  \gamma\left(  \pi_{\omega}\circ\chi_{t}\right)
\right\vert \mathfrak{F}_{t}^{X}\right]  =\mathbb{E}_{0}\left[  \left[
\gamma\left(  \mathfrak{w}_{t}\right)  \right\vert \mathfrak{x}_{t-1}\right]
=\\
\sum_{\omega\in\Omega}\Pi\left(  \omega|\mathfrak{x}_{t-1}\right)
\min_{u,v\in\mathbf{V}\ :\ u\neq v}\sum_{w,z\in\mathbf{V}}P_{u,w}\left(
\omega\right)  P_{v,z}\left(  \omega\right)  \wedge P_{u,z}\left(
\omega\right)  P_{v,w}\left(  \omega\right)  =\Gamma\left(  \mathfrak{x}%
_{t-1}\right)  \ .\nonumber
\end{gather}

\end{proof}

Let us set
\begin{align}
\mathfrak{X}\ni\left(  X,\omega\right)  \longmapsto\Pi\left(  \omega|W\left(
X\right)  \right)   &  :=\prod\limits_{e\in\mathbf{E}}\left[  \delta
_{\omega_{e},1}p\left(  W\left(  X\right)  \right)  +\delta_{\omega_{e}%
,0}\left(  1-p\left(  W\left(  X\right)  \right)  \right)  \right]
\label{pomega2}\\
&  =\prod\limits_{e\in\mathbf{E}}\left[  \omega_{e}p\left(  W\left(  X\right)
\right)  +\left(  1-\omega_{e}\right)  \left(  1-p\left(  W\left(  X\right)
\right)  \right)  \right]  \in\left[  0,1\right]  \ .\nonumber
\end{align}

Since $\Omega$ is a poset w.r.t. the partial order relation: $\omega\leq
\omega^{\prime}$ if, $\forall e\in\mathbf{E},\omega_{e}\leq\omega_{e}^{\prime
},$ we have

\begin{lemma}
\label{L1}For any $X\in\Xi,\Pi\left(  \cdot|X\right)  \overset{st}{\geq}%
\Pi\left(  \cdot|W\left(  X\right)  \right)  .$ Moreover, for any $t\geq
0,\Pi\left(  \cdot|W\left(  t+1\right)  \right)  \overset{st}{\geq}\Pi\left(
\cdot|W\left(  t\right)  \right)  .$
\end{lemma}

\begin{proof}
Let us consider first the statement $\Pi\left(  \cdot|X\right)  \overset
{st}{\geq}\Pi\left(  \cdot|W\left(  X\right)  \right)  .$ For $X\in
\mathcal{I},$ by (\ref{pomega1}) and (\ref{pomega2}) $\Pi\left(
\cdot|X\right)  $ and $\Pi\left(  \cdot|W\left(  X\right)  \right)  $
coincide. Let now $X\in\mathcal{X}.$ By (\ref{pomega1}) and (\ref{pomega2})
$\Pi\left(  \cdot|\cdot\right)  $ is irreducible, then to prove $\Pi\left(
\cdot|X\right)  \overset{st}{\geq}\Pi\left(  \cdot|W\left(  X\right)  \right)
$ is enough to prove that the Holley inequality is satisfied, namely
\begin{equation}
\Pi\left(  \omega\vee\omega^{\prime}|X\right)  \Pi\left(  \omega\wedge
\omega^{\prime}|W\left(  X\right)  \right)  \geq\Pi\left(  \omega|X\right)
\Pi\left(  \omega^{\prime}|W\left(  X\right)  \right)  \;,\quad\omega
,\omega^{\prime}\in\Omega\ ,
\end{equation}
where $\omega\vee\omega^{\prime}\in\Omega$ is such that $\forall
e\in\mathbf{E},\left(  \omega\vee\omega^{\prime}\right)  _{e}=\omega_{e}%
\vee\omega_{e}^{\prime}$ and $\omega\wedge\omega^{\prime}\in\Omega$ is such
that $\forall e\in\mathbf{E},\left(  \omega\wedge\omega^{\prime}\right)
_{e}=\omega_{e}\wedge\omega_{e}^{\prime}.$ This is equivalent to prove that,
for any $e,f\in\mathbf{E},$%
\begin{equation}
\Pi\left(  \omega^{\left\{  e\right\}  }|X\right)  \Pi\left(  \omega_{\left\{
e\right\}  }|W\left(  X\right)  \right)  \geq\Pi\left(  \omega^{\left\{
e\right\}  }|W\left(  X\right)  \right)  \Pi\left(  \omega_{\left\{
e\right\}  }|X\right)  \label{H1}%
\end{equation}
and
\begin{equation}
\Pi\left(  \omega^{\left\{  ef\right\}  }|X\right)  \Pi\left(  \omega
_{\left\{  ef\right\}  }|W\left(  X\right)  \right)  \geq\Pi\left(
\omega_{\left\{  f\right\}  }^{\left\{  e\right\}  }|W\left(  X\right)
\right)  \Pi\left(  \omega_{\left\{  e\right\}  }^{\left\{  f\right\}
}|X\right)  \ , \label{H2}%
\end{equation}
where, for any $E\subset\mathbf{E},\omega^{E}\in\Omega$ is such that $\forall
e\in\mathbf{E},\omega_{e}^{E}:=\omega_{e}\mathbf{1}_{E^{c}}\left(  e\right)
+\mathbf{1}_{E}\left(  e\right)  $ and $\omega_{E}\in\Omega$ is such that
$\forall e\in\mathbf{E},\left(  \omega_{E}\right)  _{e}:=\omega_{e}%
\mathbf{1}_{E^{c}}\left(  e\right)  $ (see e.g. \cite{Gr} Theorem 2.3). But,
by (\ref{pomega1}) and (\ref{pomega2}), $\Pi\left(  \cdot|X\right)  $ and
$\Pi\left(  \cdot|W\left(  X\right)  \right)  $ are product measures, then
(\ref{H1}) becomes
\[
\Pi\left(  \omega_{e}=1|X\right)  \Pi\left(  \omega_{e}=0|W\left(  X\right)
\right)  \geq\Pi\left(  \omega_{e}=1|W\left(  X\right)  \right)  \Pi\left(
\omega_{e}=0|X\right)  \ ,
\]
which can be rewritten as
\begin{equation}
p\left(  \left\vert \Delta_{e}X\right\vert \right)  \left(  1-p\left(
W\left(  X\right)  \right)  \right)  \geq p\left(  W\left(  X\right)  \right)
\left(  1-p\left(  \left\vert \Delta_{e}X\right\vert \right)  \right)
\label{H1'}%
\end{equation}
and (\ref{H2}) becomes
\begin{gather}
\Pi\left(  \omega_{e}=1|X\right)  \Pi\left(  \omega_{f}=1|X\right)  \Pi\left(
\omega_{e}=0|W\left(  X\right)  \right)  \Pi\left(  \omega_{f}=0|W\left(
X\right)  \right)  \geq\\
\Pi\left(  \omega_{e}=1|W\left(  X\right)  \right)  \Pi\left(  \omega
_{f}=0|W\left(  X\right)  \right)  \Pi\left(  \omega_{e}=0|X\right)
\Pi\left(  \omega_{f}=1|X\right) \nonumber
\end{gather}
which is again (\ref{H1'}). Since by (\ref{W}), for any $e\in\mathbf{E,}%
W\left(  X\right)  \geq\Delta_{e}X$ and since $p:\left[  0,1\right]
\circlearrowleft$ is non increasing, we have, for any $e\in\mathbf{E},p\left(
\left\vert \Delta_{e}X\right\vert \right)  \geq p\left(  W\left(  X\right)
\right)  $ and consequently $\left(  1-p\left(  W\left(  X\right)  \right)
\right)  \geq\left(  1-p\left(  \left\vert \Delta_{e}X\right\vert \right)
\right)  $ which proves (\ref{H1'}).

The proof of the statement $\Pi\left(  \cdot|W\left(  t+1\right)  \right)
\overset{st}{\geq}\Pi\left(  \cdot|W\left(  t\right)  \right)  ,t\geq0,$
follow the same lines of the proof of $\Pi\left(  \cdot|X\right)  \overset
{st}{\geq}\Pi\left(  \cdot|W\left(  X\right)  \right)  $ since, by
(\ref{decreasing}), $W\left(  t+1\right)  \leq W\left(  t\right)  ,$ which
implies $p\left(  W\left(  t+1\right)  \right)  \geq p\left(  W\left(
t\right)  \right)  .$
\end{proof}

\paragraph{Proof of Theorem \ref{main}}

More precisely we prove the following result.

\begin{theorem}
\label{mainp}For any $X^{0}\in\left\{  X\in\Xi:\Gamma\left(  W\left(
X\right)  \right)  >0\right\}  $ and any $\omega_{0}\in\Omega,$ the sequence
of probability measures $\left\{  \mu_{0}^{t}\right\}  _{t\geq0}$ on $\left(
\Xi,\mathcal{B}\left(  \Xi\right)  \right)  $ such that
\begin{equation}
\mathcal{B}\left(  \Xi\right)  \ni A\longmapsto\mu_{0}^{t}\left(  A\right)
:=\mathbb{P}_{0}\left\{  \mathbf{x}\in\mathfrak{X}^{\mathbb{Z}_{+}}:\pi
_{X}\circ\chi_{t}\left(  \mathbf{x}\right)  \in A\right\}  \in\left[
0,1\right]  \ ,
\end{equation}
converges to a probability measure $\mu_{0}^{\infty}$ supported on
$\mathcal{I}.$
\end{theorem}

\begin{proof}
Since $\gamma$ is an non-decreasing function, by (\ref{Gamma}) and by the
previous lemma we have that
\begin{equation}
\Gamma\left(  X\right)  \geq\sum_{\omega\in\Omega}\Pi\left(  \omega|W\left(
X\right)  \right)  \gamma\left(  \omega\right)  =:\Gamma\left(  W\left(
X\right)  \right)  \label{GammaW}%
\end{equation}
and, for any $t\geq0,\Gamma\left(  W\left(  t+1\right)  \right)  \geq
\Gamma\left(  W\left(  t\right)  \right)  .$ Then, by (\ref{superm}) and Lemma
\ref{pred} we get
\begin{align}
\mathbb{E}_{0}\left[  W\left(  t\right)  \right]   &  =\mathbb{E}_{0}\left[
\mathbb{E}_{0}\left[  W\circ\pi_{X}\circ\chi_{t}|\mathfrak{F}_{t-1}%
^{X}\right]  \right]  \leq\mathbb{E}_{0}\left[  \left(  1-\mathbb{E}%
_{0}\left[  \gamma|\mathfrak{F}_{t-1}^{X}\right]  \right)  W\circ\pi_{X}%
\circ\chi_{t-1}\right]  \label{mainest}\\
&  =\mathbb{E}_{0}\left[  \left(  1-\Gamma\left(  X\left(  t-2\right)
\right)  \right)  W\left(  t-1\right)  \right]  \leq\mathbb{E}_{0}\left[
\left(  1-\Gamma\left(  W\left(  t-2\right)  \right)  \right)  W\left(
t-1\right)  \right]  \nonumber\\
&  =\mathbb{E}_{0}\left[  \left(  1-\Gamma\left(  W\circ\pi_{X}\circ\chi
_{t-2}\right)  \right)  \mathbb{E}_{0}\left[  W\circ\pi_{X}\circ\chi
_{t-1}|\mathfrak{F}_{t-2}^{X}\right]  \right]  \nonumber\\
&  \leq\mathbb{E}_{0}\left[  \left(  1-\Gamma\left(  W\circ\pi_{X}\circ
\chi_{t-2}\right)  \right)  \left(  1-\mathbb{E}_{0}\left[  \gamma
|\mathfrak{F}_{t-2}^{X}\right]  \right)  W\circ\pi_{X}\circ\chi_{t-2}\right]
\nonumber\\
&  =\mathbb{E}_{0}\left[  \left(  1-\Gamma\left(  W\left(  t-2\right)
\right)  \right)  \left(  1-\Gamma\left(  X\left(  t-2\right)  \right)
\right)  W\left(  t-2\right)  \right]  \nonumber\\
&  \leq\mathbb{E}_{0}\left[  \left(  1-\Gamma\left(  W\left(  t-2\right)
\right)  \right)  ^{2}W\left(  t-2\right)  \right]  \leq\mathbb{E}_{0}\left[
\left(  1-\Gamma\left(  W\left(  t-3\right)  \right)  \right)  ^{2}W\left(
t-2\right)  \right]  \ .\nonumber
\end{align}
Iterating this inequality, after $k$ steps, with $k\leq t,$ we obtain
\begin{equation}
\mathbb{E}_{0}\left[  W\left(  t\right)  \right]  \leq\mathbb{E}_{0}\left[
\left(  1-\Gamma\left(  W\left(  t-k\right)  \right)  \right)  ^{k-1}W\left(
t-k+1\right)  \right]
\end{equation}
which, by (\ref{decreasing}) implies
\begin{equation}
\mathbb{E}_{0}\left[  W\left(  t\right)  \right]  \leq\mathbb{E}_{0}\left[
W\left(  1\right)  \right]  \left(  1-\Gamma\left(  W\left(  X^{0}\right)
\right)  \right)  ^{t-1}\leq W\left(  X^{0}\right)  \left(  1-\Gamma\left(
W\left(  X^{0}\right)  \right)  \right)  ^{t}\ .\label{expconv}%
\end{equation}
Therefore, for any $X^{0}\in\left\{  X\in\Xi:\Gamma\left(  W\left(  X\right)
\right)  >0\right\}  ,$ since $W\left(  X^{0}\right)  $ and for any
$\varepsilon>0$ the Markov inequality implies
\begin{equation}
\mathbb{P}_{0}\left\{  W\left(  t\right)  >\varepsilon\right\}  \leq
\frac{\mathbb{E}_{0}\left[  W\left(  t\right)  \right]  }{\varepsilon}%
\leq\varepsilon^{-1}\left(  1-\Gamma\left(  W\left(  X^{0}\right)  \right)
\right)  ^{t}\ ,
\end{equation}
by the Borel-Cantelli Lemma $\left\{  W\left(  t\right)  \right\}  _{t\geq0}$
converges to zero $\mathbb{P}_{0}$-a.s., that is\linebreak$\mu_{0}^{\infty
}:=\lim_{t\rightarrow\infty}\mathbb{P}_{0}\left\{  X\left(  t\right)  \in
\cdot\right\}  $ is supported on $\mathcal{I}.$
\end{proof}

\begin{remark}
We stress that this result give no information on the common value of the
beliefs when consensus is reached.
\end{remark}

\subsection{Convergence to the stationary measure of $\left\{  \chi
_{t}\right\}  _{t\in\mathbb{Z}_{+}}$ and $\left\{  \mathfrak{y}_{t}\right\}
_{t\in\mathbb{Z}_{+}}$}

Given $X\in\Xi,$ let us set $X=\left(  U,V\right)  $ such that $U:=\mathbb{M}%
X,V:=\left(  \mathbb{I}-\mathbb{M}\right)  X$ and consider the random
processes $\left\{  \mathfrak{u}_{t}\right\}  _{t\in\mathbb{Z}_{+}}$ and
$\left\{  \mathfrak{v}_{t}\right\}  _{t\in\mathbb{Z}_{+}}$ such that $\forall
t\geq0,\mathfrak{u}_{t}:=\mathbb{M}\mathfrak{x}_{t}$ and $\mathfrak{v}%
_{t}:=\left(  \mathbb{I}-\mathbb{M}\right)  \mathfrak{x}_{t}.$ From
(\ref{pomega1}),(\ref{TT}) and (\ref{decT}), for any bounded measurable
function $\varphi$ on $\mathcal{I}\times\mathcal{X}\times\Omega,$%
\begin{equation}
\mathfrak{T}\varphi\left(  U,V,\omega\right)  =\sum_{\omega^{\prime}\in\Omega
}\varphi\left(  U+\mathbb{M}\mathcal{T}\left(  V,\omega\right)  ,\left(
\mathbb{I}-\mathbb{M}\right)  \mathcal{T}\left(  V,\omega\right)
,\omega^{\prime}\right)  \Pi\left(  \omega^{\prime}|V\right)  \ . \label{TT2}%
\end{equation}
Hence, $\left\{  \mathfrak{z}_{t}\right\}  _{t\in\mathbb{Z}_{+}}$ such that,
$\forall t\geq0,\mathfrak{z}_{t}:=\left(  \mathfrak{v}_{t},\mathfrak{w}%
_{t}\right)  ,$ is an homogeneous Markov process.

We can rephrase (\ref{expconv}) and therefore the content of Theorem
\ref{main} in terms of exponential (more correctly geometric since
$t\in\mathbb{Z}_{+}$) convergence to an element of the set of the invariant
measures of the Markov chains defined by the transition operators
$\mathfrak{T}$ and $\mathbf{T}.$ More precisely, for any $\varepsilon>0$ and
$t>0,$ given $\chi^{0}=\left(  X^{0},\omega^{0}\right)  \in\left\{  X\in
\Xi:\Gamma\left(W\left(  X\right)\right)  >0\right\}  \times\Omega,$ by (\ref{dI}) $\left\{
\left\Vert \left(  \mathbb{I}-\mathbb{M}\right)  \mathfrak{x}_{t}\right\Vert
>\varepsilon\right\}  \subseteq\left\{  W\left(  t\right)  >\varepsilon
\right\}  .$ Hence, by Theorem \ref{main}, $\left\{  \mathfrak{v}_{t}\right\}
_{t\in\mathbb{Z}_{+}}$ converges to zero $\mathbb{P}_{0}-a.s.$ and, by
(\ref{TT2}), $\left\{  \mathfrak{w}_{t}\right\}  _{t\in\mathbb{Z}_{+}}$
converges to $\bar{1},\mathbb{P}_{0}-a.s..$ But, since, by (\ref{TT}), for any
$Y\in\mathcal{I},\omega\in\Omega,\mathfrak{T}\varphi\left(  Y,\omega\right)
=\varphi\left(  Y,\bar{1}\right)  ,\left\{  \mathfrak{u}_{t}\right\}
_{t\in\mathbb{Z}_{+}}$ converges $\mathbb{P}_{0}-a.s.$ to an element of
$\mathcal{I}$ which we denote by $X^{\infty}.$

Let us introduce on $\Omega$ the metric
\begin{equation}
\Omega\times\Omega\ni\left(  \omega,\omega^{\prime}\right)  \longmapsto
\mathbf{d}\left(  \omega,\omega^{\prime}\right)  :=\frac{1}{\left\vert
\mathbf{E}\right\vert }\sum_{e\in\mathbf{E}}\left(  1-\delta_{\omega
_{e},\omega_{e}^{\prime}}\right)  \in\left[  0,1\right]  \ .
\end{equation}

\begin{lemma}
From (\ref{W}), for any $\left(  X^{0},\omega^{0}\right)  \in\mathfrak{X}$ and
$t\geq1,$ we have
\begin{equation}
\mathbb{E}\left[  \mathbf{d}\left(  \mathfrak{w}_{t},\bar{1}\right)  |\left(
X^{0},\omega^{0}\right)  \right]  \leq\mathbb{E}\left[  \left(  1-p\left(
W\circ\mathfrak{x}_{t-1}\right)  \right)  |\left(  X^{0},\omega^{0}\right)
\right]  \ .
\end{equation}

\end{lemma}

\begin{proof}
For any $\omega\in\Omega,$ we get
\begin{equation}
\mathbf{d}\left(  \omega,\bar{1}\right)  =\frac{1}{\left\vert \mathbf{E}%
\right\vert }\sum_{e\in\mathbf{E}}\left(  1-\delta_{\omega_{e},1}\right)
=\frac{1}{\left\vert \mathbf{E}\right\vert }\sum_{e\in\mathbf{E}}\left(
1-\omega_{e}\right)  \ .
\end{equation}
Hence, by the Markov property, from (\ref{TT}) and (\ref{pomega1}) we have
\begin{align}
\mathbb{E}_{0}\left[  \mathbf{d}\left(  \mathfrak{w}_{t},\bar{1}\right)
\right]   &  =\mathbb{E}_{0}\left[  \mathbb{E}\left[  \mathbf{d}\left(
\mathfrak{w}_{t},\bar{1}\right)  |\left(  \mathfrak{x}_{t-1},\mathfrak{w}%
_{t-1}\right)  \right]  \right] \label{wbound}\\
&  =\mathbb{E}_{0}\left[  \sum_{\omega^{\prime}\in\Omega}\frac{1}{\left\vert
\mathbf{E}\right\vert }\sum_{e\in\mathbf{E}}\left(  1-\omega_{e}^{\prime
}\right)  \Pi\left(  \omega^{\prime}|\mathfrak{x}_{t-1}\right)  \right]
\nonumber\\
&  =\frac{1}{\left\vert \mathbf{E}\right\vert }\sum_{e\in\mathbf{E}}%
\mathbb{E}_{0}\left[  \sum_{\omega^{\prime}\in\Omega}\left(  1-\omega
_{e}^{\prime}\right)  \left(  \omega_{e}^{\prime}p\left(  \left\vert
\Delta_{e}\mathfrak{x}_{t-1}\right\vert \right)  \times\right.  \right.
\nonumber\\
&  \left.  \left.  \times\left(  1-\omega_{e}^{\prime}\right)  \left(
1-p\left(  \left\vert \Delta_{e}\mathfrak{x}_{t-1}\right\vert \right)
\right)  \right)  \right] \nonumber\\
&  =\frac{1}{\left\vert \mathbf{E}\right\vert }\sum_{e\in\mathbf{E}}%
\mathbb{E}_{0}\left[  \left(  1-p\left(  \left\vert \Delta_{e}\mathfrak{v}%
_{t-1}\right\vert \right)  \right)  \right] \nonumber\\
&  \leq\frac{1}{\left\vert \mathbf{E}\right\vert }\sum_{e\in\mathbf{E}%
}\mathbb{E}_{0}\left[  \left(  1-p\left(  W\left(  t-1\right)  \right)
\right)  \right] \nonumber\\
&  \leq\mathbb{E}_{0}\left[  \left(  1-p\left(  W\left(  t-1\right)  \right)
\right)  \right]  \ .\nonumber
\end{align}

\end{proof}

Given a bounded measurable function $\varphi$ on $\mathcal{X}\times\Omega\subset\mathfrak{X,}$ let
\begin{align}
\left\Vert \nabla_{V}\varphi\right\Vert _{1}  &  :=\sup_{\left(  V^{\prime
},\omega\right)  \in\mathcal{X}\times\Omega}\sum_{v\in\mathbf{V}%
}\left\vert \frac{\partial}{\partial V_{v}}\varphi\right\vert \left(
V^{\prime},\omega\right)  \ ,\label{gradXnorm}\\
\left\Vert \varphi\right\Vert _{\Omega}  &  :=\sup_{V\in\mathcal{X}}%
\sup_{\omega,\omega^{\prime}\in\Omega\ :\ \omega\neq\omega^{\prime}}%
\frac{\left\vert \varphi\left(  V,\omega\right)  -\varphi\left(
V,\omega^{\prime}\right)  \right\vert }{\mathbf{d}\left(  \omega
,\omega^{\prime}\right)  } \label{omegaLipnorm}%
\end{align}
and consider the Banach space $\mathbb{L}$ of measurable functions $\varphi$
on $\mathcal{X}\times\Omega,$ with norm
\begin{equation}
\left\Vert \varphi\right\Vert _{L}:=\sup_{\left(  V,\omega\right)  \in
\mathcal{X}\times\Omega}\left\vert \varphi\left(  V,\omega\right)
\right\vert +\left\Vert \nabla_{V}\varphi\right\Vert _{1}+\left\Vert
\varphi\right\Vert _{\Omega}\ . \label{Lnorm}%
\end{equation}
Since for any $X\in\Xi,$%
\begin{align}
\left\Vert \left(  \mathbb{I}-\mathbb{M}\right)  X\right\Vert  &  =\sup
_{v\in\mathbf{V}}\left\vert X_{v}-\left(  \mathbb{M}X\right)  _{v}\right\vert
=\sup_{v\in\mathbf{V}}\left\vert X_{v}-\frac{1}{\left\vert \mathbf{V}%
\right\vert }\sum_{u\in\mathbf{V}}X_{u}\right\vert \label{vbound}\\
&  =\sup_{v\in\mathbf{V}}\left\vert \left(  1-\frac{1}{\left\vert
\mathbf{V}\right\vert }\right)  X_{v}-\frac{1}{\left\vert \mathbf{V}%
\right\vert }\sum_{u\in\mathbf{V}\ :\ u\neq v}X_{u}\right\vert \nonumber\\
&  =\sup_{v\in\mathbf{V}}\left\vert \frac{\left\vert \mathbf{V}\right\vert
-1}{\left\vert \mathbf{V}\right\vert }X_{v}-\frac{1}{\left\vert \mathbf{V}%
\right\vert }\sum_{u\in\mathbf{V}\ :\ u\neq v}X_{u}\right\vert \nonumber\\
&  =\sup_{v\in\mathbf{V}}\left\vert \frac{1}{\left\vert \mathbf{V}\right\vert
}\sum_{u\in\mathbf{V}\ :\ u\neq v}\left(  X_{v}-X_{u}\right)  \right\vert \leq
W\left(  X\right)  \ ,\nonumber
\end{align}
then, for $\varphi\in\mathbb{L},$ by (\ref{vbound}), we have
\begin{align}
\left\vert \varphi\left(  \left(  \mathbb{I}-\mathbb{M}\right)  X,\omega
\right)  -\varphi\left(  0,\omega\right)  \right\vert  &  =\left\vert \int
_{0}^{1}ds\sum_{v\in\mathbf{V}}\left(  \frac{\partial}{\partial V_{v}}%
\varphi\right)  \left(  s\left(  \mathbb{I}-\mathbb{M}\right)  X,\omega
\right)  \left(  \left(  \mathbb{I}-\mathbb{M}\right)  X\right)
_{v}\right\vert \label{Xbound}\\
&  =\left\vert \int_{0}^{1}ds\sum_{v\in\mathbf{V}}\left(  \frac{\partial
}{\partial V_{v}}\varphi\right)  \left(  s\left(  X-\mathbb{M}X\right)
,\omega\right)  \frac{1}{\left\vert \mathbf{V}\right\vert }\sum_{u\in
\mathbf{V}\ :\ u\neq v}\left(  X_{v}-X_{u}\right)  \right\vert \nonumber\\
&  \leq\left\Vert \nabla_{V}\varphi\right\Vert _{1}W\left(  X\right)
\ .\nonumber
\end{align}

\begin{proposition}
Starting from an initial state $\chi^{0}=\left(  X^{0},\omega^{0}\right)
\in\left\{  X\in\Xi:\Gamma\left(W\left(  X\right)\right)  >0\right\}  \times\Omega
\subset\mathfrak{X},$ the Markov chain $\left\{  \chi_{t}\right\}  _{t\geq0}$
weakly converges to the degenerate random vector $\left(  X^{\infty},\bar
{1}\right)  \in\mathcal{I}\times\Omega,$ where $X^{\infty}$ is the
$\mathbb{P}_{0}-a.s.$ limit of the random process $\left\{  \mathfrak{u}%
_{t}\right\}  _{t\in\mathbb{Z}_{+}}.$ Moreover, if $p$ is concave function and
$\lim_{x\downarrow0}\frac{1-p\left(  x\right)  }{x}>0,$ the rate of
convergence is geometric.
\end{proposition}

\begin{proof}
Given $\varphi\in\mathbb{L},$ by (\ref{Xbound}) and (\ref{wbound}) we have
\begin{gather}
\left\vert \mathbb{E}_{0}\left[  \varphi\left(  \mathfrak{v}_{t}%
,\mathfrak{w}_{t}\right)  \right]  -\varphi\left(  0,\bar{1}\right)
\right\vert \leq\mathbb{E}_{0}\left[  \left\vert \varphi\left(  \mathfrak{v}%
_{t},\mathfrak{w}_{t}\right)  -\varphi\left(  0,\bar{1}\right)  \right\vert
\right] \\
\leq\mathbb{E}_{0}\left[  \left\vert \varphi\left(  \mathfrak{v}%
_{t},\mathfrak{w}_{t}\right)  -\varphi\left(  0,\mathfrak{w}_{t}\right)
\right\vert \right]  +\mathbb{E}_{0}\left[  \left\vert \varphi\left(
0,\mathfrak{w}_{t}\right)  -\varphi\left(  0,\bar{1}\right)  \right\vert
\right] \nonumber\\
\leq\left\Vert \nabla_{X}\varphi\right\Vert _{1}\mathbb{E}_{0}\left[  W\left(
t\right)  \right]  +\left\Vert \varphi\right\Vert _{\Omega}\mathbb{E}%
_{0}\left[  \left(  1-p\left(  W\left(  t-1\right)  \right)  \right)  \right]
\nonumber\\
\leq\left\Vert \varphi\right\Vert _{L}\left(  \mathbb{E}_{0}\left[  W\left(
t\right)  \right]  \vee\mathbb{E}_{0}\left[  \left(  1-p\left(  W\left(
t-1\right)  \right)  \right)  \right]  \right) \nonumber
\end{gather}
which tends to zero in the limit $t\rightarrow\infty$ by Theorem \ref{main}.
Clearly, if $p$ is concave, $1-p$ is convex, hence, by (\ref{expconv}),
\begin{align}
\mathbb{E}_{0}\left[  \left(  1-p\left(  W\left(  t-1\right)  \right)
\right)  \right]   &  \leq\left(  1-p\left(  \mathbb{E}_{0}\left[  W\left(
t-1\right)  \right]  \right)  \right) \\
&  \leq1-p\left(  W\left(  X^{0}\right)  \left(  1-\Gamma\left(W\left(  X^{0}\right)
\right)\right)  ^{t-1}\right)  \ ,\nonumber
\end{align}
therefore, if $\lim_{x\downarrow0}\frac{1-p\left(  x\right)  }{x}$ is positive
the rate of convergence is geometric.
\end{proof}

Similar conclusions hold for the Markov chain $\left\{  \mathfrak{y}%
_{t}\right\}  _{t\geq0}.$ Indeed, by (\ref{TTT}) and (\ref{decT}), setting
$X_{1}=\left(  U_{1},V_{1}\right)  ,X_{2}=\left(  U_{2},V_{2}\right)  ,$ for
any bounded measurable $\varphi:\Xi^{4}\times\Omega\rightarrow\mathbb{R},$%
\begin{align}
\left(  \mathbf{T}\varphi\right)  \left(  U_{1},V_{1},U_{2},V_{2}\right)   &
:=\sum_{\omega,\omega^{\prime}\in\Omega}\varphi\left(  U_{2}+\mathbb{M}%
\mathcal{T}\left(  V_{2},\omega\right)  ,\left(  \mathbb{I}-\mathbb{M}\right)
\mathcal{T}\left(  V_{2},\omega\right)  ,\right. \\
&  \left.  U_{2}+\mathbb{M}\mathcal{T}\left(  \mathcal{T}\left(  V_{2}%
,\omega\right)  ,\omega^{\prime}\right)  ,\left(  \mathbb{I}-\mathbb{M}%
\right)  \mathcal{T}\left(  \mathcal{T}\left(  V_{2},\omega\right)
,\omega^{\prime}\right)  \right)  \times\nonumber\\
&  \times\Pi\left(  \omega|V_{1}\right)  \Pi\left(  \omega^{\prime}%
|V_{2}\right)  \ .\nonumber
\end{align}
Hence, $\left\{  \mathfrak{Z}_{t}\right\}  _{t\geq0}$ such that $\forall
t\geq0,\mathfrak{Z}_{t}:=\left(  \mathfrak{v}_{2t-2},\mathfrak{v}%
_{2t-1}\right)  $ is an homogeneous Markov process.

Moreover, from (\ref{decreasing}), for any $t\geq0$ we get
\begin{align}
\mathbb{E}_{0}\left[  W\circ\mathfrak{x}_{t+1}|\mathfrak{F}_{t-1}^{X}\right]
&  =\mathbb{E}_{0}\left[  \left.  \mathbb{E}_{0}\left[  W\circ\mathfrak{x}%
_{t+1}|\mathfrak{F}_{t}^{X}\right]  \right\vert \mathfrak{F}_{t-1}^{X}\right]
\label{ub1}\\
&  =\mathbb{E}_{0}\left[  \left(  1-\Gamma\circ\mathfrak{x}_{t-1}\right)
W\circ\mathfrak{x}_{t}|\mathfrak{F}_{t-1}^{X}\right] \nonumber\\
&  =\mathbb{E}_{0}\left[  W\circ\mathfrak{x}_{t}|\mathfrak{F}_{t-1}%
^{X}\right]  \left(  1-\Gamma\circ\mathfrak{x}_{t-1}\right) \nonumber\\
&  \leq\left(  1-\Gamma\circ\mathfrak{x}_{t-1}\right)  \left(  1-\Gamma
\circ\mathfrak{x}_{t-2}\right)  W\circ\mathfrak{x}_{t-1}\nonumber\\
&  \leq\left(  1-\Gamma\circ W\circ\mathfrak{x}_{t-1}\right)  \left(
1-\Gamma\circ W\circ\mathfrak{x}_{t-2}\right)  W\circ\mathfrak{x}%
_{t-1}\nonumber\\
&  \leq\left(  1-\Gamma\circ W\circ\mathfrak{x}_{t-2}\right)  ^{2}%
W\circ\mathfrak{x}_{t-1}\ .\nonumber
\end{align}
Hence, setting by (\ref{pi_i}),
\begin{align}
\bar{\Gamma}  &  :=\delta_{t,2s-1}\Gamma\circ\pi_{1}+\delta_{t,2s}\Gamma
\circ\pi_{2}\\
\bar{W}  &  :=\delta_{t,2s-1}W\circ\pi_{1}+\delta_{t,2s}W\circ\pi
_{2}\ ,\;s\geq0
\end{align}
from (\ref{ub1}) we have
\begin{equation}
\mathbb{E}_{0}\left[  \bar{W}\circ\mathfrak{y}_{t+1}|\mathfrak{F}%
_{t}^{\mathfrak{y}}\right]  \leq\left(  1-\bar{\Gamma}\circ\bar{W}\circ\mathfrak{y}%
_{t-1}\right)  ^{2}\bar{W}\circ\mathfrak{y}_{t-1}\;,\ t\geq0\ . \label{ub2}%
\end{equation}
Therefore, proceeding as in (\ref{mainest}), by (\ref{ub2}) we get
\[
\mathbb{E}_{0}\left[  \bar{W}\circ\mathfrak{y}_{t}\right]  \leq\mathbb{E}%
_{0}\left[  \left(  1-\bar{\Gamma}\circ\bar{W}\circ\mathfrak{y}_{t-2}\right)  ^{2}\bar
{W}\circ\mathfrak{y}_{t-2}\right]  \ .
\]
Thus, iterating,
\begin{equation}
\mathbb{E}_{0}\left[  \mathbf{T}^{t}\bar{W}\right]  \leq\left(  1-\bar{\Gamma
}\circ\bar{W}\circ\mathfrak{y}_{0}\right)  ^{2t}\bar{W}\circ\mathfrak{y}_{0}\ .
\label{main1}%
\end{equation}

Let us denote by $\mathcal{L}$ be the Banach space of bounded measurable
functions $\varphi$ on $\mathcal{X}^{2}$ with norm
\begin{equation}
\left\Vert \varphi\right\Vert _{\mathcal{L}}:=\sup_{\left(  V_{1}%
,V_{2}\right)  \in\mathcal{X}^{2}}\left\vert \varphi\left(  V_{1}%
,V_{2}\right)  \right\vert +\left\Vert \nabla\varphi\right\Vert _{1}\ ,
\end{equation}
where%
\begin{equation}
\left\Vert \nabla\varphi\right\Vert _{1}:=\sup_{\left(  V^{\prime}%
,V^{\prime\prime}\right)  \in\mathcal{X}^{2}}\sum_{v\in\mathbf{V}%
}\left[  \left\vert \frac{\partial}{\partial\left(  V_{1}\right)  _{v}}%
\varphi\right\vert \left(  V^{\prime},V^{\prime\prime}\right)  +\left\vert
\frac{\partial}{\partial\left(  V_{2}\right)  _{v}}\varphi\right\vert \left(
V^{\prime},V^{\prime\prime}\right)  \right]  \ .
\end{equation}

\begin{corollary}
\label{gap}The Markov chain $\left\{  \mathfrak{y}_{t}\right\}  _{t\geq0}$
started at $\left(  X^{0},X^{0}\right)  $ where $X^{0}\in\left\{  X\in
\Xi:\Gamma\left(W\left(  X\right)\right)  >0\right\}  $ converges weakly to the degenerate
random vector $\left(  X^{\infty},X^{\infty}\right)  \in\Xi^{2}$ with
geometric rate.
\end{corollary}

\begin{proof}
Given $\varphi\in\mathcal{L},$ proceeding as in (\ref{Xbound}), by
(\ref{superm}) and (\ref{main1}), we have
\begin{align}
\left\vert \mathbb{E}_{0}\left[  \varphi\left(  \mathfrak{Z}_{t}\right)
\right]  -\varphi\left(  0,0\right)  \right\vert  &  \leq\mathbb{E}_{0}\left[
\left\vert \varphi\left(  \mathfrak{Z}_{t}\right)  -\varphi\left(  0,0\right)
\right\vert \right] \\
&  \leq\left\Vert \varphi\right\Vert _{\mathcal{L}}\mathbb{E}_{0}\left[
\bar{W}\circ\mathfrak{y}_{t}\right]  \ .\nonumber
\end{align}

\end{proof}

\section{Very large system evolution}

Given $N\in\mathbb{N},$ let $\mathbf{V}_{N}:=\left\{  1,..,N\right\}
\subset\mathbb{N}$ and denote by $\mathbf{E}_{N}$ the subset of $\mathbf{E}%
:=\left\{  \left(  u,v\right)  \in\mathbb{N}\times\mathbb{N}\right\}  $ such
that $\mathbf{E}_{N}:=\left\{  \left(  u,v\right)  \in\mathbf{V}_{N}%
\times\mathbf{V}_{N}\right\}  .$

In this section we set $\Xi:=\left[  0,1\right]  ^{\mathbb{N}}$ and, denoting
by $X_{N}:=\left(  X_{1},..,X_{N}\right)  $ the element of $\Xi_{N}:=\left[
0,1\right]  ^{N}$ representing the restriction of the beliefs configuration
$X\in\Xi$ to $\mathbf{V}_{N},$ we endow $\Xi$ with the metric induced by the
norm $\left\Vert X\right\Vert :=\sum_{N\in\mathbb{N}}2^{-N}\left\Vert
X_{N}\right\Vert _{\infty}.$

With a slight abuse of notation, setting as in (\ref{Omega}) $\Omega:=\left\{
\omega\in\left\{  0,1\right\}  ^{\mathbf{E}}:\forall u\in\mathbf{V}%
\ ,\ \omega_{u,u}=1\right\}  ,$ we denote by $\omega_{N}$ the restriction of
the configuration $\omega\in\Omega$ to $\Omega_{N}:=\left\{  0,1\right\}
^{\mathbf{E}_{N}}$ and, by (\ref{Eomega}), if $E:=E\left(  \omega\right)  ,$
we set $E_{N}:=E\cap\mathbf{E}_{N}.$

Then, for any $X\in\Xi,\overline{\Pi}\left(  \cdot|X\right)  $ denotes the
random field on $\left(  \Omega,\mathcal{C}\left(  \Omega\right)  \right)  $
such that, for any cylinder event $\mathbf{C}_{N}\left(  \omega^{\prime
}\right)  =\left\{  \omega\in\Omega:\omega_{N}=\omega^{\prime}\right\}
,N\geq1,\omega^{\prime}\in\Omega_{N},$%
\begin{equation}
\sum_{\omega\in\Omega}\overline{\Pi}\left(  \omega|X\right)  \mathbf{1}%
_{\mathbf{C}_{N}\left(  \omega^{\prime}\right)  }\left(  \omega\right)
=\Pi\left(  \omega^{\prime}|X_{N}\right)
\end{equation}
where $\Pi\left(  \omega^{\prime}|X_{N}\right)  $ is given by (\ref{pomega1}).
Moreover, for any $M\geq N,$%
\begin{equation}
\sum_{\omega\in\Omega_{M}}\Pi\left(  \omega|X_{M}\right)  \mathbf{1}_{\left\{
\omega\in\Omega_{M}\ :\ \omega_{N}=\omega^{\prime}\right\}  }\left(
\omega\right)  =\sum_{\omega\in\Omega}\overline{\Pi}\left(  \omega|X\right)
\mathbf{1}_{\mathbf{C}_{N}\left(  \omega^{\prime}\right)  }\left(
\omega\right)  =\Pi\left(  \omega^{\prime}|X_{N}\right)  \label{pim1}%
\end{equation}

Setting $\mathbf{V}:=\mathbb{N}$ for notational convenience, let
$\mathcal{K}\left(  \mathbf{V}\right)  $ be the set of the transition
probability kernels on $\left(  \mathbf{V},\mathcal{P}\left(  \mathbf{V}%
\right)  \right)  .$ From (\ref{Pomega}), given the $\mathcal{C}\left(
\Omega\right)  $-measurable function $\Omega\ni\omega\longmapsto P\left(
\omega\right)  \in\mathcal{K}\left(  \mathbf{V}\right)  $ such that
$\forall\omega\in\Omega,$%
\begin{equation}
P_{v,u}\left(  \omega\right)  =\frac{\sum_{e\in E_{v}^{-}\cap E\left(
\omega\right)  }r_{e}\delta_{e,\left(  u,v\right)  }}{\sum_{e\in E_{v}^{-}\cap
E\left(  \omega\right)  }r_{e}}=:\mathfrak{p}_{v,u}\left(  E\left(
\omega\right)  \right)  \in\left[  0,1\right]  \;;\;v,u\in\mathbf{V}\ ,
\end{equation}
we denote as in (\ref{T}) and (\ref{Tomega}) $\mathfrak{X}\ni\left(
X,\omega\right)  \longmapsto\overline{\mathcal{T}}_{\cdot}\left(
X,\omega\right)  \in\Xi$ such that, for any $v\in\mathbf{V}$ and any $\left(
X,\omega\right)  \in\mathfrak{X},\overline{\mathcal{T}}_{v}\left(
X,\omega\right)  :=\sum_{u\in\mathbf{V}}\mathfrak{p}_{v,u}\left(  E\left(
\omega\right)  \right)  X_{u}.$ Then, the operator
\begin{equation}
BM\left(  \mathfrak{X}\right)  \ni\varphi\longmapsto\overline{\mathfrak{T}%
}\varphi\left(  X,\omega\right)  :=\sum_{\omega^{\prime}\in\Omega}%
\varphi\left(  \overline{\mathcal{T}}\left(  X,\omega\right)  ,\omega^{\prime
}\right)  \overline{\Pi}\left(  \omega^{\prime}|X\right)  \in BM\left(
\mathfrak{X}\right)
\end{equation}
represents the transition probability kernel of the homogeneous Markov chain
$\left\{  \chi_{t}\right\}  _{t\geq0}$ on $\left(  \mathfrak{X}^{\mathbb{Z}%
_{+}},\mathfrak{C},\mathbb{P}_{0}\right)  $\ with initial condition $\left(
X^{0},\omega^{0}\right)  \in\mathfrak{X}$ such that, by (\ref{TT}), for any
$\omega^{\prime}\in\Omega_{N},B\in\mathcal{B}\left(  \left[  0,1\right]
^{N}\right)  ,$%
\begin{gather}
\mathbb{P}_{0}\left(  \left\{  \chi_{t+1}\in\mathbf{C}_{N}\left(  B\right)
\times\mathbf{C}_{N}\left(  \omega^{\prime}\right)  \right\}  |\chi
_{t}\right)  =\left(  \overline{\mathfrak{T}}\mathbf{1}_{\mathbf{C}_{N}\left(
B\right)  \times\mathbf{C}_{N}\left(  \omega^{\prime}\right)  }\right)
\left(  \chi_{t}\right) \\
=\Pi\left(  \omega^{\prime}|X_{N}\left(  t\right)  \right)  \mathbf{1}%
_{B}\left(  \overline{\mathcal{T}}_{N}\left(  X\left(  t\right)
,\omega\left(  t\right)  \right)  \right)  \ .\nonumber
\end{gather}
Consequently, the operator
\begin{equation}
\mathfrak{\bar{L}}\left(  \Xi^{2}\right)  \ni\varphi\longmapsto\left(
\overline{\mathbf{T}}\varphi\right)  \left(  X,Y\right)  :=\sum_{\omega
,\omega^{\prime}\in\Omega}\varphi\left(  \overline{\mathcal{T}}\left(
Y,\omega\right)  ,\overline{\mathcal{T}}\left(  \overline{\mathcal{T}}\left(
Y,\omega\right)  ,\omega^{\prime}\right)  \right)  \overline{\Pi}\left(
\omega|X\right)  \overline{\Pi}\left(  \omega^{\prime}|Y\right)
\in\mathfrak{\bar{L}}\left(  \Xi^{2}\right)  \ ,
\end{equation}
defined as in (\ref{TTT}), represents the transition probability kernel of the
homogeneous Markov chain $\left\{  \mathfrak{y}_{t}\right\}  _{t\geq0}$ on
$\left(  \mathfrak{X}^{\mathbb{Z}_{+}},\mathfrak{C},\mathbb{P}_{0}\right)  $
such that, for any $B\in\mathcal{B}\left(  \left[  0,1\right]  ^{2N}\right)
,$%
\begin{gather}
\mathbb{P}_{0}\left(  \left\{  \mathfrak{y}_{t+1}\in\mathbf{C}_{N}\left(
B\right)  \right\}  |\mathfrak{y}_{t}\right)  =\left(  \overline{\mathbf{T}%
}\mathbf{1}_{\mathbf{C}_{N}\left(  B\right)  }\right)  \left(  \mathfrak{y}%
_{t}\right)  =\\
\sum_{\omega,\omega^{\prime}\in\Omega}\mathbf{1}_{B}\left(  \overline
{\mathcal{T}}\left(  X\left(  2t-1\right)  ,\omega\right)  ,\overline
{\mathcal{T}}\left(  \overline{\mathcal{T}}\left(  X\left(  2t-1\right)
,\omega\right)  ,\omega^{\prime}\right)  \right)  \overline{\Pi}\left(
\omega|X\left(  2t-2\right)  \right)  \overline{\Pi}\left(  \omega^{\prime
}|X\left(  2t-1\right)  \right)  =\nonumber\\
\sum_{\omega,\omega^{\prime}\in\Omega}\mathbf{1}_{B}\left(  \overline
{\mathcal{T}}_{N}\left(  X\left(  2t-1\right)  ,\omega\right)  ,\overline
{\mathcal{T}}_{N}\left(  \overline{\mathcal{T}}\left(  X\left(  2t-1\right)
,\omega\right)  ,\omega^{\prime}\right)  \right)  \overline{\Pi}\left(
\omega|X\left(  2t-2\right)  \right)  \overline{\Pi}\left(  \omega^{\prime
}|X\left(  2t-11\right)  \right)  =\nonumber\\
\sum_{\omega,\omega^{\prime}\in\Omega}\mathbf{1}_{B}\left(  \mathcal{T}\left(
X_{N}\left(  2t-1\right)  ,\omega\right)  ,\mathcal{T}\left(  \overline
{\mathcal{T}}_{N}\left(  X\left(  2t-1\right)  ,\omega\right)  ,\omega
^{\prime}\right)  \right)  \overline{\Pi}\left(  \omega|X\left(  2t-2\right)
\right)  \overline{\Pi}\left(  \omega^{\prime}|X\left(  2t-1\right)  \right)
=\nonumber\\
\sum_{\omega,\omega^{\prime}\in\Omega}\mathbf{1}_{B}\left(  \mathcal{T}\left(
X_{N}\left(  2t-1\right)  ,\omega\right)  ,\mathcal{T}\left(  \mathcal{T}%
\left(  X_{N}\left(  2t-1\right)  ,\omega\right)  ,\omega^{\prime}\right)
\right)  \overline{\Pi}\left(  \omega|X\left(  2t-2\right)  \right)
\overline{\Pi}\left(  \omega^{\prime}|X\left(  2t-1\right)  \right)
\ .\nonumber
\end{gather}

\begin{remark}
\label{QL}Notice that if the cardinality of the set $R:=\left\{  \left(
u,v\right)  \in\mathbf{V}\times\mathbf{V}:r_{u,v}>0\right\}  $ is finite,
there exists $M>N$ such that
\begin{equation}
\mathbb{P}_{0}\left(  \left\{  \chi_{t+1}\in\mathbf{C}_{N}\left(  B\right)
\times\mathbf{C}_{N}\left(  \omega^{\prime}\right)  \right\}  |\chi
_{t}\right)  =\Pi\left(  \omega^{\prime}|X_{N}\left(  t\right)  \right)
\mathbf{1}_{B}\left(  \mathcal{T}\left(  X_{N}\left(  t\right)  ,\omega
_{M}\left(  t\right)  \right)  \right)
\end{equation}
and
\begin{gather}
\mathbb{P}_{0}\left(  \left\{  \mathfrak{y}_{t+1}\in\mathbf{C}_{N}\left(
B\right)  \right\}  |\mathfrak{y}_{t}\right)  =\\
\sum_{\omega_{M},\omega_{M}^{\prime}\in\Omega_{M}}\mathbf{1}_{B}\left(
\mathcal{T}\left(  X_{N}\left(  2t-1\right)  ,\omega\right)  ,\mathcal{T}%
\left(  \mathcal{T}\left(  X_{N}\left(  2t-1\right)  ,\omega\right)
,\omega^{\prime}\right)  \right)  \times\nonumber\\
\times\Pi\left(  \omega_{M}|X_{M}\left(  2t-2\right)  \right)  \Pi\left(
\omega_{M}^{\prime}|X_{M}\left(  2t-1\right)  \right)  =\nonumber\\
\left(  \mathbf{T1}_{\left\{  \left(  X^{\left(  1\right)  },X^{\left(
2\right)  }\right)  \in\Xi_{M}^{2}\ :\ \left(  X_{N}^{\left(  1\right)
},X_{N}^{\left(  2\right)  }\right)  \in B\right\}  }\right)  \left(
\mathfrak{y}_{t}\right)  \ .\nonumber
\end{gather}
In the following we make this assumption.
\end{remark}

\begin{proposition}
\label{cls}Let the initial datum $X^{0}\in\Xi$ be such that $\alpha:=\inf
_{N\in\mathbb{N}}\Gamma\left(W\left(  X_{N}^{0}\right)\right)  >0.$ Then, for any
$\varphi\in C\left(  \Xi\right)  ,$%
\begin{equation}
\lim_{t\rightarrow\infty}\left\vert \mathbb{E}\left[  \varphi\circ
\mathfrak{Z}_{t}|\mathfrak{Z}_{0}=\left(  X^{0},X^{0}\right)  \right]
-\varphi\left(  0,0\right)  \right\vert =0\ .
\end{equation}

\end{proposition}

\begin{proof}
For any $\varphi\in C\left(  \Xi\right)  $ there exists a sequence $\left\{
\varphi_{N}\right\}  _{N\in\mathbb{N}}$ such that, $\forall N\geq1,\varphi
_{N}$ is a continuous $\mathcal{B}\left(  \Xi_{N}\right)  $-measurable
function uniformly convergent to $\varphi$ \cite{Ge}. Hence, given
$\varepsilon>0,$ there exists $N_{\varepsilon}\geq1$ such that for any
$N>N_{\varepsilon},\left\Vert \varphi-\varphi_{N}\right\Vert _{\infty
}<\varepsilon.$ Moreover, denoting by $\mathcal{I}_{N}:=\bigcup\limits_{x\in
\left[  0,1\right]  }\left\{  X_{N}\in\Xi_{N}:\left(  X_{N}\right)
_{v}=x\ ,\ \forall v\in\mathbf{V}_{N}\right\}  $ and $\mathcal{X}%
_{N}:=\Xi_{N}\ominus\mathcal{I}_{N},$ by the Stone-Weierstrass theorem there
exists $\phi\in\mathcal{L}_{N}$ (with $\mathcal{L}_{N}$ defined as
$\mathcal{L}$ in the previous section with $\mathcal{X}$ replaced by $\mathcal{X}_{N}$) such that $\left\Vert \phi%
-\varphi_{N}\right\Vert _{\infty}<\varepsilon.$ Then, the thesis follows from
Corollary \ref{gap}. Indeed,
\begin{align}
\left\vert \mathbb{E}\left[  \varphi\circ\mathfrak{y}_{t}|\mathfrak{y}%
_{0}=\left(  X^{0},X^{0}\right)  \right]  -\varphi\left(  0,0\right)
\right\vert  &  \leq2\varepsilon+\left\vert \mathbb{E}\left[  \varphi_{N}%
\circ\mathfrak{y}_{t}|\mathfrak{y}_{0}=\left(  X^{0},X^{0}\right)  \right]
-\varphi_{N}\left(  0,0\right)  \right\vert \\
&  \leq4\varepsilon+\left\vert \mathbb{E}\left[  \phi\circ\mathfrak{y}%
_{t}|\mathfrak{y}_{0}=\left(  X^{0},X^{0}\right)  \right]  -\phi\left(
0,0\right)  \right\vert \nonumber\\
&  =4\varepsilon+\left\vert \delta^{\left(  X_{M}^{0},X_{M}^{0}\right)
}\mathbf{T}^{t}\left(  \phi-\phi\left(  0,0\right)  \right)
\right\vert \ ,\nonumber
\end{align}
where $M\geq N.$ But,
\begin{equation}
\delta^{\left(  X_{M}^{0},X_{M}^{0}\right)  }\mathbf{T}^{t}\left\vert \phi
-\phi\left(  0,0\right)  \right\vert \leq\left\Vert \phi
\right\Vert _{\mathcal{L}_{N}}\left(  1-\alpha\right)  ^{2t}\ .
\end{equation}

\end{proof}

\subsection{Monokinetic-type limit}

Given a metric space $A,$ let us denote by $\mathcal{M}_{A}$ the set of Radon
measures on $\left(  A,\mathcal{B}\left(  A\right)  \right)  $ equipped with
the weak topology generated by the distance
\begin{equation}
\mathfrak{d}\left(  \mu,\nu\right)  :=\sup_{\varphi\in Lip\left(  A\right)
\ :\ \left\Vert \varphi\right\Vert _{L}\leq1}\left\vert \mu\left[
\varphi\right]  -\nu\left[  \varphi\right]  \right\vert \;,\;\mu,\nu
\in\mathcal{M}_{A}\ ,
\end{equation}
which makes it a complete metric space. We then denote by $\mathcal{M}_{A}%
^{+}$ the convex and closed set of positive elements of $\mathcal{M}_{A},$ and
by $\mathfrak{P}\left(  A\right)  $ the set of probability measures which is a
compact subset of $\mathcal{M}_{A}^{+}.$ Moreover, given an operator
$T:BM\left(  A\right)  \circlearrowleft$ and $\mu\in\mathcal{M}_{A},T_{\#}\mu$
denotes the push-forward of $\mu$ under $T.$

In particular, if $\mathbb{R}^{2}\ni\left(  x,y\right)  \longmapsto\xi\left(
x,y\right)  =\left(  \xi_{1}\left(  x,y\right)  ,\xi_{2}\left(  x,y\right)
\right)  \in\mathbb{R}^{2}$ is a random vector distributed according to
$\mu\in\mathfrak{P}\left(  \mathbb{R}^{2}\right)  ,$ we denote by
$\mu^{\left(  1\right)  }$ the marginal w.r.t. the first component of $\xi$
and by $\mu^{\left(  2\right)  }$ the marginal w.r.t. the second component of
$\xi.$

%Given $N\geq1,$ let us denote by $\mathbf{V}_{N}$ the set of vertices
%$\mathbf{V}$ whose cardinality is equal to $N.$

In this section, for technical reasons, we consider the believes variables
$X_{v},v\in\mathbf{V}_{N}, N\geq 1,$ to take values in $\mathbb{R}_{+}$ rather than in
$\left[  0,1\right]  $ and, with abuse of notation, redefine $\Xi_{N}$ to be
equal to $\mathbb{R}_{+}^{N}$ so that a belief configuration is now a map
$\mathbf{V}_{N}\ni v\longmapsto X_{v}\in\Xi_{N}:=\mathbb{R}_{+}^{N}.$

We denote by $\mu_{N}^{X}:=\frac{1}{N}\sum_{v\in\mathbf{V}_{N}}\delta
_{\left\{  X_{v}\right\}  }$ the empirical probability measure on $\left(
\mathbb{R},\mathcal{B}\left(  \mathbb{R}\right)  \right)  $ associated to the
believes configuration $X\in\Xi_{N}$ and set $\mu_{N}^{X,Y}:=\frac{1}{N}%
\sum_{v\in\mathbf{V}_{N}}\delta_{\left\{  X_{v}\right\}  }\otimes
\delta_{\left\{  Y_{v}\right\}  }$ the empirical probability measure on
$\left(  \mathbb{R}^{2},\mathcal{B}\left(  \mathbb{R}^{2}\right)  \right)  $
relative to the believes configuration $\left(  X,Y\right)  \in\Xi_{N}^{2}.$

Given $X\in\Xi_{N},$ we denote by $\Pi_{N}\left(  \cdot|X\right)  $ the
probability measure defined in (\ref{pomega1}).

To fix the argument, let us suppose first that the confidence levels $r_{i,j}$
are equal to $1$ for any $\left(  i,j\right)  \in\mathbf{V}_{N}\times
\mathbf{V}_{N}.$\ We will discuss the more general case where $\forall\left(
i,j\right)  \in\mathbf{V}_{N}\times\mathbf{V}_{N}=:\mathbf{E}_{N}%
,r_{i,j}:=\varrho\left(  X_{j},X_{i}\right)  ,$ with $\varrho:\left[
0,1\right]  ^{2}\longrightarrow\left[  0,1\right]  $ chosen to satisfy some
particular assumptions, at the end of this subsection.

For any $\omega\in\Omega_{N}:=\left\{  \eta\in\left\{  0,1\right\}
^{\mathbf{E}_{N}}:\forall u\in\mathbf{V}_{N},\eta_{u,u}=1\right\}  ,$ let us
consider the maps $\Xi_{N}\ni X\longmapsto\mathcal{T}\left(
X,\omega\right)  \in\Xi_{N}$ and $\Xi_{N}\ni X\longmapsto\mathcal{\tilde{T}%
}\left(  X,\omega\right)  \in\Xi_{N},$ where, as in (\ref{Tomega}), $\forall
v\in\mathbf{V}_{N},$%
\begin{equation}
\left(  \Xi_{N},\Omega_{N}\right)  \ni\left(  X,\omega\right)  \longmapsto
\mathcal{T}_{v}\left(  X,\omega\right)  :=\frac{\sum_{u\in\mathbf{V}}%
\omega_{v,u}X_{u}}{\sum_{u\in\mathbf{V}}\omega_{v,u}}\in\mathbb{R}_{+}\ ,
\label{Tomega1}%
\end{equation}
and
\begin{equation}
\left(  \Xi_{N},\Omega_{N}\right)  \ni\left(  X,\omega\right)  \longmapsto
\mathcal{\tilde{T}}_{v}\left(  X,\omega\right)  :=\frac{\frac{1}{N}\sum
_{u\in\mathbf{V}}\omega_{v,u}X_{u}}{\sum_{\omega\in\Omega_{N}}\Pi_{N}\left(
\omega|X\right)  \sum_{u\in\mathbf{V}_{N}}\frac{\omega_{v,u}}{N}}=\frac
{\frac{1}{N}\sum_{u\in\mathbf{V}}\omega_{v,u}X_{u}}{\mu_{N}^{X}\left[
p\left(  \left\vert X_{v}-\cdot\right\vert \right)  \right]  }\in
\mathbb{R}_{+}\ .
\end{equation}
Denoting by $\mathbf{\tilde{T}}$ the operator on $BM\left(  \Xi_{N}%
^{2}\right)  $ defined as in (\ref{TTT}), i.e.
\begin{equation}
\left(  \mathbf{\tilde{T}}\varphi\right)  \left(  X^{\left(  1\right)
},X^{\left(  2\right)  }\right)  :=\sum_{\omega,\omega^{\prime}\in\Omega_{N}%
}\varphi\left(  \mathcal{\tilde{T}}\left(  X^{\left(  2\right)  }%
,\omega\right)  ,\mathcal{\tilde{T}}\left(  \mathcal{\tilde{T}}\left(
X^{\left(  2\right)  },\omega\right)  ,\omega^{\prime}\right)  \right)
\Pi\left(  \omega|X^{\left(  1\right)  }\right)  \Pi\left(  \omega^{\prime
}|X^{\left(  2\right)  }\right)  \ , \label{TTTtilda}%
\end{equation}
we have:

\begin{lemma}
\label{mf1}If the sequence $\left\{  \mu_{N}^{X^{\left(  1\right)
},X^{\left(  2\right)  }}\right\}  _{N\geq1}$ weakly converges to $\mu,$ the
sequence of the characteristic function of the elements of the sequence of
measures $\left\{  \mathbf{\tilde{T}}_{\#}\mu_{N}^{X^{\left(  1\right)
},X^{\left(  2\right)  }}\right\}  _{N\geq1}$ converges pointwise to
\begin{equation}
\mathbb{R}^{2}\ni\left(  \lambda_{1},\lambda_{2}\right)  \longmapsto
\int_{\mathbb{R}_{+}^{2}}\mu\left(  dx,dy\right)  e^{i\lambda_{1}\frac
{\int_{\mathbb{R}_{+}^{2}}\mu\left(  dx^{\prime},dy^{\prime}\right)  p\left(
\left\vert x-x^{\prime}\right\vert \right)  y^{\prime}}{\int_{\mathbb{R}_{+}%
}\mu^{\left(  1\right)  }\left(  dx^{\prime}\right)  p\left(  \left\vert
x-x^{\prime}\right\vert \right)  }+i\lambda_{2}\int_{\mathbb{R}_{+}^{2}}%
\mu\left(  dx^{\prime},dy^{\prime}\right)  \frac{p\left(  \left\vert
y-y^{\prime}\right\vert \right)  \int_{\mathbb{R}_{+}^{2}}\mu\left(
dx^{\prime\prime},dy^{\prime\prime}\right)  p\left(  \left\vert x^{\prime
}-x^{\prime\prime}\right\vert \right)  y^{\prime\prime}}{\int_{\mathbb{R}_{+}%
}\mu^{\left(  2\right)  }\left(  dy^{\prime}\right)  p\left(  \left\vert
y-y^{\prime}\right\vert \right)  \int_{\mathbb{R}_{+}}\mu^{\left(  1\right)
}\left(  dx^{\prime\prime}\right)  p\left(  \left\vert x^{\prime}%
-x^{\prime\prime}\right\vert \right)  }}\in\mathbb{C} \label{CFT}%
\end{equation}
provided that $\forall i=1,2,$ the function $\mathbb{R}^{+}\ni x\longmapsto
\int_{\mathbb{R}_{+}}\mu^{\left(  i\right)  }\left(  dx^{\prime}\right)
p\left(  \left\vert x-x^{\prime}\right\vert \right)  \in\left[  0,1\right]  $
is strictly positive.
\end{lemma}

\begin{proof}
Setting for any $\left(  u,v\right)  \in\mathbf{E}_{N}$ and $X\in\Xi
_{N},p\left(  \left\vert \Delta_{u,v}X\right\vert \right)  :=p\left(
\left\vert X_{u}-X_{v}\right\vert \right)  ,$ given $\left(  X^{\left(
1\right)  },X^{\left(  2\right)  }\right)  \in\Xi_{N}^{2},$ the characteristic
function of $\mathbf{\bar{T}}_{\#}\mu_{N}^{X^{\left(  1\right)  },X^{\left(
2\right)  }},$ namely
\begin{equation}
\mathbb{R}^{2}\ni\left(  \lambda_{1},\lambda_{2}\right)  \longmapsto
\int_{\mathbb{R}_{+}^{2}}\mathbf{\tilde{T}}_{\#}\mu_{N}^{X^{\left(  1\right)
},X^{\left(  2\right)  }}\left(  dx,dy\right)  e^{i\lambda_{1}x+i\lambda_{2}%
y}\in\mathbb{C}\ ,
\end{equation}
writes
\begin{gather}
\int_{\mathbb{R}_{+}^{2}}\mathbf{\tilde{T}}_{\#}\mu_{N}^{X^{\left(  1\right)
},X^{\left(  2\right)  }}\left(  dx,dy\right)  e^{i\lambda_{1}x+i\lambda_{2}%
y}=\frac{1}{N}\sum_{v=1}^{N}\sum_{\omega,\omega^{\prime}\in\Omega_{N}}%
\Pi\left(  \omega|X^{\left(  1\right)  }\right)  \Pi\left(  \omega^{\prime
}|X^{\left(  2\right)  }\right)  \times\\
\times e^{i\frac{\lambda_{1}}{N}\sum_{u=1}^{N}\frac{\omega_{v,u}}{\mu
_{N}^{X^{\left(  1\right)  }}\left[  p\left(  \left\vert X_{v}^{\left(
1\right)  }-\cdot\right\vert \right)  \right]  }X_{u}^{\left(  2\right)
}+i\frac{\lambda_{2}}{N}\sum_{u=1}^{N}\frac{\omega_{v,u}^{\prime}}{\mu
_{N}^{X^{\left(  2\right)  }}\left[  p\left(  \left\vert X_{v}^{\left(
2\right)  }-\cdot\right\vert \right)  \right]  }\frac{1}{N}\sum_{w=1}^{N}%
\frac{\omega_{u,w}}{\mu_{N}^{X^{\left(  1\right)  }}\left[  p\left(
\left\vert X_{u}^{\left(  1\right)  }-\cdot\right\vert \right)  \right]
}X_{w}^{\left(  2\right)  }}\nonumber\\
=\frac{1}{N}\sum_{v=1}^{N}\sum_{\omega,\omega^{\prime}\in\Omega_{N}}\Pi\left(
\omega|X^{\left(  1\right)  }\right)  \Pi\left(  \omega^{\prime}|X^{\left(
2\right)  }\right)  e^{i\left(  \frac{\frac{\lambda_{1}}{\mu_{N}^{X^{\left(
1\right)  }}\left[  p\left(  \left\vert X_{v}^{\left(  1\right)  }%
-\cdot\right\vert \right)  \right]  }}{N}+\frac{\frac{\lambda_{2}}{\mu
_{N}^{X^{\left(  1\right)  }}\left[  p\left(  \left\vert X_{v}^{\left(
1\right)  }-\cdot\right\vert \right)  \right]  \mu_{N}^{X^{\left(  2\right)
}}\left[  p\left(  \left\vert X_{v}^{\left(  2\right)  }-\cdot\right\vert
\right)  \right]  }}{N^{2}}\right)  \sum_{w=1}^{N}\omega_{v,w}X_{w}^{\left(
2\right)  }}\times\nonumber\\
\times e^{i\frac{\lambda_{2}}{N}\sum_{u\in\mathbf{V}_{N}\backslash\left\{
v\right\}  }\frac{\omega_{v,u}^{\prime}}{\mu_{N}^{X^{\left(  2\right)  }%
}\left[  p\left(  \left\vert X_{v}^{\left(  2\right)  }-\cdot\right\vert
\right)  \right]  }\frac{1}{N}\sum_{w=1}^{N}\frac{\omega_{u,w}}{\mu
_{N}^{X^{\left(  1\right)  }}\left[  p\left(  \left\vert X_{u}^{\left(
1\right)  }-\cdot\right\vert \right)  \right]  }X_{w}^{\left(  2\right)  }%
}\nonumber\\
=\frac{1}{N}\sum_{v=1}^{N}\sum_{\omega,\omega^{\prime}\in\Omega_{N}}\Pi\left(
\omega|X^{\left(  1\right)  }\right)  \Pi\left(  \omega^{\prime}|X^{\left(
2\right)  }\right)  \prod\limits_{w=1}^{N}\left[  e^{i\left(  \frac
{\frac{\lambda_{1}}{\mu_{N}^{X^{\left(  1\right)  }}\left[  p\left(
\left\vert X_{v}^{\left(  1\right)  }-\cdot\right\vert \right)  \right]  }}%
{N}+\frac{\frac{\lambda_{2}}{\mu_{N}^{X^{\left(  1\right)  }}\left[  p\left(
\left\vert X_{v}^{\left(  1\right)  }-\cdot\right\vert \right)  \right]
\mu_{N}^{X^{\left(  2\right)  }}\left[  p\left(  \left\vert X_{v}^{\left(
2\right)  }-\cdot\right\vert \right)  \right]  }}{N^{2}}\right)
X_{w}^{\left(  2\right)  }}\omega_{v,w}+\right.  \nonumber\\
\left.  +\left(  1-\omega_{v,w}\right)  \right]  \prod\limits_{u\in
\mathbf{V}_{N}\backslash\left\{  v\right\}  }\prod\limits_{w\in\mathbf{V}_{N}%
}\left[  e^{i\frac{\frac{\lambda_{2}}{\mu_{N}^{X^{\left(  1\right)  }}\left[
p\left(  \left\vert X_{u}^{\left(  1\right)  }-\cdot\right\vert \right)
\right]  \mu_{N}^{X^{\left(  2\right)  }}\left[  p\left(  \left\vert
X_{v}^{\left(  2\right)  }-\cdot\right\vert \right)  \right]  }}{N^{2}}%
X_{w}^{\left(  2\right)  }}\omega_{v,u}^{\prime}\omega_{u,w}+\left(
1-\omega_{v,u}^{\prime}\omega_{u,w}\right)  \right]  \nonumber\\
=\frac{1}{N}\sum_{v=1}^{N}\sum_{\omega,\omega^{\prime}\in\Omega_{N}}%
\prod\limits_{w=1}^{N}\left[  e^{i\left(  \frac{\frac{\lambda_{1}}{\mu
_{N}^{X^{\left(  1\right)  }}\left[  p\left(  \left\vert X_{v}^{\left(
1\right)  }-\cdot\right\vert \right)  \right]  }}{N}+\frac{\frac{\lambda_{2}%
}{\mu_{N}^{X^{\left(  1\right)  }}\left[  p\left(  \left\vert X_{v}^{\left(
1\right)  }-\cdot\right\vert \right)  \right]  \mu_{N}^{X^{\left(  2\right)
}}\left[  p\left(  \left\vert X_{v}^{\left(  2\right)  }-\cdot\right\vert
\right)  \right]  }}{N^{2}}\right)  X_{w}^{\left(  2\right)  }}\omega
_{v,w}+\left(  1-\omega_{v,w}\right)  \right]  \times\nonumber\\
\times\left[  \omega_{v,w}p\left(  \left\vert \Delta_{v,w}X^{\left(  1\right)
}\right\vert \right)  +\left(  1-\omega_{v,w}\right)  \left(  1-p\left(
\left\vert \Delta_{v,w}X^{\left(  1\right)  }\right\vert \right)  \right)
\right]  \times\nonumber\\
\times\prod\limits_{u\in\mathbf{V}_{N}\backslash\left\{  v\right\}  }%
\prod\limits_{w\in\mathbf{V}_{N}}\left[  e^{i\frac{\frac{\lambda_{1}}{\mu
_{N}^{X^{\left(  1\right)  }}\left[  p\left(  \left\vert X_{u}^{\left(
1\right)  }-\cdot\right\vert \right)  \right]  \mu_{N}^{X^{\left(  2\right)
}}\left[  p\left(  \left\vert X_{v}^{\left(  2\right)  }-\cdot\right\vert
\right)  \right]  }}{N^{2}}X_{w}^{\left(  2\right)  }}\omega_{v,u}^{\prime
}\omega_{u,w}+\left(  1-\omega_{v,u}^{\prime}\omega_{u,w}\right)  \right]
\times\nonumber\\
\times\left[  \omega_{v,u}^{\prime}\omega_{u,w}p\left(  \left\vert
\Delta_{v,u}X^{\left(  2\right)  }\right\vert \right)  p\left(  \left\vert
\Delta_{u,w}X^{\left(  1\right)  }\right\vert \right)  +\left(  1-\omega
_{v,u}^{\prime}\omega_{u,w}\right)  \left(  1-p\left(  \left\vert \Delta
_{v,u}X^{\left(  2\right)  }\right\vert \right)  p\left(  \left\vert
\Delta_{u,w}X^{\left(  1\right)  }\right\vert \right)  \right)  \right]
\nonumber
\end{gather}%
\begin{gather*}
=\frac{1}{N}\sum_{v=1}^{N}\prod\limits_{w=1}^{N}\left[  e^{i\left(
\frac{\frac{\lambda_{1}}{\mu_{N}^{X^{\left(  1\right)  }}\left[  p\left(
\left\vert X_{v}^{\left(  1\right)  }-\cdot\right\vert \right)  \right]  }}%
{N}+\frac{\frac{\lambda_{2}}{\mu_{N}^{X^{\left(  1\right)  }}\left[  p\left(
\left\vert X_{v}^{\left(  1\right)  }-\cdot\right\vert \right)  \right]
\mu_{N}^{X^{\left(  2\right)  }}\left[  p\left(  \left\vert X_{v}^{\left(
2\right)  }-\cdot\right\vert \right)  \right]  }}{N^{2}}\right)
X_{w}^{\left(  2\right)  }}p\left(  \left\vert \Delta_{v,w}X^{\left(
1\right)  }\right\vert \right)  +\right.  \\
\left.  +\left(  1-p\left(  \left\vert \Delta_{v,w}X^{\left(  1\right)
}\right\vert \right)  \right)  \right]  \prod\limits_{u\in\mathbf{V}%
_{N}\backslash\left\{  v\right\}  }\prod\limits_{w\in\mathbf{V}_{N}}\left[
e^{i\frac{\frac{\lambda_{1}}{\mu_{N}^{X^{\left(  1\right)  }}\left[  p\left(
\left\vert X_{u}^{\left(  1\right)  }-\cdot\right\vert \right)  \right]
\mu_{N}^{X^{\left(  2\right)  }}\left[  p\left(  \left\vert X_{v}^{\left(
2\right)  }-\cdot\right\vert \right)  \right]  }}{N^{2}}X_{w}^{\left(
2\right)  }}\times\right.  \\
\left.  \times p\left(  \left\vert \Delta_{v,u}X^{\left(  2\right)
}\right\vert \right)  p\left(  \left\vert \Delta_{u,w}X^{\left(  1\right)
}\right\vert \right)  +\left(  1-p\left(  \left\vert \Delta_{v,u}X^{\left(
2\right)  }\right\vert \right)  p\left(  \left\vert \Delta_{u,w}X^{\left(
1\right)  }\right\vert \right)  \right)  \right]  \\
=\frac{1}{N}\sum_{v=1}^{N}\exp\left\{  \sum_{w=1}^{N}\log\left[  1+p\left(
\left\vert \Delta_{v,w}X^{\left(  1\right)  }\right\vert \right)  \left(
e^{i\left(  \frac{\frac{\lambda_{1}}{\mu_{N}^{X^{\left(  1\right)  }}\left[
p\left(  \left\vert X_{v}^{\left(  1\right)  }-\cdot\right\vert \right)
\right]  }}{N}+\frac{\frac{\lambda_{2}}{\mu_{N}^{X^{\left(  1\right)  }%
}\left[  p\left(  \left\vert X_{v}^{\left(  1\right)  }-\cdot\right\vert
\right)  \right]  \mu_{N}^{X^{\left(  2\right)  }}\left[  p\left(  \left\vert
X_{v}^{\left(  2\right)  }-\cdot\right\vert \right)  \right]  }}{N^{2}%
}\right)  X_{w}^{\left(  2\right)  }}-1\right)  \right]  \right\}  \times\\
\times\exp\left\{  \sum_{u\in\mathbf{V}_{N}\backslash\left\{  v\right\}  }%
\sum_{w\in\mathbf{V}_{N}}\log\left[  1+p\left(  \left\vert \Delta
_{v,u}X^{\left(  2\right)  }\right\vert \right)  p\left(  \left\vert
\Delta_{u,w}X^{\left(  1\right)  }\right\vert \right)  \left(  e^{i\frac
{\frac{\lambda_{2}}{\mu_{N}^{X^{\left(  1\right)  }}\left[  p\left(
\left\vert X_{u}^{\left(  1\right)  }-\cdot\right\vert \right)  \right]
\mu_{N}^{X^{\left(  2\right)  }}\left[  p\left(  \left\vert X_{v}^{\left(
2\right)  }-\cdot\right\vert \right)  \right]  }}{N^{2}}X_{w}^{\left(
2\right)  }}-1\right)  \right]  \right\}  \\
=\frac{1}{N}\sum_{v=1}^{N}\exp\left\{  \sum_{w=1}^{N}\log\left[  1+p\left(
\left\vert \Delta_{v,w}X^{\left(  1\right)  }\right\vert \right)  \left(
e^{i\left(  \frac{\frac{\lambda_{1}}{\mu_{N}^{X^{\left(  1\right)  }}\left[
p\left(  \left\vert X_{v}^{\left(  1\right)  }-\cdot\right\vert \right)
\right]  }}{N}+\frac{\frac{\lambda_{2}}{\mu_{N}^{X^{\left(  1\right)  }%
}\left[  p\left(  \left\vert X_{v}^{\left(  1\right)  }-\cdot\right\vert
\right)  \right]  \mu_{N}^{X^{\left(  2\right)  }}\left[  p\left(  \left\vert
X_{v}^{\left(  2\right)  }-\cdot\right\vert \right)  \right]  }}{N^{2}%
}\right)  X_{w}^{\left(  2\right)  }}-1\right)  \right]  +\right.  \\
\left.  -\sum_{w=1}^{N}\log\left[  1+p\left(  \left\vert \Delta_{v,w}%
X^{\left(  1\right)  }\right\vert \right)  \left(  e^{i\frac{\frac{\lambda
_{2}}{\mu_{N}^{X^{\left(  1\right)  }}\left[  p\left(  \left\vert
X_{v}^{\left(  1\right)  }-\cdot\right\vert \right)  \right]  \mu
_{N}^{X^{\left(  2\right)  }}\left[  p\left(  \left\vert X_{v}^{\left(
2\right)  }-\cdot\right\vert \right)  \right]  }}{N^{2}}X_{w}^{\left(
2\right)  }}-1\right)  \right]  \right\}  \times\\
\times\exp\left\{  \sum_{u=1}^{N}\sum_{w=1}^{N}\log\left[  1+p\left(
\left\vert \Delta_{v,u}X^{\left(  2\right)  }\right\vert \right)  p\left(
\left\vert \Delta_{u,w}X^{\left(  1\right)  }\right\vert \right)  \left(
e^{i\frac{\frac{\lambda_{2}}{\mu_{N}^{X^{\left(  1\right)  }}\left[  p\left(
\left\vert X_{u}^{\left(  1\right)  }-\cdot\right\vert \right)  \right]
\mu_{N}^{X^{\left(  2\right)  }}\left[  p\left(  \left\vert X_{v}^{\left(
2\right)  }-\cdot\right\vert \right)  \right]  }}{N^{2}}X_{w}^{\left(
2\right)  }}-1\right)  \right]  \right\}
\end{gather*}%
\begin{gather*}
=\int_{\mathbb{R}_{+}^{2}}\mu_{N}^{X^{\left(  1\right)  }X^{\left(  2\right)
}}\left(  dx,dy\right)  \left\{  \exp\left[  N\int_{\mathbb{R}_{+}}\mu
_{N}^{X^{\left(  1\right)  },X^{\left(  2\right)  }}\left(  dy^{\prime
}\right)  \times\right.  \right.  \\
\left.  \times\log\left(  1+p\left(  \left\vert y-x^{\prime}\right\vert
\right)  \left(  e^{i\left(  \frac{\frac{\lambda_{2}}{\mu_{N}^{X^{\left(
1\right)  }}\left[  p\left(  \left\vert x-\cdot\right\vert \right)  \right]
}}{N}+\frac{\frac{\lambda_{1}}{\mu_{N}^{X^{\left(  1\right)  }}\left[
p\left(  \left\vert x-\cdot\right\vert \right)  \right]  \mu_{N}^{X^{\left(
2\right)  }}\left[  p\left(  \left\vert y-\cdot\right\vert \right)  \right]
}}{N^{2}}\right)  y^{\prime}}-1\right)  \right)  \right]  +\\
\left.  -N\int_{\mathbb{R}_{+}^{2}}\mu_{N}^{X^{\left(  1\right)  },X^{\left(
2\right)  }}\left(  dx^{\prime},dy^{\prime}\right)  \log\left[  1+p\left(
\left\vert x-x^{\prime}\right\vert \right)  \left(  e^{i\frac{\frac
{\lambda_{1}}{\mu_{N}^{X^{\left(  1\right)  }}\left[  p\left(  \left\vert
x-\cdot\right\vert \right)  \right]  \mu_{N}^{X^{\left(  2\right)  }}\left[
p\left(  \left\vert y-\cdot\right\vert \right)  \right]  }}{N^{2}}y^{\prime}%
}-1\right)  \right]  \right\}  \times\\
\times\left\{  \exp N^{2}\int_{\mathbb{R}_{+}^{2}}\mu_{N}^{X^{\left(
1\right)  },X^{\left(  2\right)  }}\left(  dx^{\prime},dy^{\prime}\right)
\int_{\mathbb{R}_{+}^{2}}\mu_{N}^{X^{\left(  1\right)  },X^{\left(  2\right)
}}\left(  dx^{\prime\prime}dy^{\prime\prime}\right)  \times\right.  \\
\left.  \times\log\left[  1+p\left(  \left\vert y-y^{\prime}\right\vert
\right)  p\left(  \left\vert x^{\prime}-x^{\prime\prime}\right\vert \right)
\left(  e^{i\frac{\frac{\lambda_{1}}{\mu_{N}^{X^{\left(  1\right)  }}\left[
p\left(  \left\vert x^{\prime}-\cdot\right\vert \right)  \right]  \mu
_{N}^{X^{\left(  2\right)  }}\left[  p\left(  \left\vert y-\cdot\right\vert
\right)  \right]  }}{N^{2}}y^{\prime\prime}}-1\right)  \right]  \right\}  \ .
\end{gather*}
Taking the limit as $N\uparrow\infty,$ if $\left\{  \mu_{N}^{X^{\left(
1\right)  },X^{\left(  2\right)  }}\right\}  _{N\geq1}$ weakly converges to
$\mu,$ the sequence of the characteristic functions of $\mathbf{\tilde{T}%
}_{\#}\mu_{N}^{X^{\left(  1\right)  },X^{\left(  2\right)  }}$ converges
pointwise to (\ref{CFT}).
\end{proof}

Let $\left\{  \mathbf{Z}_{t}\right\}  _{t\geq0},$ where, for any
$t\geq0,\mathbf{Z}_{t}:=\left(  Z_{t}^{\left(  1\right)  },Z_{t}^{\left(
2\right)  }\right)  ,$ be the non-linear Markov chain on $\left(  \left[
0,1\right]  ^{2},\mathcal{B}\left(  \left[  0,1\right]  ^{2}\right)  \right)
$ with degenerate kernel such that, for any bounded measurable $\varphi
:\left[  0,1\right]  ^{2}\longrightarrow\mathbb{R},$ if $\mu_{t}$ is the law
of $\mathbf{Z}_{t},$%
\begin{align}
\mathbb{E}\left[  \varphi\left(  \mathbf{Z}_{t+1}\right)  |\mathbf{Z}%
_{t}\right]   &  :=\varphi\left(  \frac{\int_{\left[  0,1\right]  ^{2}}\mu
_{t}\left(  dx^{\prime},dy^{\prime}\right)  p\left(  \left\vert Z_{t}^{\left(
1\right)  }-x^{\prime}\right\vert \right)  y^{\prime}}{\int_{0}^{1}\mu
_{t}^{\left(  1\right)  }\left(  dx^{\prime}\right)  p\left(  \left\vert
Z_{t}^{\left(  1\right)  }-x^{\prime}\right\vert \right)  },\right. \\
&  \left.  \int_{\left[  0,1\right]  ^{2}}\mu_{t}\left(  dx^{\prime
},dy^{\prime}\right)  \frac{p\left(  \left\vert Z_{t}^{\left(  2\right)
}-y^{\prime}\right\vert \right)  \int_{\left[  0,1\right]  ^{2}}\mu_{t}\left(
dx^{\prime\prime},dy^{\prime\prime}\right)  p\left(  \left\vert x^{\prime
}-x^{\prime\prime}\right\vert \right)  y^{\prime\prime}}{\int_{0}^{1}\mu
_{t}^{\left(  2\right)  }\left(  dy^{\prime}\right)  \left[  p\left(
\left\vert Z_{t}^{\left(  2\right)  }-y^{\prime}\right\vert \right)  \right]
\int_{0}^{1}\mu_{t}^{\left(  1\right)  }\left(  dx^{\prime\prime}\right)
\left[  p\left(  \left\vert x^{\prime}-x^{\prime\prime}\right\vert \right)
\right]  }\right)  \ .\nonumber
\end{align}
For any $\mu\in\mathfrak{P}\left(  \mathcal{B}\left(  \left[  0,1\right]
^{2}\right)  \right)  $ we denote by
\begin{gather}
\Theta_{\#}\mu\left[  \varphi\right]  :=\int_{0}^{1}\mu\left(  dx,dy\right)
\varphi\left(  \frac{\int_{\left[  0,1\right]  ^{2}}\mu\left(  dx^{\prime
},dy^{\prime}\right)  p\left(  \left\vert x-x^{\prime}\right\vert \right)
y^{\prime}}{\int_{0}^{1}\mu^{\left(  1\right)  }\left(  dx^{\prime}\right)
p\left(  \left\vert x-x^{\prime}\right\vert \right)  },\right.  \label{Theta}%
\\
\left.  \int_{\left[  0,1\right]  ^{2}}\mu\left(  dx^{\prime},dy^{\prime
}\right)  \frac{p\left(  \left\vert y-y^{\prime}\right\vert \right)
\int_{\left[  0,1\right]  ^{2}}\mu\left(  dx^{\prime\prime},dy^{\prime\prime
}\right)  p\left(  \left\vert x^{\prime}-x^{\prime\prime}\right\vert \right)
y^{\prime\prime}}{\int_{0}^{1}\mu^{\left(  2\right)  }\left(  dy^{\prime
}\right)  \left[  p\left(  \left\vert y-y^{\prime}\right\vert \right)
\right]  \int_{0}^{1}\mu^{\left(  1\right)  }\left(  dx^{\prime\prime}\right)
\left[  p\left(  \left\vert x^{\prime}-x^{\prime\prime}\right\vert \right)
\right]  }\right)  \;,\;\varphi\in BM\left(  \left[  0,1\right]  ^{2}\right)
\ .\nonumber
\end{gather}
the pushforward of $\mu$ under the map
\begin{align}
\left[  0,1\right]  ^{2}  &  \ni\left(  x,y\right)  \longmapsto\Theta_{\mu
}\left(  x,y\right)  :=\left(  \frac{\int_{\left[  0,1\right]  ^{2}}\mu\left(
dx^{\prime},dy^{\prime}\right)  p\left(  \left\vert x-x^{\prime}\right\vert
\right)  y^{\prime}}{\int_{0}^{1}\mu^{\left(  1\right)  }\left(  dx^{\prime
}\right)  p\left(  \left\vert x-x^{\prime}\right\vert \right)  },\right.
\label{theta_mu}\\
&  \left.  \int_{\left[  0,1\right]  ^{2}}\mu\left(  dx^{\prime},dy^{\prime
}\right)  \frac{p\left(  \left\vert y-y^{\prime}\right\vert \right)
\int_{\left[  0,1\right]  ^{2}}\mu\left(  dx^{\prime\prime},dy^{\prime\prime
}\right)  p\left(  \left\vert x^{\prime}-x^{\prime\prime}\right\vert \right)
y^{\prime\prime}}{\int_{0}^{1}\mu^{\left(  2\right)  }\left(  dy^{\prime
}\right)  \left[  p\left(  \left\vert y-y^{\prime}\right\vert \right)
\right]  \int_{0}^{1}\mu^{\left(  1\right)  }\left(  dx^{\prime\prime}\right)
\left[  p\left(  \left\vert x^{\prime}-x^{\prime\prime}\right\vert \right)
\right]  }\right)  \in\left[  0,1\right]  ^{2}\nonumber
\end{align}
defining the transition probability kernel which describes the process
$\left\{  \mathbf{Z}_{t}\right\}  _{t\geq0}.$

\begin{remark}
\label{mf3}We stress that, if the elements of the sequence $\left\{  \mu
_{N}^{X^{\left(  1\right)  },X^{\left(  2\right)  }}\right\}  _{N\geq1}$ are
supported in $\left[  0,1\right]  ^{2},$ since the support of a measure is
stable under weak limits, $\mu$ is also supported on $\left[  0,1\right]
^{2}.$ Moreover, since by (\ref{theta_mu}) and (\ref{Theta}) $\Theta
_{\#}\mathfrak{P}\left(  \left[  0,1\right]  ^{2}\right)  \subseteq
\mathfrak{P}\left(  \left[  0,1\right]  ^{2}\right)  ,$ considering $\left\{
\mu_{N}^{X^{\left(  1\right)  },X^{\left(  2\right)  }}\right\}  _{N\geq
1}\subset\mathfrak{P}\left(  \left[  0,1\right]  ^{2}\right)  $ as a sequence
in $\mathfrak{P}\left(  \mathbb{R}^{2}\right)  ,$ by the L\'{e}vy's continuity
theorem, Lemma \ref{mf1} implies that if $\left\{  \mu_{N}^{X^{\left(
1\right)  },X^{\left(  2\right)  }}\right\}  _{N\geq1}$ weakly converges to
$\mu,\left\{  \mathbf{\tilde{T}}_{\#}\mu_{N}^{X^{\left(  1\right)
},X^{\left(  2\right)  }}\right\}  _{N\geq1}$ weakly converges to $\Theta
_{\#}\mu\in\mathfrak{P}\left(  \left[  0,1\right]  ^{2}\right)  .$

In this case, as already pointed out at the beginning of Section 4 of
\cite{GO}, for (\ref{CFT}) to be defined one needs to assume that
$\forall i=1,2, \mathsf{supp}\mu^{\left(  i\right)  }\cap\mathsf{supp}p$ is not empty
and this will be so, without any restriction on $\mu,$ under the Assumption
\ref{a1} given below, which represents a sufficient condition for the
characteristic function of $\Theta_{\#}\mu$ to be well defined.
\end{remark}

\begin{assumption}
\label{a1}$\mathsf{supp}p=\left[  0,1\right]  .$
\end{assumption}

Notice that, by (\ref{Tomega1}), if, for any $N\geq2,\mu_{N}^{X^{\left(
1\right)  },X^{\left(  2\right)  }}\in\mathfrak{P}\left(  \left[  0,1\right]
^{2}\right)  ,$ then $\mathbf{T}_{\#}\mu_{N}^{X^{\left(  1\right)
},X^{\left(  2\right)  }}\in\mathfrak{P}\left(  \left[  0,1\right]
^{2}\right)  .$

\begin{theorem}
\label{mf2}Under Assumption \ref{a1}, if the sequence $\left\{  \mu
_{N}^{X^{\left(  1\right)  },X^{\left(  2\right)  }}\right\}  _{N\geq1}%
\subset\mathfrak{P}\left(  \left[  0,1\right]  ^{2}\right)  $ weakly converges
to $\mu,$ the sequence of measures $\left\{  \mathbf{T}_{\#}\mu_{N}%
^{X^{\left(  1\right)  },X^{\left(  2\right)  }}\right\}  _{N\geq1}$ weakly
converges to $\Theta_{\#}\mu.$
\end{theorem}

\begin{proof}
For any $\varepsilon>0,$ let us set
\begin{align}
\mathcal{A}_{v,N}^{\left(  1\right)  }\left(  \varepsilon\right)   &
:=\left\{  \omega\in\Omega_{N\mu}:\left\vert \frac{1}{N}\sum_{u=1}^{N}%
\omega_{v,u}-\mu_{N}^{X^{\left(  1\right)  }}\left[  p\left(  \left\vert
X_{v}^{\left(  1\right)  }-\cdot\right\vert \right)  \right]  \right\vert
\leq\varepsilon\right\}  \ ,\\
\mathcal{A}_{v,N}^{\left(  2\right)  }\left(  \varepsilon\right)   &
:=\left\{  \omega\in\Omega_{N}:\left\vert \frac{1}{N}\sum_{u=1}^{N}%
\omega_{v,u}-\mu_{N}^{X^{\left(  2\right)  }}\left[  p\left(  \left\vert
X_{v}^{\left(  2\right)  }-\cdot\right\vert \right)  \right]  \right\vert
\leq\varepsilon\right\}
\end{align}
and $\mathcal{A}_{N}^{\left(  i\right)  }\left(  \varepsilon\right)
:=\bigcap\limits_{v\in\mathbf{V}_{N}}\mathcal{A}_{v,N}^{\left(  i\right)
}\left(  \varepsilon\right)  ,i=1,2.$

By the Chernoff bound for Bernoulli r.v.'s, for any $i=1,2,$ we get
\begin{gather}
\sum_{\omega\in\Omega_{N}}\Pi\left(  \omega|X^{\left(  i\right)  }\right)
\mathbf{1}_{\left(  \mathcal{A}_{N}^{\left(  i\right)  }\left(  \varepsilon
\right)  \right)  ^{c}}\left(  \omega\right)  \leq\sum_{v=1}^{N}\sum
_{\omega\in\Omega_{N}}\Pi\left(  \omega|X^{\left(  i\right)  }\right)
\mathbf{1}_{\left(  \mathcal{A}_{v,N}^{\left(  i\right)  }\left(
\varepsilon\right)  \right)  ^{c}}\left(  \omega\right) \\
\leq2N\int_{\left[  0,1\right]  }\mu_{N}^{X^{\left(  i\right)  }}\left(
dx\right)  \exp\left[  -N\int_{\left[  0,1\right]  }\mu_{N}^{X^{\left(
i\right)  }}\left(  dy\right)  \left(  \mathcal{H}\left(  p\left(  \left\vert
x-y\right\vert \right)  +\varepsilon\right)  \wedge\mathcal{H}\left(  p\left(
\left\vert x-y\right\vert \right)  -\varepsilon\right)  \right)  \right]
\ ,\nonumber
\end{gather}
where, for any $\left(  x,y\right)  \in\left[  0,1\right]  ^{2},\mathcal{H}%
\left(  p\left(  \left\vert x-y\right\vert \right)  \pm\varepsilon\right)  $
is the relative entropy (Kullback-Leibler divergence) of a Bernoulli
distribution with parameter $p\left(  \left\vert x-y\right\vert \right)
\pm\varepsilon$ w.r.t. a Bernoulli distribution with parameter $p\left(
\left\vert x-y\right\vert \right)  ,$ which for $\varepsilon$ sufficiently
small gives
\begin{equation}
\sum_{\omega\in\Omega_{N}}\Pi\left(  \omega|X^{\left(  i\right)  }\right)
\mathbf{1}_{\left(  \mathcal{A}_{N}^{\left(  i\right)  }\left(  \varepsilon
\right)  \right)  ^{c}}\left(  \omega\right)  \leq2Ne^{-2N\varepsilon^{2}}\ .
\label{LDe}%
\end{equation}

Considering $\left\{  \mu_{N}^{X^{\left(  1\right)  },X^{\left(  2\right)  }%
}\right\}  _{N\geq1}$ as a sequence in $\mathfrak{P}\left(  \mathbb{R}%
^{2}\right)  ,$ from (\ref{TTT}) and (\ref{TTTtilda}), for any $\varphi\in
Lip\left(  \mathbb{R}^{2}\right)  $ and any sufficiently small $\varepsilon>0$
we get
\begin{gather}
\left\vert \mathbf{T}_{\#}\mu_{N}^{X^{\left(  1\right)  },X^{\left(  2\right)
}}\left[  \varphi\right]  -\mathbf{\tilde{T}}_{\#}\mu_{N}^{X^{\left(
1\right)  },X^{\left(  2\right)  }}\left[  \varphi\right]  \right\vert
\leq\frac{1}{N}\sum_{v=1}^{N}\sum_{\omega,\omega^{\prime}\in\Omega_{N}}%
\Pi\left(  \omega|X^{\left(  1\right)  }\right)  \Pi\left(  \omega^{\prime
}|X^{\left(  2\right)  }\right)  \times\\
\times\left\vert \varphi\left(  \frac{1}{N}\sum_{u=1}^{N}\frac{\omega_{v,u}%
}{\frac{1}{N}\sum_{u=1}^{N}\omega_{v,u}}X_{u}^{\left(  2\right)  },\frac{1}%
{N}\sum_{u=1}^{N}\frac{\omega_{v,u}^{\prime}}{\frac{1}{N}\sum_{u=1}^{N}%
\omega_{v,u}^{\prime}}\frac{1}{N}\sum_{w=1}^{N}\frac{\omega_{u,w}}{\frac{1}%
{N}\sum_{u=1}^{N}\omega_{u,w}}X_{w}^{\left(  2\right)  }\right)  -\right.
\nonumber\\
\left.  \varphi\left(  \frac{1}{N}\sum_{u=1}^{N}\frac{\omega_{v,u}}{\mu
_{N}^{X^{\left(  1\right)  }}\left[  p\left(  \left\vert X_{v}^{\left(
1\right)  }-\cdot\right\vert \right)  \right]  }X_{u}^{\left(  2\right)
},\frac{1}{N}\sum_{u=1}^{N}\frac{\omega_{v,u}^{\prime}}{\mu_{N}^{X^{\left(
2\right)  }}\left[  p\left(  \left\vert X_{v}^{\left(  2\right)  }%
-\cdot\right\vert \right)  \right]  }\frac{1}{N}\sum_{w=1}^{N}\frac
{\omega_{u,w}}{\mu_{N}^{X^{\left(  2\right)  }}\left[  p\left(  \left\vert
X_{u}^{\left(  1\right)  }-\cdot\right\vert \right)  \right]  }X_{w}^{\left(
2\right)  }\right)  \right\vert \times\nonumber\\
\times\left[  \mathbf{1}_{\mathcal{A}_{N}^{\left(  1\right)  }\left(
\varepsilon\right)  }\left(  \omega\right)  \mathbf{1}_{\mathcal{A}%
_{N}^{\left(  2\right)  }\left(  \varepsilon\right)  }\left(  \omega^{\prime
}\right)  +\mathbf{1}_{\left(  \mathcal{A}_{N}^{\left(  1\right)  }\left(
\varepsilon\right)  \right)  ^{c}}\left(  \omega\right)  +\mathbf{1}_{\left(
\mathcal{A}_{N}^{\left(  2\right)  }\left(  \varepsilon\right)  \right)  ^{c}%
}\left(  \omega^{\prime}\right)  \right]  \ .\nonumber
\end{gather}
From (\ref{LDe}) we have
\begin{gather}
\frac{1}{N}\sum_{v=1}^{N}\sum_{\omega,\omega^{\prime}\in\Omega_{N}}\Pi\left(
\omega|X^{\left(  1\right)  }\right)  \Pi\left(  \omega^{\prime}|X^{\left(
2\right)  }\right)  \times\label{ld1}\\
\times\left\vert \varphi\left(  \frac{1}{N}\sum_{u=1}^{N}\frac{\omega_{v,u}%
}{\frac{1}{N}\sum_{u=1}^{N}\omega_{v,u}}X_{u}^{\left(  2\right)  },\frac{1}%
{N}\sum_{u=1}^{N}\frac{\omega_{v,u}^{\prime}}{\frac{1}{N}\sum_{u=1}^{N}%
\omega_{v,u}^{\prime}}\frac{1}{N}\sum_{w=1}^{N}\frac{\omega_{u,w}}{\frac{1}%
{N}\sum_{u=1}^{N}\omega_{u,w}}X_{w}^{\left(  2\right)  }\right)  -\right.
\nonumber\\
\left.  \varphi\left(  \frac{1}{N}\sum_{u=1}^{N}\frac{\omega_{v,u}}{\mu
_{N}^{X^{\left(  1\right)  }}\left[  p\left(  \left\vert X_{v}^{\left(
1\right)  }-\cdot\right\vert \right)  \right]  }X_{u}^{\left(  2\right)
},\frac{1}{N}\sum_{u=1}^{N}\frac{\omega_{v,u}^{\prime}}{\mu_{N}^{X^{\left(
2\right)  }}\left[  p\left(  \left\vert X_{v}^{\left(  2\right)  }%
-\cdot\right\vert \right)  \right]  }\frac{1}{N}\sum_{w=1}^{N}\frac
{\omega_{u,w}}{\mu_{N}^{X^{\left(  1\right)  }}\left[  p\left(  \left\vert
X_{u}^{\left(  1\right)  }-\cdot\right\vert \right)  \right]  }X_{w}^{\left(
2\right)  }\right)  \right\vert \times\nonumber\\
\times\left[  \mathbf{1}_{\left(  \mathcal{A}_{N}^{\left(  1\right)  }\left(
\varepsilon\right)  \right)  ^{c}}\left(  \omega\right)  +\mathbf{1}_{\left(
\mathcal{A}_{N}^{\left(  2\right)  }\left(  \varepsilon\right)  \right)  ^{c}%
}\left(  \omega^{\prime}\right)  \right]  \leq8\left\Vert \varphi\right\Vert
_{Lip}Ne^{-2N\varepsilon^{2}}\ ,\nonumber
\end{gather}
while, setting
\begin{gather}
D_{N}^{\varepsilon}\left(  \varphi\right)  :=\frac{1}{N}\sum_{v=1}^{N}%
\sum_{\omega,\omega^{\prime}\in\Omega_{N}}\Pi\left(  \omega|X^{\left(
1\right)  }\right)  \Pi\left(  \omega^{\prime}|X^{\left(  2\right)  }\right)
\mathbf{1}_{\mathcal{A}_{N}^{\left(  1\right)  }\left(  \varepsilon\right)
}\left(  \omega\right)  \mathbf{1}_{\mathcal{A}_{N}^{\left(  2\right)
}\left(  \varepsilon\right)  }\left(  \omega^{\prime}\right)  \times\\
\times\left\vert \varphi\left(  \frac{1}{N}\sum_{u=1}^{N}\frac{\omega_{v,u}%
}{\frac{1}{N}\sum_{u=1}^{N}\omega_{v,u}}X_{u}^{\left(  2\right)  },\frac{1}%
{N}\sum_{u=1}^{N}\frac{\omega_{v,u}^{\prime}}{\frac{1}{N}\sum_{u=1}^{N}%
\omega_{v,u}^{\prime}}\frac{1}{N}\sum_{w=1}^{N}\frac{\omega_{u,w}}{\frac{1}%
{N}\sum_{u=1}^{N}\omega_{u,w}}X_{w}^{\left(  2\right)  }\right)  -\right.
\nonumber\\
\left.  \varphi\left(  \frac{1}{N}\sum_{u=1}^{N}\frac{\omega_{v,u}}{\mu
_{N}^{X^{\left(  1\right)  }}\left[  p\left(  \left\vert X_{v}^{\left(
1\right)  }-\cdot\right\vert \right)  \right]  }X_{u}^{\left(  2\right)
},\frac{1}{N}\sum_{u=1}^{N}\frac{\omega_{v,u}^{\prime}}{\mu_{N}^{X^{\left(
1\right)  }}\left[  p\left(  \left\vert X_{v}^{\left(  2\right)  }%
-\cdot\right\vert \right)  \right]  }\frac{1}{N}\sum_{w=1}^{N}\frac
{\omega_{u,w}}{\mu_{N}^{X^{\left(  2\right)  }}\left[  p\left(  \left\vert
X_{u}^{\left(  1\right)  }-\cdot\right\vert \right)  \right]  }X_{w}^{\left(
2\right)  }\right)  \right\vert \nonumber
\end{gather}
we obtain
\begin{gather}
D_{N}^{\varepsilon}\left(  \varphi\right)  \leq\left\Vert \varphi\right\Vert
_{Lip}\frac{1}{N}\sum_{v=1}^{N}\sum_{\omega,\omega^{\prime}\in\Omega_{N}}%
\Pi\left(  \omega|X^{\left(  1\right)  }\right)  \Pi\left(  \omega^{\prime
}|X^{\left(  2\right)  }\right)  \times\\
\times\mathbf{1}_{\mathcal{A}_{N}^{\left(  1\right)  }\left(  \varepsilon
\right)  }\left(  \omega\right)  \mathbf{1}_{\mathcal{A}_{N}^{\left(
2\right)  }\left(  \varepsilon\right)  }\left(  \omega^{\prime}\right)
\left[  \frac{1}{N}\sum_{u=1}^{N}\omega_{v,u}X_{u}^{\left(  2\right)
}\left\vert \frac{1}{\frac{1}{N}\sum_{u=1}^{N}\omega_{v,u}}-\frac{1}{\mu
_{N}^{X^{\left(  1\right)  }}\left[  p\left(  \left\vert X_{v}^{\left(
1\right)  }-\cdot\right\vert \right)  \right]  }\right\vert +\right.
\nonumber\\
+\frac{1}{N}\sum_{u=1}^{N}\omega_{v,u}^{\prime}\sum_{w=1}^{N}\frac{1}{N}%
\omega_{u,w}X_{w}^{\left(  2\right)  }\times\nonumber\\
\left.  \times\left\vert \frac{1}{\frac{1}{N}\sum_{u=1}^{N}\omega
_{v,u}^{\prime}\frac{1}{N}\sum_{w=1}^{N}\omega_{u,w}}-\frac{1}{\mu
_{N}^{X^{\left(  2\right)  }}\left[  p\left(  \left\vert X_{v}^{\left(
2\right)  }-\cdot\right\vert \right)  \right]  \mu_{N}^{X^{\left(  1\right)
}}\left[  p\left(  \left\vert X_{u}^{\left(  1\right)  }-\cdot\right\vert
\right)  \right]  }\right\vert \right] \nonumber\\
=\left\Vert \varphi\right\Vert _{Lip}\frac{1}{N}\sum_{v=1}^{N}\sum
_{\omega,\omega^{\prime}\in\Omega_{N}}\Pi\left(  \omega|X^{\left(  1\right)
}\right)  \Pi\left(  \omega^{\prime}|X^{\left(  2\right)  }\right)
\mathbf{1}_{\mathcal{A}_{N}^{\left(  1\right)  }\left(  \varepsilon\right)
}\left(  \omega\right)  \mathbf{1}_{\mathcal{A}_{N}^{\left(  2\right)
}\left(  \varepsilon\right)  }\left(  \omega^{\prime}\right)  \times
\nonumber\\
\times\left[  \frac{1}{N}\sum_{u=1}^{N}\frac{\omega_{v,u}}{\frac{1}{N}%
\sum_{u=1}^{N}\omega_{v,u}}X_{u}^{\left(  2\right)  }\frac{\left\vert \frac
{1}{N}\sum_{u=1}^{N}\omega_{v,u}-\mu_{N}^{X^{\left(  1\right)  }}\left[
p\left(  \left\vert X_{v}^{\left(  1\right)  }-\cdot\right\vert \right)
\right]  \right\vert }{\mu_{N}^{X^{\left(  1\right)  }}\left[  p\left(
\left\vert X_{v}^{\left(  1\right)  }-\cdot\right\vert \right)  \right]
}+\right. \nonumber\\
+\frac{1}{N}\sum_{u=1}^{N}\frac{\omega_{v,u}^{\prime}}{\frac{1}{N}\sum
_{u=1}^{N}\omega_{v,u}^{\prime}}\frac{1}{N}\sum_{w=1}^{N}\frac{\omega_{u,w}%
}{\frac{1}{N}\sum_{w=1}^{N}\omega_{u,w}}X_{w}^{\left(  2\right)  }%
\times\nonumber\\
\times\left.  \frac{\left\vert \frac{1}{N}\sum_{u=1}^{N}\omega_{v,u}^{\prime
}\frac{1}{N}\sum_{w=1}^{N}\omega_{u,w}-\mu_{N}^{X^{\left(  2\right)  }}\left[
p\left(  \left\vert X_{v}^{\left(  2\right)  }-\cdot\right\vert \right)
\right]  \mu_{N}^{X^{\left(  1\right)  }}\left[  p\left(  \left\vert
X_{u}^{\left(  1\right)  }-\cdot\right\vert \right)  \right]  \right\vert
}{\mu_{N}^{X^{\left(  2\right)  }}\left[  p\left(  \left\vert X_{v}^{\left(
2\right)  }-\cdot\right\vert \right)  \right]  \mu_{N}^{X^{\left(  1\right)
}}\left[  p\left(  \left\vert X_{u}^{\left(  1\right)  }-\cdot\right\vert
\right)  \right]  }\right]  \ .\nonumber
\end{gather}
Since $\frac{1}{N}\sum_{u=1}^{N}\frac{\omega_{v,u}}{\frac{1}{N}\sum_{u=1}%
^{N}\omega_{v,u}}X_{u}\leq1$ uniformly in $\omega\in\Omega_{N},X\in\left[
0,1\right]  ^{N}$ and
\begin{gather}
\left\vert \frac{1}{N}\sum_{u=1}^{N}\omega_{v,u}^{\prime}\frac{1}{N}\sum
_{w=1}^{N}\omega_{u,w}-\mu_{N}^{X^{\left(  2\right)  }}\left[  p\left(
\left\vert X_{v}^{\left(  2\right)  }-\cdot\right\vert \right)  \right]
\mu_{N}^{X^{\left(  1\right)  }}\left[  p\left(  \left\vert X_{u}^{\left(
1\right)  }-\cdot\right\vert \right)  \right]  \right\vert \\
\leq\frac{1}{N}\sum_{u=1}^{N}\omega_{v,u}^{\prime}\left\vert \frac{1}{N}%
\sum_{w=1}^{N}\omega_{u,w}-\mu_{N}^{X^{\left(  1\right)  }}\left[  p\left(
\left\vert X_{u}^{\left(  1\right)  }-\cdot\right\vert \right)  \right]
\right\vert +\nonumber\\
+\mu_{N}^{X^{\left(  1\right)  }}\left[  p\left(  \left\vert X_{u}^{\left(
1\right)  }-\cdot\right\vert \right)  \right]  \left\vert \frac{1}{N}%
\sum_{u=1}^{N}\omega_{v,u}^{\prime}-\mu_{N}^{X^{\left(  2\right)  }}\left[
p\left(  \left\vert X_{v}^{\left(  2\right)  }-\cdot\right\vert \right)
\right]  \right\vert \ ,\nonumber
\end{gather}
we have
\begin{align}
D_{N}^{\varepsilon}\left(  \varphi\right)   &  \leq\left\Vert \varphi
\right\Vert _{Lip}\frac{1}{N}\sum_{v=1}^{N}\sum_{\omega,\omega^{\prime}%
\in\Omega_{N}}\Pi\left(  \omega|X^{\left(  1\right)  }\right)  \Pi\left(
\omega^{\prime}|X^{\left(  2\right)  }\right)  \mathbf{1}_{\mathcal{A}%
_{N}^{\left(  1\right)  }\left(  \varepsilon\right)  }\left(  \omega\right)
\mathbf{1}_{\mathcal{A}_{N}^{\left(  2\right)  }\left(  \varepsilon\right)
}\left(  \omega^{\prime}\right)  \times\\
&  \times\left[  \frac{\left\vert \frac{1}{N}\sum_{u=1}^{N}\omega_{v,u}%
-\mu_{N}^{X^{\left(  1\right)  }}\left[  p\left(  \left\vert X_{v}^{\left(
1\right)  }-\cdot\right\vert \right)  \right]  \right\vert }{\mu
_{N}^{X^{\left(  1\right)  }}\left[  p\left(  \left\vert X_{v}^{\left(
1\right)  }-\cdot\right\vert \right)  \right]  }+\right. \nonumber\\
&  +\frac{1}{N}\sum_{u=1}^{N}\frac{\omega_{v,u}^{\prime}}{\mu_{N}^{X^{\left(
1\right)  }}\left[  p\left(  \left\vert X_{u}^{\left(  1\right)  }%
-\cdot\right\vert \right)  \right]  }\frac{\left\vert \frac{1}{N}\sum
_{w=1}^{N}\omega_{u,w}-\mu_{N}^{X^{\left(  1\right)  }}\left[  p\left(
\left\vert X_{u}^{\left(  1\right)  }-\cdot\right\vert \right)  \right]
\right\vert }{\mu_{N}^{X^{\left(  2\right)  }}\left[  p\left(  \left\vert
X_{v}^{\left(  2\right)  }-\cdot\right\vert \right)  \right]  }+\nonumber\\
&  \left.  +\frac{1}{N}\sum_{u=1}^{N}\frac{\omega_{v,u}^{\prime}}{\frac{1}%
{N}\sum_{u=1}^{N}\omega_{v,u}^{\prime}}\frac{\left\vert \frac{1}{N}\sum
_{u=1}^{N}\omega_{v,u}^{\prime}-\mu_{N}^{X^{\left(  2\right)  }}\left[
p\left(  \left\vert X_{v}^{\left(  2\right)  }-\cdot\right\vert \right)
\right]  \right\vert }{\mu_{N}^{X^{\left(  2\right)  }}\left[  p\left(
\left\vert X_{v}^{\left(  2\right)  }-\cdot\right\vert \right)  \right]
}\right] \nonumber\\
&  \leq\left\Vert \varphi\right\Vert _{Lip}\left[  \int_{\left[  0,1\right]
}\mu_{N}^{X^{\left(  1\right)  }}\left(  dx\right)  \frac{\varepsilon}{\mu
_{N}^{X^{\left(  1\right)  }}\left[  p\left(  \left\vert x-\cdot\right\vert
\right)  \right]  }+\right. \nonumber\\
&  +\int_{\left[  0,1\right]  }\mu_{N}^{X^{\left(  2\right)  }}\left(
dx\right)  \int_{\left[  0,1\right]  ^{2}}\mu_{N}^{X^{\left(  1\right)
},X^{\left(  2\right)  }}\left(  dy,dz\right)  \frac{p\left(  \left\vert
x-z\right\vert \right)  }{\mu_{N}^{X^{\left(  2\right)  }}\left[  p\left(
\left\vert y-\cdot\right\vert \right)  \right]  }\frac{\varepsilon}{\mu
_{N}^{X^{\left(  1\right)  }}\left[  p\left(  \left\vert x-\cdot\right\vert
\right)  \right]  }+\nonumber\\
&  \left.  +\frac{1}{N}\sum_{v=1}^{N}\frac{\mu_{N}^{X^{\left(  2\right)  }%
}\left[  p\left(  \left\vert X_{v}^{\left(  2\right)  }-\cdot\right\vert
\right)  \right]  +\varepsilon}{\mu_{N}^{X^{\left(  2\right)  }}\left[
p\left(  \left\vert X_{v}^{\left(  2\right)  }-\cdot\right\vert \right)
\right]  -\varepsilon}\frac{\varepsilon}{\mu_{N}^{X^{\left(  2\right)  }%
}\left[  p\left(  \left\vert X_{v}^{\left(  2\right)  }-\cdot\right\vert
\right)  \right]  }\right]  \ .\nonumber
\end{align}
Therefore
\begin{align}
D_{N}^{\varepsilon}\left(  \varphi\right)   &  \leq\varepsilon\left\Vert
\varphi\right\Vert _{Lip}\left[  \int_{\left[  0,1\right]  }\mu_{N}%
^{X^{\left(  1\right)  }}\left(  dx\right)  \frac{1}{\mu_{N}^{X^{\left(
1\right)  }}\left[  p\left(  \left\vert x-\cdot\right\vert \right)  \right]
}+\right. \\
&  +\int_{\left[  0,1\right]  }\mu_{N}^{X^{\left(  2\right)  }}\left(
dx\right)  \int_{\left[  0,1\right]  ^{2}}\mu_{N}^{X^{\left(  1\right)
},X^{\left(  2\right)  }}\left(  dy,dz\right)  \frac{p\left(  \left\vert
x-z\right\vert \right)  }{\mu_{N}^{X^{\left(  2\right)  }}\left[  p\left(
\left\vert y-\cdot\right\vert \right)  \right]  \mu_{N}^{X^{\left(  1\right)
}}\left[  p\left(  \left\vert x-\cdot\right\vert \right)  \right]
}+\nonumber\\
&  \left.  +\int_{\left[  0,1\right]  }\mu_{N}^{X^{\left(  2\right)  }}\left(
dy\right)  \frac{1}{\mu_{N}^{X^{\left(  2\right)  }}\left[  p\left(
\left\vert y-\cdot\right\vert \right)  \right]  }\frac{\mu_{N}^{X^{\left(
2\right)  }}\left[  p\left(  \left\vert y-\cdot\right\vert \right)  \right]
+\varepsilon}{\mu_{N}^{X^{\left(  2\right)  }}\left[  p\left(  \left\vert
y-\cdot\right\vert \right)  \right]  -\varepsilon}\right]  \ .\nonumber
\end{align}

Hence, since $\left\{  \mu_{N}^{X^{\left(  1\right)  },X^{\left(  2\right)  }%
}\right\}  _{N\geq1}$ weakly converges to $\mu$ the limit as $N\uparrow
\infty,$ there exists $C:=C\left(  p,\mu\right)  >0$ such that $\mathfrak{d}%
\left(  \mathbf{T}_{\#}\mu_{N}^{X^{\left(  1\right)  },X^{\left(  2\right)  }%
},\mathbf{\tilde{T}}_{\#}\mu_{N}^{X^{\left(  1\right)  },X^{\left(  2\right)
}}\right)  \leq C\varepsilon$ for any arbitrary choice of $\varepsilon>0.$
Moreover,
\begin{equation}
\mathfrak{d}\left(  \mathbf{T}_{\#}\mu_{N}^{X^{\left(  1\right)  },X^{\left(
2\right)  }},\Theta_{\#}\mu\right)  \leq\mathfrak{d}\left(  \mathbf{T}_{\#}%
\mu_{N}^{X^{\left(  1\right)  },X^{\left(  2\right)  }},\mathbf{\tilde{T}%
}_{\#}\mu_{N}^{X^{\left(  1\right)  },X^{\left(  2\right)  }}\right)
+\mathfrak{d}\left(  \mathbf{\tilde{T}}_{\#}\mu_{N}^{X^{\left(  1\right)
},X^{\left(  2\right)  }},\Theta_{\#}\mu\right)
\end{equation}
and since, by the previous Lemma, $\lim_{N\rightarrow\infty}\mathfrak{d}%
\left(  \mathbf{\tilde{T}}_{\#}\mu_{N}^{X^{\left(  1\right)  },X^{\left(
2\right)  }},\Theta_{\#}\mu\right)  =0,$ the thesis follows.
\end{proof}

In the more general case where 
\begin{equation}
r_{i,j}:=\varrho\left(  X_{j},X_{i}\right)  \;,\qquad\left(  i,j\right)
\in\mathbf{V}_{N}\times\mathbf{V}_{N}=:\mathbf{E}_{N}\label{r=rho}%
\end{equation} 
with $\varrho\left(  X_{i},X_{i}\right)=1 ,$ besides Assumption \ref{a1} let us also assume that:

\begin{assumption}
\label{a2}$\mathsf{supp}\varrho=\left[  0,1\right]  ^{2}$ and $\varrho\in
Lip\left(  \left[  0,1\right]  ^{2},\left[  0,1\right]  \right)  .$
\end{assumption}

Given $v\in\mathbf{V}_{N},$ (\ref{Tomega1}) rewrites as
\begin{equation}
\left(  \Xi_{N},\Omega_{N}\right)  \ni\left(  X,\omega\right)  \longmapsto
\mathcal{T}_{v}\left(  X,\omega\right)  :=\frac{\sum_{u\in\mathbf{V}}%
\omega_{v,u}\varrho\left(  X_{v},X_{u}\right)  X_{u}}{\sum_{u\in\mathbf{V}%
}\omega_{v,u}\varrho\left(  X_{v},X_{u}\right)  }\in\mathbb{R}_{+}\ .
\end{equation}
Setting $\Xi_{N}\ni X\longmapsto\mathcal{\breve{T}}\left(  X,\omega\right)
\in\Xi_{N},$ where, for any $v\in\mathbf{V}_{N},$%
\begin{equation}
\left(  \Xi_{N},\Omega_{N}\right)  \ni\left(  X,\omega\right)  \longmapsto
\mathcal{\breve{T}}_{v}\left(  X,\omega\right)  :=\frac{\frac{1}{N}\sum
_{u\in\mathbf{V}}\omega_{v,u}\varrho\left(  X_{v},X_{u}\right)  X_{u}}{\mu
_{N}^{X}\left[  p\left(  \left\vert X_{v}-\cdot\right\vert \right)
\varrho\left(  X_{v},\cdot\right)  \right]  }\in\mathbb{R}_{+}%
\end{equation}
and consequently denoting by $\mathbf{\breve{T}}$ the operator on $BM\left(
\Xi_{N}^{2}\right)  $ defined as in (\ref{TTTtilda}), namely
\begin{equation}
\left(  \mathbf{\breve{T}}\varphi\right)  \left(  X^{\left(  1\right)
},X^{\left(  2\right)  }\right)  :=\sum_{\omega,\omega^{\prime}\in\Omega_{N}%
}\varphi\left(  \mathcal{\breve{T}}\left(  X^{\left(  2\right)  }%
,\omega\right)  ,\mathcal{\breve{T}}\left(  \mathcal{\breve{T}}\left(
X^{\left(  2\right)  },\omega\right)  ,\omega^{\prime}\right)  \right)
\Pi\left(  \omega|X^{\left(  1\right)  }\right)  \Pi\left(  \omega^{\prime
}|X^{\left(  2\right)  }\right)  \ , \label{TTTv}%
\end{equation}
the proof of the following results rely on the same strategy which led to the
proofs of Lemma \ref{mf1} and Theorem \ref{mf2}, although in this last case
the details are more involved.

Let us set, for any $\mu\in\mathfrak{P}\left(  \left[  0,1\right]
^{2}\right)  $%
\begin{equation}
\left[  0,1\right]  ^{2}\ni\left(  x,y\right)  \longmapsto\Theta_{\mu
}^{\varrho}\left(  x,y\right) =(\Theta_{\mu}^{\varrho,(1)}(x,y),\Theta_{\mu}^{\varrho,(2)}(x,y))
:=\left(  \theta_{\mu}^{\varrho}\left(
x,y\right)  ,\theta_{\mu}^{\varrho}\circ\theta_{\mu}^{\varrho}\left(
x,y\right)  \right)  \in\left[  0,1\right]  ^{2}\ ,
\end{equation}
where, $\theta_{\mu}^{\varrho}:\left[  0,1\right]  ^{2}\longrightarrow\left[
0,1\right]  $ is given in (\ref{theta}), namely
\begin{equation}
\left[  0,1\right]  ^{2}\ni\left(  x,y\right)  \longmapsto\theta_{\mu
}^{\varrho}\left(  x,y\right)  :=\frac{\int_{\left[  0,1\right]  ^{2}}%
\mu\left(  dx^{\prime},dy^{\prime}\right)  p\left(  \left\vert x-x^{\prime
}\right\vert \right)  \varrho\left(  y,y^{\prime}\right)  y^{\prime}}%
{\int_{\left[  0,1\right]  ^{2}}\mu\left(  dx^{\prime},dy^{\prime}\right)
p\left(  \left\vert x-x^{\prime}\right\vert \right)  \varrho\left(
y,y^{\prime}\right)  }\in\left[  0,1\right]  \ .
\end{equation}

\begin{lemma}
Under the Assumption \ref{a1}, if the sequence $\left\{  \mu_{N}^{X^{\left(
1\right)  },X^{\left(  2\right)  }}\right\}  _{N\geq1}\subset\mathfrak{P}%
\left(  \left[  0,1\right]  ^{2}\right)  $ weakly converges to $\mu,$ then the
sequence of measures $\left\{  \mathbf{\breve{T}}_{\#}\mu_{N}^{X^{\left(
1\right)  },X^{\left(  2\right)  }}\right\}  _{N\geq1}$ weakly converges to
\begin{equation}
\Theta_{\#}^{\varrho}\mu\left[  \varphi\right]  :=\mu\left[  \varphi\left(
\Theta_{\mu}^{\varrho,\left(  1\right)  }\left(  \cdot,\cdot\right)
,\Theta_{\mu}^{\varrho,\left(  2\right)  }\left(  \cdot,\cdot\right)
\right)  \right].
\end{equation}

\end{lemma}

\begin{proof}
As already pointed out in Remark \ref{mf3}, by the L\'{e}vy's continuity
theorem its enough to show that
\begin{equation}
\mathbb{R}^{2}\ni\left(  \lambda_{1},\lambda_{2}\right)  =:\lambda
\longmapsto\mu\left[  e^{i\left\langle \lambda,\Theta_{\mu}^{\varrho}\left(
\cdot,\cdot\right)  \right\rangle }\right]  =\int_{\left[  0,1\right]  ^{2}%
}\Theta_{\#}^{\varrho}\mu\left(  dx,dy\right)  \left[  e^{i\left(  \lambda
_{1}x+\lambda_{2}y\right)  }\right]  \in\mathbb{C}%
\end{equation}
is the pointwise limit of the sequence of the characteristic functions of the
elements of the sequence $\left\{  \mathbf{\breve{T}}_{\#}\mu_{N}^{X^{\left(
1\right)  },X^{\left(  2\right)  }}\right\}  _{N\geq1}.$ The proof of this result is
identical of that of Lemma \ref{mf1}\ and therefore omitted.
\end{proof}

\begin{theorem}
\label{mf4}Under the Assumptions \ref{a1} and \ref{a2}, if the sequence
$\left\{  \mu_{N}^{X^{\left(  1\right)  },X^{\left(  2\right)  }}\right\}
_{N\geq1}\subset\mathfrak{P}\left(  \left[  0,1\right]  ^{2}\right)  $ weakly
converges to $\mu,$ the sequence of measures $\left\{  \mathbf{T}_{\#}\mu
_{N}^{X^{\left(  1\right)  },X^{\left(  2\right)  }}\right\}  _{N\geq1}$
weakly converges to $\Theta_{\#}^{\varrho}\mu.$
\end{theorem}

\begin{proof}
For any $\varepsilon>0,$ let
\begin{align}
\mathcal{A}_{v,N}^{\left(  1\right)  ,\varrho}\left(  \varepsilon\right)   &
:=\left\{  \left(  \omega,\omega^{\prime}\right)  \in\Omega_{N}^{2}:\left\vert
\frac{1}{N}\sum_{u=1}^{N}\omega_{v,u}\varrho\left(  X_{v}^{\left(  2\right)
},X_{u}^{\left(  2\right)  }\right)  -\mu_{N}^{X^{\left(  1\right)
},X^{\left(  2\right)  }}\left[  p\left(  \left\vert X_{v}^{\left(  1\right)
}-\cdot\right\vert \right)  \varrho\left(  X_{v}^{\left(  2\right)  }%
,\cdot\right)  \right]  \right\vert \leq\varepsilon\right\}  \ ,\\
\mathcal{A}_{v,N}^{\left(  1\right)  ,\varrho}\left(  \varepsilon\right)   &
:=\left\{  \left(  \omega,\omega^{\prime}\right)  \in\Omega_{N}^{2}:\left\vert
\frac{1}{N}\sum_{u=1}^{N}\omega_{v,u}^{\prime}\varrho\left(  \mathcal{\breve
{T}}_{v}\left(  X^{\left(  2\right)  },\omega\right)  ,\mathcal{\breve{T}}%
_{u}\left(  X^{\left(  2\right)  },\omega\right)  \right)  -\right.  \right.
\\
&  \left.  \left.  \frac{1}{N}\sum_{u=1}^{N}p\left(  \left\vert X_{v}^{\left(
2\right)  }-X_{u}^{\left(  2\right)  }\right\vert \right)  \varrho\left(
\mathcal{\breve{T}}_{v}\left(  X^{\left(  2\right)  },\omega\right)
,\mathcal{\breve{T}}_{u}\left(  X^{\left(  2\right)  },\omega\right)  \right)
\right\vert \leq\varepsilon\right\} \nonumber
\end{align}
and $\mathcal{A}_{N}^{\left(  i\right)  ,\varrho}\left(  \varepsilon\right)
:=\bigcap\limits_{v\in\mathbf{V}_{N}}\mathcal{A}_{v,N}^{\left(  i\right)
,\varrho}\left(  \varepsilon\right)  ,i=1,2,$ as in the proof of Theorem
\ref{mf2}. From (\ref{TTT}) and (\ref{TTTv}), for any $\varphi\in Lip\left(
\mathbb{R}^{2}\right)  $ we get
\begin{gather}
\left\vert \mathbf{T}_{\#}\mu_{N}^{X^{\left(  1\right)  },X^{\left(  2\right)
}}\left[  \varphi\right]  -\mathbf{\breve{T}}_{\#}\mu_{N}^{X^{\left(
1\right)  },X^{\left(  2\right)  }}\left[  \varphi\right]  \right\vert
\leq\frac{1}{N}\sum_{v=1}^{N}\sum_{\omega,\omega^{\prime}\in\Omega_{N}}%
\Pi\left(  \omega|X^{\left(  1\right)  }\right)  \Pi\left(  \omega^{\prime
}|X^{\left(  2\right)  }\right)  \times\\
\times\left\vert \varphi\left(  \mathcal{T}_{v}\left(  X^{\left(  2\right)
},\omega\right)  ,\mathcal{T}_{v}\left(  \mathcal{T}\left(  X^{\left(
2\right)  },\omega\right)  ,\omega^{\prime}\right)  \right)  -\left(
\varphi\left(  \mathcal{\breve{T}}_{v}\left(  X^{\left(  2\right)  }%
,\omega\right)  ,\mathcal{\breve{T}}_{v}\left(  \mathcal{\breve{T}}\left(
X^{\left(  2\right)  },\omega\right)  ,\omega^{\prime}\right)  \right)
\right)  \right\vert \times\nonumber\\
\times\left[  \mathbf{1}_{\mathcal{A}_{N}^{\left(  1\right)  ,\varrho}\left(
\varepsilon\right)  \cap\mathcal{A}_{N}^{\left(  2\right)  ,\varrho}\left(
\varepsilon\right)  }\left(  \omega,\omega^{\prime}\right)  +\mathbf{1}%
_{\left(  \mathcal{A}_{N}^{\left(  1\right)  ,\varrho}\left(  \varepsilon
\right)  \right)  ^{c}}\left(  \omega,\omega^{\prime}\right)  +\mathbf{1}%
_{\left(  \mathcal{A}_{N}^{\left(  2\right)  ,\rho}\left(  \varepsilon\right)
\right)  ^{c}}\left(  \omega,\omega^{\prime}\right)  \right]  \ .\nonumber
\end{gather}
By the Hoeffding's inequality for the sums r.v.'s $\sum_{u=1}^{N}\omega
_{v,u}\varrho\left(  X_{v}^{\left(  2\right)  },X_{u}^{\left(  2\right)
}\right)  ,v=1,..,N,$%
\begin{equation}
\sum_{\omega,\omega^{\prime}\in\Omega_{N}}\Pi\left(  \omega|X^{\left(
1\right)  }\right)  \Pi\left(  \omega^{\prime}|X^{\left(  2\right)  }\right)
\mathbf{1}_{\left(  \mathcal{A}_{N}^{\left(  1\right)  ,\varrho}\left(
\varepsilon\right)  \right)  ^{c}}\left(  \omega,\omega^{\prime}\right)
\leq2Ne^{-2N\varepsilon^{2}}\ .
\end{equation}
Moreover, since conditionally on $\omega\in\Omega_{N}$ the sums of r.v.'s
$\sum_{u=1}^{N}\omega_{v,u}^{\prime}\varrho\left(  \mathcal{\breve{T}}%
_{v}\left(  X^{\left(  2\right)  },\omega\right)  ,\mathcal{\breve{T}}%
_{u}\left(  X^{\left(  2\right)  },\omega\right)  \right)  ,v=1,..,N,$ are
sums of independent r.v.'s, by the Hoeffding's inequality we get
\begin{equation}
\sum_{\omega,\omega^{\prime}\in\Omega_{N}}\Pi\left(  \omega|X^{\left(
1\right)  }\right)  \Pi\left(  \omega^{\prime}|X^{\left(  2\right)  }\right)
\mathbf{1}_{\left(  \mathcal{A}_{N}^{\left(  2\right)  ,\varrho}\left(
\varepsilon\right)  \right)  ^{c}}\left(  \omega,\omega^{\prime}\right)
\leq2Ne^{-2N\varepsilon^{2}}\ .
\end{equation}
Hence, arguing as in (\ref{ld1}),
\begin{align}
&  \frac{1}{N}\sum_{v=1}^{N}\sum_{\omega,\omega^{\prime}\in\Omega_{N}}%
\Pi\left(  \omega|X^{\left(  1\right)  }\right)  \Pi\left(  \omega^{\prime
}|X^{\left(  2\right)  }\right)  \left[  \mathbf{1}_{\left(  \mathcal{A}%
_{N}^{\left(  1\right)  ,\varrho}\left(  \varepsilon\right)  \right)  ^{c}%
}\left(  \omega,\omega^{\prime}\right)  +\mathbf{1}_{\left(  \mathcal{A}%
_{N}^{\left(  2\right)  }\left(  \varepsilon\right)  \right)  ^{c}}\left(
\omega,\omega^{\prime}\right)  \right]  \times\\
&  \times\left\vert \varphi\left(  \mathcal{T}_{v}\left(  X^{\left(  2\right)
},\omega\right)  ,\mathcal{T}_{v}\left(  \mathcal{T}\left(  X^{\left(
2\right)  },\omega\right)  ,\omega^{\prime}\right)  \right)  -\left(
\varphi\left(  \mathcal{\breve{T}}_{v}\left(  X^{\left(  2\right)  }%
,\omega\right)  ,\mathcal{\breve{T}}_{v}\left(  \mathcal{\breve{T}}\left(
X^{\left(  2\right)  },\omega\right)  ,\omega^{\prime}\right)  \right)
\right)  \right\vert \nonumber\\
&  \leq8\left\Vert \varphi\right\Vert _{Lip}Ne^{-2N\varepsilon^{2}%
}\ .\nonumber
\end{align}
We are then left with the estimate of
\begin{gather}
D_{N}^{\varrho,\varepsilon}\left(  \varphi\right)  :=\frac{1}{N}\sum_{v=1}%
^{N}\sum_{\omega,\omega^{\prime}\in\Omega_{N}}\Pi\left(  \omega|X^{\left(
1\right)  }\right)  \Pi\left(  \omega^{\prime}|X^{\left(  2\right)  }\right)
\mathbf{1}_{\mathcal{A}_{N}^{\left(  1\right)  ,\varrho}\left(  \varepsilon
\right)  \cap\mathcal{A}_{N}^{\left(  2\right)  ,\varrho}\left(
\varepsilon\right)  }\left(  \omega,\omega^{\prime}\right)  \times
\label{Dre}\\
\times\left\vert \varphi\left(  \mathcal{T}_{v}\left(  X^{\left(  2\right)
},\omega\right)  ,\mathcal{T}_{v}\left(  \mathcal{T}\left(  X^{\left(
2\right)  },\omega\right)  ,\omega^{\prime}\right)  \right)  -\varphi\left(
\mathcal{\breve{T}}_{v}\left(  X^{\left(  2\right)  },\omega\right)
,\mathcal{\breve{T}}_{v}\left(  \mathcal{\breve{T}}\left(  X^{\left(
2\right)  },\omega\right)  ,\omega^{\prime}\right)  \right)  \right\vert
\ .\nonumber
\end{gather}
But
\begin{align}
\left\vert \mathcal{T}_{v}\left(  \mathcal{T}\left(  X^{\left(  2\right)
},\omega\right)  ,\omega^{\prime}\right)  -\mathcal{\breve{T}}_{v}\left(
\mathcal{\breve{T}}\left(  X^{\left(  2\right)  },\omega\right)
,\omega^{\prime}\right)  \right\vert  &  \leq\left\vert \mathcal{T}_{v}\left(
\mathcal{T}\left(  X^{\left(  2\right)  },\omega\right)  ,\omega^{\prime
}\right)  -\mathcal{\breve{T}}_{v}\left(  \mathcal{T}\left(  X^{\left(
2\right)  },\omega\right)  ,\omega^{\prime}\right)  \right\vert \label{b0}\\
&  +\left\vert \mathcal{\breve{T}}_{v}\left(  \mathcal{T}\left(  X^{\left(
2\right)  },\omega\right)  ,\omega^{\prime}\right)  -\mathcal{\breve{T}}%
_{v}\left(  \mathcal{\breve{T}}\left(  X^{\left(  2\right)  },\omega\right)
,\omega^{\prime}\right)  \right\vert \ ,\nonumber
\end{align}
where, since
\begin{gather}
\left\vert \varrho\left(  \mathcal{T}_{v}\left(  X^{\left(  2\right)  }%
,\omega\right)  ,\mathcal{T}_{u}\left(  X^{\left(  2\right)  },\omega\right)
\right)  -\varrho\left(  \mathcal{\breve{T}}_{v}\left(  X^{\left(  2\right)
},\omega\right)  ,\mathcal{\breve{T}}_{u}\left(  X^{\left(  2\right)  }%
,\omega\right)  \right)  \right\vert \leq\\
\left\Vert \varrho\right\Vert _{Lip}\left(  \left\vert \mathcal{T}_{v}\left(
X^{\left(  2\right)  },\omega\right)  -\mathcal{\breve{T}}_{v}\left(
X^{\left(  2\right)  },\omega\right)  \right\vert +\left\vert \mathcal{T}%
_{u}\left(  X^{\left(  2\right)  },\omega\right)  -\mathcal{\breve{T}}%
_{u}\left(  X^{\left(  2\right)  },\omega\right)  \right\vert \right)
\ ,\nonumber
\end{gather}
we have%
\begin{gather}
\left\vert \mathcal{T}_{v}\left(  \mathcal{T}\left(  X^{\left(  2\right)
},\omega\right)  ,\omega^{\prime}\right)  -\mathcal{\breve{T}}_{v}\left(
\mathcal{T}\left(  X^{\left(  2\right)  },\omega\right)  ,\omega^{\prime
}\right)  \right\vert =\label{b1}\\
\left\vert \frac{1}{N}\sum_{u=1}^{N}\frac{\omega_{v,u}^{\prime}\varrho\left(
\mathcal{T}_{v}\left(  X^{\left(  2\right)  },\omega\right)  ,\mathcal{T}%
_{u}\left(  X^{\left(  2\right)  },\omega\right)  \right)  \mathcal{T}%
_{u}\left(  X^{\left(  2\right)  },\omega\right)  }{\frac{1}{N}\sum_{u=1}%
^{N}\omega_{v,u}^{\prime}\varrho\left(  \mathcal{T}_{v}\left(  X^{\left(
2\right)  },\omega\right)  ,\mathcal{T}_{u}\left(  X^{\left(  2\right)
},\omega\right)  \right)  }-\right. \nonumber\\
\left.  \frac{1}{N}\sum_{u=1}^{N}\frac{\omega_{v,u}^{\prime}\varrho\left(
\mathcal{T}_{v}\left(  X^{\left(  2\right)  },\omega\right)  ,\mathcal{T}%
_{u}\left(  X^{\left(  2\right)  },\omega\right)  \right)  \mathcal{T}%
_{u}\left(  X^{\left(  2\right)  },\omega\right)  }{\frac{1}{N}\sum_{u=1}%
^{N}p\left(  \left\vert X_{v}^{\left(  2\right)  }-X_{u}^{\left(  2\right)
}\right\vert \right)  \varrho\left(  \mathcal{T}_{v}\left(  X^{\left(
2\right)  },\omega\right)  ,\mathcal{T}_{u}\left(  X^{\left(  2\right)
},\omega\right)  \right)  }\right\vert \leq\nonumber\\
\frac{1}{N}\sum_{u=1}^{N}\frac{\omega_{v,u}^{\prime}\varrho\left(
\mathcal{T}_{v}\left(  X^{\left(  2\right)  },\omega\right)  ,\mathcal{T}%
_{u}\left(  X^{\left(  2\right)  },\omega\right)  \right)  \mathcal{T}%
_{u}\left(  X^{\left(  2\right)  },\omega\right)  }{\frac{1}{N}\sum_{u=1}%
^{N}\omega_{v,u}^{\prime}\varrho\left(  \mathcal{T}_{v}\left(  X^{\left(
2\right)  },\omega\right)  ,\mathcal{T}_{u}\left(  X^{\left(  2\right)
},\omega\right)  \right)  }\times\nonumber\\
\times\frac{\left\vert \frac{1}{N}\sum_{u=1}^{N}\omega_{v,u}^{\prime}%
\varrho\left(  \mathcal{T}_{v}\left(  X^{\left(  2\right)  },\omega\right)
,\mathcal{T}_{u}\left(  X^{\left(  2\right)  },\omega\right)  \right)
-\frac{1}{N}\sum_{u=1}^{N}p\left(  \left\vert X_{v}^{\left(  2\right)  }%
-X_{u}^{\left(  2\right)  }\right\vert \right)  \varrho\left(  \mathcal{T}%
_{v}\left(  X^{\left(  2\right)  },\omega\right)  ,\mathcal{T}_{u}\left(
X^{\left(  2\right)  },\omega\right)  \right)  \right\vert }{\frac{1}{N}%
\sum_{u=1}^{N}p\left(  \left\vert X_{v}^{\left(  2\right)  }-X_{u}^{\left(
2\right)  }\right\vert \right)  \varrho\left(  \mathcal{T}_{v}\left(
X^{\left(  2\right)  },\omega\right)  ,\mathcal{T}_{u}\left(  X^{\left(
2\right)  },\omega\right)  \right)  }\leq\nonumber\\
\left\{  \left\vert \frac{1}{N}\sum_{u=1}^{N}\omega_{v,u}^{\prime}%
\varrho\left(  \mathcal{\breve{T}}_{v}\left(  X^{\left(  2\right)  }%
,\omega\right)  ,\mathcal{\breve{T}}_{u}\left(  X^{\left(  2\right)  }%
,\omega\right)  \right)  -\frac{1}{N}\sum_{u=1}^{N}p\left(  \left\vert
X_{v}^{\left(  2\right)  }-X_{u}^{\left(  2\right)  }\right\vert \right)
\varrho\left(  \mathcal{\breve{T}}_{v}\left(  X^{\left(  2\right)  }%
,\omega\right)  ,\mathcal{\breve{T}}_{u}\left(  X^{\left(  2\right)  }%
,\omega\right)  \right)  \right\vert +\right. \nonumber\\
\left.  +2\left\Vert \varrho\right\Vert _{Lip}\left(  \left\vert
\mathcal{T}_{v}\left(  X^{\left(  2\right)  },\omega\right)  -\mathcal{\breve
{T}}_{v}\left(  X^{\left(  2\right)  },\omega\right)  \right\vert +\left\vert
\mathcal{T}_{u}\left(  X^{\left(  2\right)  },\omega\right)  -\mathcal{\breve
{T}}_{u}\left(  X^{\left(  2\right)  },\omega\right)  \right\vert \right)
\right\}  \times\nonumber\\
\times\left\{  \frac{1}{N}\sum_{u=1}^{N}p\left(  \left\vert X_{v}^{\left(
2\right)  }-X_{u}^{\left(  2\right)  }\right\vert \right)  \left[
\varrho\left(  \mathcal{\breve{T}}_{v}\left(  X^{\left(  2\right)  }%
,\omega\right)  ,\mathcal{\breve{T}}_{u}\left(  X^{\left(  2\right)  }%
,\omega\right)  \right)  -\right.  \right. \nonumber\\
\left.  \left.  \left\Vert \varrho\right\Vert _{Lip}\left(  \left\vert
\mathcal{T}_{v}\left(  X^{\left(  2\right)  },\omega\right)  -\mathcal{\breve
{T}}_{v}\left(  X^{\left(  2\right)  },\omega\right)  \right\vert +\left\vert
\mathcal{T}_{u}\left(  X^{\left(  2\right)  },\omega\right)  -\mathcal{\breve
{T}}_{u}\left(  X^{\left(  2\right)  },\omega\right)  \right\vert \right)
\right]  \right\}  ^{-1}\nonumber
\end{gather}
and
\begin{gather}
\left\vert \mathcal{\breve{T}}_{v}\left(  \mathcal{T}\left(  X^{\left(
2\right)  },\omega\right)  ,\omega^{\prime}\right)  -\mathcal{\breve{T}}%
_{v}\left(  \mathcal{\breve{T}}\left(  X^{\left(  2\right)  },\omega\right)
,\omega^{\prime}\right)  \right\vert =\label{b2}\\
\left\vert \frac{1}{N}\sum_{u=1}^{N}\frac{\omega_{v,u}^{\prime}\varrho\left(
\mathcal{T}_{v}\left(  X^{\left(  2\right)  },\omega\right)  ,\mathcal{T}%
_{u}\left(  X^{\left(  2\right)  },\omega\right)  \right)  \mathcal{T}%
_{u}\left(  X^{\left(  2\right)  },\omega\right)  }{\frac{1}{N}\sum_{u=1}%
^{N}p\left(  \left\vert X_{v}^{\left(  2\right)  }-X_{u}^{\left(  2\right)
}\right\vert \right)  \varrho\left(  \mathcal{T}_{v}\left(  X^{\left(
2\right)  },\omega\right)  ,\mathcal{T}_{u}\left(  X^{\left(  2\right)
},\omega\right)  \right)  }-\right. \nonumber\\
\left.  \frac{1}{N}\sum_{u=1}^{N}\frac{\omega_{v,u}^{\prime}\varrho\left(
\mathcal{\breve{T}}_{v}\left(  X^{\left(  2\right)  },\omega\right)
,\mathcal{\breve{T}}_{u}\left(  X^{\left(  2\right)  },\omega\right)  \right)
\mathcal{\breve{T}}_{u}\left(  X^{\left(  2\right)  },\omega\right)  }%
{\frac{1}{N}\sum_{u=1}^{N}p\left(  \left\vert X_{v}^{\left(  2\right)  }%
-X_{u}^{\left(  2\right)  }\right\vert \right)  \varrho\left(  \mathcal{\breve
{T}}_{v}\left(  X^{\left(  2\right)  },\omega\right)  ,\mathcal{\breve{T}}%
_{u}\left(  X^{\left(  2\right)  },\omega\right)  \right)  }\right\vert
\leq\nonumber\\
\frac{1}{N}\sum_{u=1}^{N}\omega_{v,u}^{\prime}\left\{  \varrho\left(
\mathcal{T}_{v}\left(  X^{\left(  2\right)  },\omega\right)  ,\mathcal{T}%
_{u}\left(  X^{\left(  2\right)  },\omega\right)  \right)  \mathcal{T}%
_{u}\left(  X^{\left(  2\right)  },\omega\right)  \times\right. \nonumber\\
\times\frac{\frac{1}{N}\sum_{u=1}^{N}\omega_{v,u}^{\prime}p\left(  \left\vert
X_{v}^{\left(  2\right)  }-X_{u}^{\left(  2\right)  }\right\vert \right)
\left\vert \varrho\left(  \mathcal{\breve{T}}_{v}\left(  X^{\left(  2\right)
},\omega\right)  ,\mathcal{\breve{T}}_{u}\left(  X^{\left(  2\right)  }%
,\omega\right)  \right)  -\varrho\left(  \mathcal{T}_{v}\left(  X^{\left(
2\right)  },\omega\right)  ,\mathcal{T}_{u}\left(  X^{\left(  2\right)
},\omega\right)  \right)  \right\vert }{\frac{1}{N}\sum_{u=1}^{N}p\left(
\left\vert X_{v}^{\left(  2\right)  }-X_{u}^{\left(  2\right)  }\right\vert
\right)  \varrho\left(  \mathcal{T}_{v}\left(  X^{\left(  2\right)  }%
,\omega\right)  ,\mathcal{T}_{u}\left(  X^{\left(  2\right)  },\omega\right)
\right)  \frac{1}{N}\sum_{u=1}^{N}p\left(  \left\vert X_{v}^{\left(  2\right)
}-X_{u}^{\left(  2\right)  }\right\vert \right)  \varrho\left(
\mathcal{\breve{T}}_{v}\left(  X^{\left(  2\right)  },\omega\right)
,\mathcal{\breve{T}}_{u}\left(  X^{\left(  2\right)  },\omega\right)  \right)
}+\nonumber\\
+\frac{1}{\frac{1}{N}\sum_{u=1}^{N}p\left(  \left\vert X_{v}^{\left(
2\right)  }-X_{u}^{\left(  2\right)  }\right\vert \right)  \varrho\left(
\mathcal{\breve{T}}_{v}\left(  X^{\left(  2\right)  },\omega\right)
,\mathcal{\breve{T}}_{u}\left(  X^{\left(  2\right)  },\omega\right)  \right)
}\times\nonumber\\
\times\left[  \left\vert \mathcal{T}_{u}\left(  X^{\left(  2\right)  }%
,\omega\right)  -\mathcal{\breve{T}}_{u}\left(  X^{\left(  2\right)  }%
,\omega\right)  \right\vert +\mathcal{T}_{u}\left(  X^{\left(  2\right)
},\omega\right)  \times\right. \nonumber\\
\left.  \left.  \times\left\vert \varrho\left(  \mathcal{\breve{T}}_{v}\left(
X^{\left(  2\right)  },\omega\right)  ,\mathcal{\breve{T}}_{u}\left(
X^{\left(  2\right)  },\omega\right)  \right)  -\varrho\left(  \mathcal{T}%
_{v}\left(  X^{\left(  2\right)  },\omega\right)  ,\mathcal{T}_{u}\left(
X^{\left(  2\right)  },\omega\right)  \right)  \right\vert \right]  \right\}
\leq\nonumber\\
\left\Vert \varrho\right\Vert _{Lip}\frac{\frac{1}{N}\sum_{u=1}^{N}%
\omega_{v,u}^{\prime}\varrho\left(  \mathcal{T}_{v}\left(  X^{\left(
2\right)  },\omega\right)  ,\mathcal{T}_{u}\left(  X^{\left(  2\right)
},\omega\right)  \right)  \mathcal{T}_{u}\left(  X^{\left(  2\right)  }%
,\omega\right)  }{\frac{1}{N}\sum_{u=1}^{N}p\left(  \left\vert X_{v}^{\left(
2\right)  }-X_{u}^{\left(  2\right)  }\right\vert \right)  \varrho\left(
\mathcal{T}_{v}\left(  X^{\left(  2\right)  },\omega\right)  ,\mathcal{T}%
_{u}\left(  X^{\left(  2\right)  },\omega\right)  \right)  }\times\nonumber\\
\times\left(  \left\vert \mathcal{T}_{v}\left(  X^{\left(  2\right)  }%
,\omega\right)  -\mathcal{\breve{T}}_{v}\left(  X^{\left(  2\right)  }%
,\omega\right)  \right\vert +\left\vert \mathcal{T}_{u}\left(  X^{\left(
2\right)  },\omega\right)  -\mathcal{\breve{T}}_{u}\left(  X^{\left(
2\right)  },\omega\right)  \right\vert \right)  +\nonumber\\
\frac{\frac{1}{N}\sum_{u=1}^{N}\omega_{v,u}^{\prime}}{\frac{1}{N}\sum
_{u=1}^{N}p\left(  \left\vert X_{v}^{\left(  2\right)  }-X_{u}^{\left(
2\right)  }\right\vert \right)  \varrho\left(  \mathcal{\breve{T}}_{v}\left(
X^{\left(  2\right)  },\omega\right)  ,\mathcal{\breve{T}}_{u}\left(
X^{\left(  2\right)  },\omega\right)  \right)  }\left[  \left\vert
\mathcal{T}_{u}\left(  X^{\left(  2\right)  },\omega\right)  -\mathcal{\breve
{T}}_{u}\left(  X^{\left(  2\right)  },\omega\right)  \right\vert +\right.
\nonumber\\
\left.  +\left\Vert \varrho\right\Vert _{Lip}\left(  \left\vert \mathcal{T}%
_{v}\left(  X^{\left(  2\right)  },\omega\right)  -\mathcal{\breve{T}}%
_{v}\left(  X^{\left(  2\right)  },\omega\right)  \right\vert +\left\vert
\mathcal{T}_{u}\left(  X^{\left(  2\right)  },\omega\right)  -\mathcal{\breve
{T}}_{u}\left(  X^{\left(  2\right)  },\omega\right)  \right\vert \right)
\right]  \leq\nonumber\\
\left\Vert \varrho\right\Vert _{Lip}\left\{  \frac{1}{N}\sum_{u=1}^{N}p\left(
\left\vert X_{v}^{\left(  2\right)  }-X_{u}^{\left(  2\right)  }\right\vert
\right)  \times\right. \nonumber\\
\left.  \times\left[  \varrho\left(  \mathcal{\breve{T}}_{v}\left(  X^{\left(
2\right)  },\omega\right)  ,\mathcal{\breve{T}}_{u}\left(  X^{\left(
2\right)  },\omega\right)  \right)  -\left\Vert \varrho\right\Vert
_{Lip}\left(  \left\vert \mathcal{T}_{v}\left(  X^{\left(  2\right)  }%
,\omega\right)  -\mathcal{\breve{T}}_{v}\left(  X^{\left(  2\right)  }%
,\omega\right)  \right\vert +\left\vert \mathcal{T}_{u}\left(  X^{\left(
2\right)  },\omega\right)  -\mathcal{\breve{T}}_{u}\left(  X^{\left(
2\right)  },\omega\right)  \right\vert \right)  \right]  \right\}  ^{-1}%
\times\nonumber\\
\times\left(  \left\vert \mathcal{T}_{v}\left(  X^{\left(  2\right)  }%
,\omega\right)  -\mathcal{\breve{T}}_{v}\left(  X^{\left(  2\right)  }%
,\omega\right)  \right\vert +\frac{1}{N}\sum_{u=1}^{N}\left\vert
\mathcal{T}_{u}\left(  X^{\left(  2\right)  },\omega\right)  -\mathcal{\breve
{T}}_{u}\left(  X^{\left(  2\right)  },\omega\right)  \right\vert \right)
+\nonumber\\
\frac{\left\vert \mathcal{T}_{v}\left(  X^{\left(  2\right)  },\omega\right)
-\mathcal{\breve{T}}_{v}\left(  X^{\left(  2\right)  },\omega\right)
\right\vert +\left(  1+\left\Vert \varrho\right\Vert _{Lip}\right)  \frac
{1}{N}\sum_{u=1}^{N}\left\vert \mathcal{T}_{u}\left(  X^{\left(  2\right)
},\omega\right)  -\mathcal{\breve{T}}_{u}\left(  X^{\left(  2\right)  }%
,\omega\right)  \right\vert }{\frac{1}{N}\sum_{u=1}^{N}p\left(  \left\vert
X_{v}^{\left(  2\right)  }-X_{u}^{\left(  2\right)  }\right\vert \right)
\varrho\left(  \mathcal{\breve{T}}_{v}\left(  X^{\left(  2\right)  }%
,\omega\right)  ,\mathcal{\breve{T}}_{u}\left(  X^{\left(  2\right)  }%
,\omega\right)  \right)  }\ .\nonumber
\end{gather}
Since, for any $v\in\left\{  1,..N\right\}  ,$%
\begin{gather}
\left\vert \mathcal{T}_{v}\left(  X^{\left(  2\right)  },\omega\right)
-\mathcal{\breve{T}}_{v}\left(  X^{\left(  2\right)  },\omega\right)
\right\vert \leq\frac{1}{N}\sum_{u=1}^{N}\frac{\omega_{v,u}\varrho\left(
X_{v}^{\left(  2\right)  },X_{u}^{\left(  2\right)  }\right)  }{\frac{1}%
{N}\sum_{u=1}^{N}\omega_{v,u}\varrho\left(  X_{v}^{\left(  2\right)  }%
,X_{u}^{\left(  2\right)  }\right)  }X_{u}^{\left(  2\right)  }\times\\
\times\frac{\left\vert \frac{1}{N}\sum_{u=1}^{N}\omega_{v,u}\varrho\left(
X_{v}^{\left(  2\right)  },X_{u}^{\left(  2\right)  }\right)  -\mu
_{N}^{X^{\left(  1\right)  },X^{\left(  2\right)  }}\left[  p\left(
\left\vert X_{v}^{\left(  1\right)  }-\cdot\right\vert \right)  \varrho\left(
X_{v}^{\left(  2\right)  },\cdot\right)  \right]  \right\vert }{\mu
_{N}^{X^{\left(  1\right)  },X^{\left(  2\right)  }}\left[  p\left(
\left\vert X_{v}^{\left(  1\right)  }-\cdot\right\vert \right)  \varrho\left(
X_{v}^{\left(  2\right)  },\cdot\right)  \right]  }\ ,\nonumber
\end{gather}
which on the event $\mathcal{A}_{N}^{\left(  1\right)  ,\varrho}\left(
\varepsilon\right)  $ can be bounded by
\begin{equation}
\left\vert \mathcal{T}_{v}\left(  X^{\left(  2\right)  },\omega\right)
-\mathcal{\breve{T}}_{v}\left(  X^{\left(  2\right)  },\omega\right)
\right\vert \leq\frac{\varepsilon}{\mu_{N}^{X^{\left(  1\right)  },X^{\left(
2\right)  }}\left[  p\left(  \left\vert X_{v}^{\left(  1\right)  }%
-\cdot\right\vert \right)  \varrho\left(  X_{v}^{\left(  2\right)  }%
,\cdot\right)  \right]  }\ ,
\end{equation}
we get that, on this event%
\begin{gather}
\left\vert \varrho\left(  \mathcal{T}_{v}\left(  X^{\left(  2\right)  }%
,\omega\right)  ,\mathcal{T}_{u}\left(  X^{\left(  2\right)  },\omega\right)
\right)  -\varrho\left(  \mathcal{\breve{T}}_{v}\left(  X^{\left(  2\right)
},\omega\right)  ,\mathcal{\breve{T}}_{u}\left(  X^{\left(  2\right)  }%
,\omega\right)  \right)  \right\vert \leq\\
\left\Vert \varrho\right\Vert _{Lip}\left(  \frac{1}{\mu_{N}^{X^{\left(
1\right)  },X^{\left(  2\right)  }}\left[  p\left(  \left\vert X_{v}^{\left(
1\right)  }-\cdot\right\vert \right)  \varrho\left(  X_{v}^{\left(  2\right)
},\cdot\right)  \right]  }+\frac{1}{\mu_{N}^{X^{\left(  1\right)  },X^{\left(
2\right)  }}\left[  p\left(  \left\vert X_{u}^{\left(  1\right)  }%
-\cdot\right\vert \right)  \varrho\left(  X_{u}^{\left(  2\right)  }%
,\cdot\right)  \right]  }\right)  \varepsilon\nonumber
\end{gather}
and
\begin{equation}
\left\vert \varrho\left(  \mathcal{\breve{T}}_{v}\left(  X^{\left(  2\right)
},\omega\right)  ,\mathcal{\breve{T}}_{u}\left(  X^{\left(  2\right)  }%
,\omega\right)  \right)  -\varrho\left(  \theta_{\mu_{N}^{X^{\left(  1\right)
},X^{\left(  2\right)  }}}^{\varrho},\theta_{\mu_{N}^{X^{\left(  1\right)
},X^{\left(  2\right)  }}}^{\varrho}\circ\theta_{\mu_{N}^{X^{\left(  1\right)
},X^{\left(  2\right)  }}}^{\varrho}\right)  \right\vert \leq\left\Vert
\varrho\right\Vert _{Lip}\varepsilon\ .
\end{equation}
Therefore, for any $v=1,..,N,$ on the event $\mathcal{A}_{N}^{\left(
1\right)  ,\varrho}\left(  \varepsilon\right)  \cap\mathcal{A}_{N}^{\left(
2\right)  ,\varrho}\left(  \varepsilon\right)  ,$ there exist two bounded
positive functions $\phi_{v}^{\prime}\left(  \cdot;\varrho,p\right)  $ and
$\phi_{v}^{\prime\prime}\left(  \cdot;\varrho,p\right)  $ on $\mathfrak{P}%
\left(  \left[  0,1\right]  ^{2}\right)  $ such that, by (\ref{b1})
\begin{equation}
\left\vert \mathcal{T}_{v}\left(  \mathcal{T}\left(  X^{\left(  2\right)
},\omega\right)  ,\omega^{\prime}\right)  -\mathcal{\breve{T}}_{v}\left(
\mathcal{T}\left(  X^{\left(  2\right)  },\omega\right)  ,\omega^{\prime
}\right)  \right\vert \leq\phi_{v}^{\prime}\left(  \mu_{N}^{X^{\left(
1\right)  },X^{\left(  2\right)  }};\varrho,p\right)  \varepsilon
\end{equation}
and (\ref{b2}),
\begin{equation}
\left\vert \mathcal{\breve{T}}_{v}\left(  \mathcal{T}\left(  X^{\left(
2\right)  },\omega\right)  ,\omega^{\prime}\right)  -\mathcal{\breve{T}}%
_{v}\left(  \mathcal{\breve{T}}\left(  X^{\left(  2\right)  },\omega\right)
,\omega^{\prime}\right)  \right\vert \leq\phi_{v}^{\prime\prime}\left(
\mu_{N}^{X^{\left(  1\right)  },X^{\left(  2\right)  }};\varrho,p\right)
\varepsilon
\end{equation}
so that, by (\ref{Dre}) and (\ref{b0}), for any $\varepsilon>0,$ there exists
a positive constant $C^{\prime}:=C^{\prime}\left(  \varrho,p,\mu\right)  $
such that $D_{N}^{\varrho,\varepsilon}\left(  \varphi\right)  \leq C^{\prime
}\left\Vert \varphi\right\Vert _{Lip}\varepsilon$ implying that $\lim
_{N\rightarrow\infty}\mathfrak{d}\left(  \mathbf{T}_{\#}\mu_{N}^{X^{\left(
1\right)  },X^{\left(  2\right)  }},\mathbf{\breve{T}}_{\#}\mu_{N}^{X^{\left(
1\right)  },X^{\left(  2\right)  }}\right)  =0.$ Thus, since
\begin{equation}
\mathfrak{d}\left(  \mathbf{T}_{\#}\mu_{N}^{X^{\left(  1\right)  },X^{\left(
2\right)  }},\Theta_{\#}^{\varrho}\mu\right)  \leq\mathfrak{d}\left(
\mathbf{T}_{\#}\mu_{N}^{X^{\left(  1\right)  },X^{\left(  2\right)  }%
},\mathbf{\breve{T}}_{\#}\mu_{N}^{X^{\left(  1\right)  },X^{\left(  2\right)
}}\right)  +\mathfrak{d}\left(  \mathbf{\breve{T}}_{\#}\mu_{N}^{X^{\left(
1\right)  },X^{\left(  2\right)  }},\Theta_{\#}\mu\right)
\end{equation}
and since, by the previous Lemma, $\lim_{N\rightarrow\infty}\mathfrak{d}%
\left(  \mathbf{\breve{T}}_{\#}\mu_{N}^{X^{\left(  1\right)  },X^{\left(
2\right)  }},\Theta_{\#}\mu\right)  =0,$ the thesis follows.
\end{proof}

\subsubsection{Asymptotic limit as $t\rightarrow\infty$ of the
monokinetic-type evolution}

Let us denote by $\left\{  \mathbf{Z}_{t}^{\varrho}\right\}  _{t\geq0}$ with,
for any $t\geq0,\mathbf{Z}_{t}^{\varrho}:=\left(  Z_{t}^{\varrho,\left(
1\right)  },Z_{t}^{\varrho,\left(  2\right)  }\right)  ,$ the non-linear
Markov chain on $\left(  \left[  0,1\right]  ^{2},\mathcal{B}\left(  \left[
0,1\right]  ^{2}\right)  \right)  $ with degenerate kernel such that, for any
bounded measurable $\varphi:\left[  0,1\right]  ^{2}\longrightarrow
\mathbb{R},$ if $\mu_{t}$ is the law of $\mathbf{Z}_{t}^{\varrho},$%
\begin{equation}
\mathbb{E}\left[  \varphi\left(  \mathbf{Z}_{t+1}^{\varrho}\right)
|\mathbf{Z}_{t}^{\varrho}\right]  :=\varphi\left(  \theta_{\mu_{t}}^{\varrho
}\left(  Z_{t}^{\varrho,\left(  1\right)  },Z_{t}^{\varrho,\left(  2\right)
}\right)  ,\theta_{\mu_{t}}^{\varrho}\circ\theta_{\mu_{t}}^{\varrho
}\left(  Z_{t}^{\varrho,\left(  1\right)  },Z_{t}^{\varrho,\left(  2\right)
}\right)  \right)  \ ,
\end{equation}
so that, if $\varrho=\mathbf{1}_{\left[  0,1\right]  },\left\{  \mathbf{Z}%
_{t}^{\varrho}\right\}  _{t\geq0}=\left\{  \mathbf{Z}_{t}\right\}  _{t\geq0}.$

Clearly, by definition, for any $\mu\in\mathfrak{P}\left(  \left[  0,1\right]
^{2}\right)  ,\Theta_{\mu}^{\varrho}$ leaves the constant functions invariant,
in other words $\Theta_{\#}^{\varrho}\mathfrak{P}\left(  \left[  0,1\right]
^{2}\right)  \subseteq\mathfrak{P}\left(  \left[  0,1\right]  ^{2}\right)  .$
Moreover, from (\ref{Theta}), it follows that any degenerate probability
distribution supported on $\left(  x,x\right)  ,x\in\left[  0,1\right]  $ is
stationary for $\left\{  \mathbf{Z}_{t}^{\varrho}\right\}  _{t\geq0}.$

In order to apply the results presented for the finite size system when the
interaction among the agents is defined through a non constant strictly positive function
$\varrho:\left[  0,1\right]  ^{2}\longrightarrow\left[  0,1\right]  $ as
described by (\ref{r=rho}), the definitions (\ref{Pomega})
and (\ref{gamma}) must be modified in such a way that the matrix elements of
$\Xi_{N}\times\Omega_{N}\ni\left(  X,\omega\right)  \longmapsto P\left(
X,\omega\right)  \in BL\left(  \mathbb{R}^{N}\right)  ,$ namely, $\forall
u,v\in\mathbf{V}_{N},P_{v,u}\left(  X,\omega\right)  :=\frac{\varrho\left(
X_{v},X_{u}\right)  \omega_{v,u}}{\sum_{u\in\mathbf{V}_{N}}\varrho\left(
X_{v},X_{u}\right)  \omega_{v,u}},$ replace (\ref{Pomega}), and (\ref{gamma})
is replaced by
\begin{equation}
\Omega\ni\omega\longmapsto\gamma\left(  \omega\right)  :=\min_{u,v\in
\mathbf{V}\ :\ u\neq v}\inf_{X\in\Xi_{N}}\sum_{w,z\in\mathbf{V}}P_{u,w}\left(
X,\omega\right)  P_{v,z}\left(  X,\omega\right)  \wedge P_{u,z}\left(
X,\omega\right)  P_{v,w}\left(  X,\omega\right)  \in\lbrack0,1)\ .
\end{equation}

\begin{proposition}
\label{mfc}For any initial datum $\mu\in\mathfrak{P}\left(  \left[
0,1\right]  ^{2}\right)  ,$ there exists $x\in\left[  0,1\right]  $ such that
$\left\{  \mathbf{Z}_{t}^{\varrho}\right\}  _{t\geq0}$ converges to $\left(
x,x\right)  $ in probability geometrically fast as $t$ tends to infinity.
\end{proposition}

\begin{proof}
The result is obvious if $\mu$ is a Dirac mass at $\left(x,x\right)$ for some $x \in \left[0,1\right].$ 
Hence, given $\mu\in\mathfrak{P}\left(  \left[  0,1\right]  ^{2}\right)  $  with positive variance, if
$\mathbf{Z}_{0}^{\varrho}$ if the random vector with law $\mu,$ for any
$t\geq1,$ let us set $\left(  \Theta_{\#}^{\varrho}\right)  ^{t}\mu
:=\Theta_{\#}^{\varrho}\mu_{t-1}.$ Moreover, for any $N\geq1,$ given an
initial datum $\left(  X^{\left(  1\right)  },X^{\left(  2\right)  }\right)
\in\left(  \left[  0,1\right]  ^{2}\right)  ^{N},$ for any $t\geq1,$ we set
$\left(  \mathbf{T}_{\#}\right)  ^{t}\mu_{N}^{X^{\left(  1\right)
},X^{\left(  2\right)  }}:=\mathbf{T}_{\#}\mu_{N}^{X^{\left(  1\right)
}\left(  t-1\right)  ,X^{\left(  2\right)  }\left(  t-1\right)  }$ where the
random vector $\left(  X^{\left(  1\right)  }\left(  t-1\right)  ,X^{\left(
2\right)  }\left(  t-1\right)  \right)  \in\left(  \left[  0,1\right]
^{2}\right)  ^{N}$ is defined in (\ref{x_t}) and in Remark \ref{Rem1}. Thus,
if $\left\{  \mu_{N}^{X^{\left(  1\right)  },X^{\left(  2\right)  }}\right\}
_{N\geq1}\subset\mathfrak{P}\left(  \left[  0,1\right]  ^{2}\right)  $ weakly
converges to $\mu,$ by Theorem \ref{mf4} $\left\{  \mathbf{T}_{\#}\mu
_{N}^{X^{\left(  1\right)  },X^{\left(  2\right)  }}\right\}  _{N\geq1}%
\subset\mathfrak{P}\left(  \left[  0,1\right]  ^{2}\right)  $ weakly converges
to $\Theta_{\#}^{\varrho}\mu$ which, again by Theorem \ref{mf4}, implies
\begin{equation}
\lim_{N\rightarrow\infty}\mathfrak{d}\left(  \left(  \mathbf{T}_{\#}\right)
^{2}\mu_{N}^{X^{\left(  1\right)  },X^{\left(  2\right)  }},\left(
\Theta_{\#}^{\varrho}\right)  ^{2}\mu\right)  =\lim_{N\rightarrow\infty
}\mathfrak{d}\left(  \mathbf{T}_{\#}\left(  \mathbf{T}_{\#}\mu_{N}^{X^{\left(
1\right)  },X^{\left(  2\right)  }}\right)  ,\Theta_{\#}^{\varrho}\left(
\Theta_{\#}^{\varrho}\mu\right)  \right)  =0
\end{equation}
and therefore by induction that for any $t\geq1,\lim_{N\rightarrow\infty
}\mathfrak{d}\left(  \left(  \mathbf{T}_{\#}\right)  ^{t}\mu_{N}^{X^{\left(
1\right)  },X^{\left(  2\right)  }},\left(  \Theta_{\#}^{\varrho}\right)
^{t}\mu\right)  =0.$

For any $\varepsilon>0,$ there exists $N_{\varepsilon}\geq1$ such that, for
any $N>N_{\varepsilon},\mathfrak{d}\left(  \mu_{N}^{X^{\left(  1\right)
},X^{\left(  2\right)  }},\mu\right)  <\varepsilon$ and the support of
$\mu_{N}^{X^{\left(  1\right)  },X^{\left(  2\right)  }}$ satisfies the
hypothesis of Corollary \ref{gap} so that the associated Markov chain
$\left\{  \mathfrak{y}_{t}^{N}\right\}  _{t\geq0}$ defined in Remark
\ref{Rem1} converges as $t\rightarrow\infty$ to some $\left(  X^{\infty
},X^{\infty}\right)  \in\mathcal{I}_{N}\times\mathcal{I}_{N}.$ This implies
that there exists $x\in\left[  0,1\right]  $ such that $\mu_{N}^{X^{\infty
},X^{\infty}}=\delta_{\left(  x,x\right)  }.$ But
\begin{equation}
\mathfrak{d}\left(  \left(  \Theta_{\#}^{\varrho}\right)  ^{t}\mu
,\delta_{\left(  x,x\right)  }\right)  =\mathfrak{d}\left(  \left(
\Theta_{\#}^{\varrho}\right)  ^{t}\mu,\left(  \mathbf{T}_{\#}\right)  ^{t}%
\mu_{N}^{X^{\left(  1\right)  },X^{\left(  2\right)  }}\right)  +\mathfrak{d}%
\left(  \left(  \mathbf{T}_{\#}\right)  ^{t}\mu_{N}^{X^{\left(  1\right)
},X^{\left(  2\right)  }},\mu_{N}^{X^{\infty},X^{\infty}}\right)  \ .
\end{equation}
Taking first the limit as $N\uparrow\infty$ and then the limit as
$t\uparrow\infty$ we get that $\left\{  \left(  \Theta_{\#}^{\varrho}\right)
^{t}\mu\right\}  _{t\geq0}$ weakly converges to $\delta_{\left(  x,x\right)
},$ i.e. that $\left\{  \mathbf{Z}_{t}^{\varrho}\right\}  _{t\geq0}$ converges
in distribution to the degenerate r.v. $x$ and therefore in probability.

Since in the proof of Theorem \ref{main} we have shown that the support of
$\left(  \mathbf{T}_{\#}\right)  ^{t}\mu_{N}^{X^{\left(  1\right)
},X^{\left(  2\right)  }}=\mu_{N}^{X^{\left(  1\right)  }\left(  t\right)
,X^{\left(  2\right)  }\left(  t\right)  }$ contracts at geometric rate, this
proves that also the support of $\left(  \Theta_{\#}^{\varrho}\right)  ^{t}%
\mu$ contracts at geometric rate because the support of a measure is stable
under weak limits.
\end{proof}

\end{document}